\documentclass[11pt,english]{article}
\usepackage{libertineRoman}
\usepackage{libertineMono}
\usepackage[T1]{fontenc}
\usepackage[latin9]{inputenc}
\usepackage{geometry}
\PassOptionsToPackage{vlined, boxruled}{algorithm2e}
\usepackage{algorithm2e}
\usepackage{amsmath}
\usepackage{amsthm}
\usepackage{amssymb}
\usepackage{graphicx}
\usepackage{xargs}[2008/03/08]
\PassOptionsToPackage{normalem}{ulem}
\usepackage{ulem}

\makeatletter

\providecommand{\tabularnewline}{\\}

\theoremstyle{plain}
\newtheorem{thm}{\protect\theoremname}[section]
\theoremstyle{definition}
\newtheorem{defn}[thm]{\protect\definitionname}
\theoremstyle{plain}
\newtheorem{prop}[thm]{\protect\propositionname}
\theoremstyle{plain}
\newtheorem{cor}[thm]{\protect\corollaryname}
\theoremstyle{plain}
\newtheorem{lem}[thm]{\protect\lemmaname}
\theoremstyle{plain}
\newtheorem{obs}[thm]{\protect\Observationname}
\theoremstyle{remark}
\newtheorem{rem}[thm]{\protect\remarkname}

\date{}
\usepackage{float}
\floatstyle{boxed} 
\restylefloat{figure}

\@ifundefined{showcaptionsetup}{}{%
 \PassOptionsToPackage{caption=false}{subfig}}
\usepackage{subfig}
\makeatother

\usepackage{babel}
\providecommand{\Observationname}{Observation}
\providecommand{\corollaryname}{Corollary}
\providecommand{\definitionname}{Definition}
\providecommand{\lemmaname}{Lemma}
\providecommand{\propositionname}{Proposition}
\providecommand{\remarkname}{Remark}
\providecommand{\theoremname}{Theorem}

\begin{document}
\global\long\def\l{\ell}%

\global\long\def\N{\mathbb{N}}%

\global\long\def\Z{\mathbb{Z}}%

\global\long\def\R{\mathbb{R}}%

\global\long\def\C{\mathbb{C}}%

\global\long\def\E{\mathbb{E}}%

\global\long\def\cs{\mathcal{S}}%

\global\long\def\defi{\coloneqq}%

\global\long\def\s{\subseteq}%

\newcommandx\q[1][usedefault, addprefix=\global, 1=]{#1^{\ast}}%

\newcommandx\p[1][usedefault, addprefix=\global, 1=]{#1^{\prime}}%

\global\long\def\A{\mathcal{A}}%

\global\long\def\Open{\texttt{Open}}%

\global\long\def\Explore{\texttt{Explore}}%

\global\long\def\Color{\texttt{Color}}%

\global\long\def\Bought{\texttt{Bought}}%

\global\long\def\Invest{\texttt{Invest}}%

\global\long\def\True{\texttt{True}}%

\global\long\def\False{\texttt{False}}%

\global\long\def\None{\texttt{None}}%

\global\long\def\Special{\texttt{Special}}%

\global\long\def\UponDeadline{\texttt{UponDeadline}}%

\global\long\def\UponCritical{\texttt{UponCritical}}%

\global\long\def\SetColor{\texttt{SetColor}}%

\global\long\def\t{\texttt{{}}}%

\global\long\def\T{\mathcal{T}}%

\global\long\def\alg{\mathrm{ALG}}%

\global\long\def\opt{\mathrm{OPT}}%

\title{General Framework for Metric Optimization Problems with Delay or with
Deadlines}
\author{%
\begin{tabular}{c}
Yossi Azar\tabularnewline
\texttt{\footnotesize{}azar@tau.ac.il}\tabularnewline
{\small{}Tel Aviv University}\tabularnewline
\end{tabular}\and%
\begin{tabular}{c}
Noam Touitou\tabularnewline
\texttt{\footnotesize{}noamtouitou@mail.tau.ac.il}\tabularnewline
{\small{}Tel Aviv University}\tabularnewline
\end{tabular}}
\maketitle
\begin{abstract}
In this paper, we present a framework used to construct and analyze
algorithms for online optimization problems with deadlines or with
delay over a metric space. Using this framework, we present algorithms
for several different problems. We present an $O(D^{2})$-competitive
deterministic algorithm for online multilevel aggregation with delay
on a tree of depth $D$, an exponential improvement over the $O(D^{4}2^{D})$-competitive
algorithm of Bienkowski et al. (ESA '16). We also present an $O(\log^{2}n)$-competitive
randomized algorithm for online service with delay over any general
metric space of $n$ points, improving upon the $O(\log^{4}n)$-competitive
algorithm by Azar et al. (STOC '17).

In addition, we present the problem of online facility location with
deadlines. In this problem, requests arrive over time in a metric
space, and need to be served until their deadlines by facilities that
are opened momentarily for some cost. We also consider the problem
of facility location with delay, in which the deadlines are replaced
with arbitrary delay functions. For those problems, we present $O(\log^{2}n)$-competitive
algorithms, with $n$ the number of points in the metric space.

The algorithmic framework we present includes techniques for the design
of algorithms as well as techniques for their analysis.
\end{abstract}
\newpage{}

\section{Introduction}

Recently in the field of online algorithms, there has been an increasing
interest in online problems involving deadlines or delay. In such
problems, requests of some form arrive over time, requiring service.
In problems with deadlines, each request is equipped with a deadline,
by which the request must be served. In problems with delay, this
hard constraint is replaced with a more general constraint. In those
problems, each request is equipped with a delay function, such that
an algorithm accumulates delay cost while the request remains pending.
This provides an incentive for the algorithm to serve the request
as soon as possible. Deadlines are a special case of delay, as deadlines
can be approximated arbitrarily well by delay functions.

The mechanism of adding delay or deadlines can be used to convert
a problem over a sequence into a problem over time. For example, a
problem in which an arriving request must immediately be served by
the algorithm can be converted into a problem with deadlines, providing
more flexibility to a possible solution. This conversion often creates
interesting problems over time from problems that are trivial over
a sequence, as well as enables much better solutions (i.e. lower cost).

A case of special interest is the case of such problems over a metric
space. A notable example, which we consider in this paper, is the
\textbf{online multilevel aggregation problem}. In this problem, the
requests arrive on the leaves of a tree. At any time, the algorithm
may choose to transmit any subtree that includes the root of the tree,
at a cost which is the sum of the weights of the subtree's edges.
Pending requests on any leaves contained in the transmitted subtree
are served by the transmission. The general delay case of this problem
was first considered by Bienkowski et al. \cite{DBLP:conf/esa/BienkowskiBBCDF16},
who gave a $O(D^{4}2^{D})$-competitive algorithm for the problem,
with $D$ the depth of the tree. Buchbinder et al. \cite{DBLP:conf/soda/BuchbinderFNT17}
then showed a $O(D)$-competitive deterministic algorithm for the
deadline case. In this paper, we improve the result of \cite{DBLP:conf/esa/BienkowskiBBCDF16}
for general delay exponentially.

Another notable example is the \textbf{online service with delay}
problem, presented in \cite{DBLP:conf/stoc/AzarGGP17}. In this problem,
requests arrive on points in a metric space, accumulating delay while
pending. There is a single server in the metric space, which can be
moved from one point to another at a cost which is the distance between
the two points. Moving a server to a point at which there exists a
pending request serves that request. In \cite{DBLP:conf/stoc/AzarGGP17},
an $O(\log^{4}n)$-competitive randomized algorithm is given for the
problem, where $n$ is the number of points in the metric space. This
algorithm encompasses a random embedding to an hierarchical well-separated
tree (HST) of depth $h=O(\log n)$, and an $O(h^{3})$-competitive
deterministic algorithm for online service with delay on HSTs. In
this paper, we also improve this result to $O(\log^{2}n)$ competitiveness.

In addition, we also present the problem of \textbf{online facility
location with deadlines}. In this problem, requests arrive over time
on points of a metric space, each equipped with a deadline. The algorithm
can open a facility at any point of the metric space, at some fixed
cost. Immediately upon opening a facility, the algorithm may connect
any number of pending requests to that facility, serving these requests.
Connecting a request to a facility incurs a connection cost which
is the distance between the location of the request and the location
of the facility. In contrast to previous considerations of online
facility location, in our problem the facility is only opened momentarily,
disappearing immediately after connecting the requests. We also consider
the problem of \textbf{online facility location with delay}, in which
the deadlines are replaced with arbitrary delay functions. For those
problems we present $O(\log^{2}n)$-competitive algorithms, with $n$
the number of points in the metric space.

The problem of facility location is a widely researched classic problem.
The modification of ephemeral facilities is highly motivated, as it
describes an option of renting facilities instead of buying them.
As renting shared resources is a growing trend (e.g. in cloud computing),
this problem captures many practical scenarios.

Our paper presents algorithms for online facility location with deadlines,
online facility location with delay, online multilevel aggregation
with delay and online service with delay. These algorithms all share
a common framework that we develop. The framework includes techniques
for both the design of the algorithms and their analysis. We believe
the flexibility and generality of this framework would enable designing
and analyzing algorithms for additional problems with deadlines or
with delay.

\subsection*{Our Results}

In this paper, we present a framework used to construct and analyze
online optimization problems with deadlines or with delay over a metric
space. Using this framework, we present the following algorithms.
\begin{enumerate}
\item An $O(D^{2})$-competitive deterministic algorithm for online multilevel
aggregation with delay on a tree of depth $D$. This is an exponential
improvement over the $O(D^{4}2^{D})$-competitive algorithm in \cite{DBLP:conf/esa/BienkowskiBBCDF16}.
\item An $O(\log^{2}n)$-competitive randomized algorithm for online service
with delay over a metric space with $n$ points. This improves upon
the $O(\log^{4}n)$-competitive randomized algorithm in \cite{DBLP:conf/stoc/AzarGGP17}.
\item An $O(\log^{2}n)$-competitive randomized algorithm for online facility
location with deadlines over a metric space with $n$ points.
\item An $O(\log^{2}n)$-competitive randomized algorithm for online facility
location with delay over a metric space with $n$ points.
\end{enumerate}
Our algorithms all share a common framework, which we present. The
framework provides general structure to both the algorithm and its
analysis.

Such an improvement for the online multilevel aggregation problem
is only known for the special case of deadlines, as given in \cite{DBLP:conf/soda/BuchbinderFNT17}.

The algorithms for online facility location with deadlines and with
delay can be easily extended to the case in which the cost of opening
a facility is different for each point in the metric space. This changes
the competitiveness of the algorithms to $O(\log^{2}\Delta+\log\Delta\log n)$,
where $\Delta$ is the aspect ratio of the metric space.

\subsection*{Our Techniques}

All of our algorithms are based on corresponding competitive algorithms
for HSTs. The randomized algorithms for general metric spaces are
obtained through randomized HST embedding. The $O(D^{2})$-competitive
deterministic algorithm for online multilevel aggregation with delay
on a tree is based on decomposing the tree into a forest of HSTs.
This decomposition is similar to that used in \cite{DBLP:conf/soda/BuchbinderFNT17}
for the case of deadlines.

\textbf{The framework -- algorithm design. }In designing algorithms
for the problems over HSTs, we use a certain framework. In an algorithm
designed using the framework, there is a counter for every node (in
the case of facility location) or every edge (in the case of online
multilevel aggregation and service with delay). The sizes of the counters
vary between the problems considered. When the counter for a tree
element (either node or edge) is full, the algorithm resets the counter
and explores the subtree rooted at that element. 

The process of exploration serves some of the pending requests at
that subtree, while simultaneously charging counters of descendant
tree elements. The exploration takes place in a DFS fashion -- if
at any time during the exploration of an element the counter of a
descendant element is full, the algorithm immediately suspends the
exploration of the current element in favor of its descendant. The
exploration of the original element resumes only when the exploration
of the descendant is complete.

The exploration of specific element has a certain budget, used to
charge counters of descendants. This budget is equal to the size of
the counter of the element being explored. The algorithm adheres to
the budget very strictly, spending exactly the amount specified. This
is a crucial part of the framework, as exceeding the budget (or falling
below budget) by even a constant factor would yield a competitiveness
which is exponential in the depth of the tree. 

This DFS exploration method is very different from previous algorithms,
and enables us to get our improved results. The counter-based structure
of our algorithms enables this DFS exploration while controlling the
budget. Using the counter-based structure is, in turn, enabled by
the techniques that we present in our framework's analysis.

\textbf{The framework -- analysis.} The analysis of the algorithms
of this framework require constructing a \emph{preflow} - a weighted
directed graph which is similar to a flow network, but in which we
allow nonnegative excesses at nodes (i.e. more incoming than outgoing).
We refer to nodes of the preflow as \emph{charging nodes}. We construct
a source charging node, from which the output is proportional to the
cost of the optimum, and then use the preflow to propagate this output
throughout the preflow graph. Since the excesses are nonnegative,
the sum of the excesses of any subset of charging nodes is a lower
bound of the total output from the source charging node, and thus
also some lower bound on the cost of the optimum. We construct the
preflow in a manner that allows us to locate such a subset of high-excess
charging nodes, thus providing the required lower bound. 

In the preflows we construct, each tree element (node or edge) is
converted to multiple charging nodes, each corresponding to an exploration
of that tree element. The possible edges between charging nodes in
the preflow depend on the structure of the tree and the operation
of our algorithm. Of those possible edges, we describe a procedure
that chooses the actual edges of the preflow. This procedure depends
on the optimal solution. Though the original metric space is a tree,
the multiple copies of each tree element cause the resulting preflow
to be a general directed graph.

The goal of the preflow creation procedure is to propagate the optimum's
costs to some ``top layer'' of charging nodes. This top layer usually
consists of nodes corresponding to explorations of the root tree element,
though in the case of online service with delay the definition is
different. The charging nodes of that top layer are then chosen to
lower bound the optimum, as described.

The preflow creation procedure involves creating colors at the ``top''
layer of the charging nodes. These colors are then propagated, through
some set of propagation rules, to nodes in lower layers. Each color
corresponds to the charging node in which it originated, with the
exception of two colors -- the empty color, and an additional ``special''
color. As nodes are colored, the possible edges that contain them
become actual edges of the preflow.

We now discuss the techniques used in each of the problems in this
paper.

\textbf{Online facility location with deadlines. }We use our framework
in constructing an algorithm for this problem over an HST. The algorithm
maintains a counter on each node (other than the root node), such
that each counter is of size $f$, where $f$ is the cost of opening
a facility. Whenever a counter is full, it resets and triggers an
exploration of that node. Whenever the deadline of a pending request
expires, the algorithm starts an exploration of the root node. 

In the exploration of a node $u$, the algorithm opens a facility
at $u$, and considers pending requests in the node's subtree according
to increasing deadline. For each request considered, it raises the
counter of the child node on the path to the request by the cost of
connecting that request to $u$. If the counter of the child is full,
an exploration of that child is called recursively, which would surely
serve the considered request. Otherwise, the algorithm connects that
request to $u$. As per the framework, the budget of $u$'s exploration
for raising these counters is exactly $f$. 

\textbf{Online facility location with delay. }The algorithm for this
problem is an extension of the deadline case. An exploration of the
root node is now triggered upon a set of requests which is \emph{critical},
i.e. has accumulated large delay. 

The significant difference between the delay case and the deadline
case is in the exploration itself. In the deadline case, the exploration
of a node $u$ spends its budget attempting to ``push back'' the
next occurrence of a single event (i.e. the earliest deadline of a
pending request in the subtree rooted at $u$). In the delay case,
there are two events to consider. The first event is a single request
with a delay large enough to justify connection to $u$. The second
event is a ``coalition'' of many tightly-grouped requests with small
individual delay, but large overall delay. This coalition does not
justify connection to $u$, but does merit opening a facility near
the coalition.

\textbf{Online multilevel aggregation with delay.} In our algorithm
for this problem over HSTs, each edge has a counter. The size of the
counter is the weight of edge. This is in contrast to our algorithms
for the facility location problems, in which all counters were of
the same size. We assume, without loss of generality, that there exists
a single edge exiting the root node, called the root edge. As in the
facility location case, an exploration of the root edge is triggered
when the delay of a set of requests becomes high.

In our algorithm, exploring an edge means adding descendant edges
to the transmitted subtree. The explored edge again has a budget equal
to its weight. The exploration repeatedly chooses the earliest point
in time in which the delay of a set of requests exceeds the cost of
expanding the transmission to include these requests. It then raises
the counter of the descendant edge in the direction of that request
set. Note the contrast with the algorithms for facility location --
the explored edge is allowed to raise the counters of its descendant
edges, and not just of its immediate children. 

While the analysis for our facility location problems required constructing
a single preflow to get a lower bound on the cost of the optimum,
the analysis for online multilevel aggregation with delay requires
constructing an additional preflow to get an upper bound on the cost
of the algorithm.

\textbf{Online service with delay. }Our algorithm for this problem
uses the exploration method of the algorithm for online multilevel
aggregation with delay. However, the tree to be explored is not the
entire tree, but rather some subtree according to the location of
the server. The concepts of relative trees and major edges are defined
in a similar way to \cite{DBLP:conf/stoc/AzarGGP17}. We also use
a potential function based on the distance of the algorithm's server
from the optimum's server. As the algorithm consists (mainly) of making
calls to the multilevel aggregation exploration, the analysis divides
these explorations to those for which the optimum can be charged (using
similar arguments to the analysis of the multilevel aggregation algorithm),
and explorations for which the costs are covered by the potential
function.

\subsection*{Related Work}

The online multilevel aggregation problem generalizes a range of studied
problems, such as the TCP acknowledgment problem \cite{TCPAck_DBLP:conf/esa/BuchbinderJN07,TCPAck_DBLP:conf/stoc/DoolyGS98,TCPAck_DBLP:journals/algorithmica/KarlinKR03}
and the joint replenishment problem \cite{JointRep_DBLP:conf/soda/BienkowskiBCJNS14,JointRep_DBLP:journals/algorithmica/BritoKV12,JointRep_DBLP:conf/soda/BuchbinderKLMS08}.
For both the deadline and delay variants of online multilevel aggregation,
the best known lower bounds are only constant \cite{JointRep_DBLP:conf/soda/BienkowskiBCJNS14}.
Bienkowski et al. \cite{DBLP:conf/esa/BienkowskiBBCDF16} were the
first to present an algorithm for the online multilevel aggregation
problem with arbitrary delay functions, which is $O(D^{4}2^{D})$-competitive.
Buchbinder et al. \cite{DBLP:conf/soda/BuchbinderFNT17} presented
an $O(D)$-competitive algorithm for the special case of deadlines. 

The problem of online service with delay was presented in \cite{DBLP:conf/stoc/AzarGGP17},
along with the $O(\log^{4}n)$-competitive randomized algorithm for
a general metric space of $n$ points. The problem has also been studied
over specific metric spaces, such as uniform metric and line metric,
in which improved results can be achieved \cite{DBLP:conf/stoc/AzarGGP17,DBLP:conf/sirocco/BienkowskiKS18}.

Another metric optimization problem with delay is the problem of matching
with delay \cite{Matching_DBLP:conf/approx/AshlagiACCGKMWW17,Matching_DBLP:conf/ciac/EmekSW17,Matching_DBLP:conf/stoc/EmekKW16,Matching_DBLP:conf/waoa/AzarF18,Matching_DBLP:conf/waoa/BienkowskiKLS18,Matching_DBLP:conf/waoa/BienkowskiKS17}.
For this problem, arbitrary delay functions are intractable, and thus
the main line of work focuses on linear delay functions. 

Additional problems with delay exist other than those over a metric
space. The set aggregation problem, presented in \cite{DBLP:conf/latin/CarrascoPSV18},
is a variant of set cover with delay. The problem of bin packing with
delay is presented in \cite{conf/spaa/Azar19}.

The classic online facility location problem, suggested by Meyerson
\cite{OFL_DBLP:conf/focs/Meyerson01}, has also been studied \cite{OFL_DBLP:journals/algorithmica/Fotakis08,OFL_DBLP:journals/iandc/AnagnostopoulosBUH04}.
In this problem, requests arrive one after the other, and the algorithm
must either connect a request to an existing facility immediately
upon the request's arrival, or open a facility at the request's location.
This problem is different from the problems of facility location with
deadlines and facility location with delay presented in this paper.
The main difference is that in our problems, a facility is only opened
momentarily, which only allows immediate connection of pending requests.
In contrast, an opened facility in the online facility location of
\cite{OFL_DBLP:conf/focs/Meyerson01} is permanent, allowing the connection
of any future request to that facility.

\paragraph*{Paper Organization}

Section \ref{sec:FLDeadline} presents the problem of online facility
location with deadlines, and an $O(\log^{2}n)$-competitive randomized
algorithm for the problem, as well as its analysis. Section \ref{sec:FLDelay}
discusses the more general problem of online facility location with
delay, and extending the algorithm for the deadline case in section
\ref{sec:FLDeadline} to an $O(\log^{2}n)$-competitive algorithm
for the case of delay. 

Section \ref{sec:MAD} presents the $O(D^{2})$-competitive deterministic
algorithm for online multilevel aggregation with delay. Section \ref{sec:OSD}
presents the $O(\log^{2}n)$-competitive randomized algorithm for
online service with delay, which relies on the algorithm for online
multilevel aggregation with delay given in Section \ref{sec:MAD}.

\section{\label{sec:FLDeadline}Online Facility Location with Deadlines}

\subsection{Problem and Notation}

In the online facility location with deadlines problem, requests arrive
on points of a metric space over time. Each request is associated
with a deadline, by which it must be served. An algorithm for the
problem can choose, in any point in time, to open a facility at any
point in the metric space momentarily, at a cost of $f$. Immediately
upon opening the facility, the algorithm must choose the subset of
pending requests (i.e. requests that have arrived but have not been
served) to connect to the facility. The cost of connecting each request
to the facility is the distance between the request's location and
the facility's location. Connecting a request to a facility serves
that request. Immediately after connecting the requests, the facility
disappears. We allow opening a facility at the same point more than
once, at different times.

Formally, we are given a metric space $\A=(A,\delta_{\A})$ such that
$|A|=n$. A request is a tuple $q=(v_{q},r_{q},d_{q})$ such that
$v_{q}$ is a point in $\A$, the arrival time of the request is $r_{q}$
and the deadline of the request is $d_{q}$. We assume, without loss
of generality, that all deadlines are distinct. For any instance of
the problem, the algorithm's solution has two costs. The first is
the \emph{buying cost} (or \emph{opening cost}) $\alg^{B}=mf$, where
$m$ is the number of facilities opened by the algorithm. Denoting
by $Q$ the set of requests in the instance, and denoting by $\beta_{q}$
the location of the facility to which the algorithm connects request
$q$, the second cost of the algorithm is the \emph{connection cost}
$\alg^{C}=\sum_{q\in Q}\delta_{\A}(v_{q},\beta_{q})$. Wherever a
single metric space $\A$ is considered, we write $\delta=\delta_{\A}$.

The goal of the algorithm is to minimize the total cost, which is
\[
\alg=\alg^{B}+\alg^{C}
\]

For the special case in which $A$ is a tree $T$, and $\delta$ is
the distance between nodes in $T$, we denote the root of $T$ by
$r$ and the weight function on the edges of the tree by $w$. We
assume, without loss of generality, that the requests only arrive
on leaves of the tree.

The following definitions regarding trees are used throughout the
paper.
\begin{defn}
\label{def:FLDeadline_TreeDefinitions}For every tree node $u\in T$,
we use the following notations:
\begin{itemize}
\item For $u\neq r$, we denote by $p(u)$ the parent of $u$ in the tree. 
\item We denote by $T_{u}$ the subtree rooted at $u$.
\item For a set of requests $Q\subseteq T_{u}$, we denote by $T_{u}^{Q}\subseteq T_{u}$
the subtree spanned by $u$ and the leaves of $Q$.
\item We define the \emph{height} of $u$ to be the depth of $T_{u}$.
\end{itemize}
\end{defn}

The following definition is similar to the usual definition of a $\beta$-HST,
except that we allow a child edge to be strictly smaller than $\frac{1}{\beta}$
times its parent edge.
\begin{defn}[$\left(\ge\beta\right)$-HST]
 A rooted tree $T$ is a $\left(\ge\beta\right)$-HST if for any
two edges $e,e^{\prime}\in T$ such that $e$ is a parent edge of
$e^{\prime}$, we have that $w(e)\ge\beta w(e^{\prime})$.
\end{defn}

When considering the problem over a tree $T$, we assume, without
loss of generality, that $w(e)\le f$ for any edge $e\in T$. Indeed,
if this is not the case, no request would be connected over $e$,
effectively yielding two disjoint instances of the problem.

In this section, we prove the following theorem.
\begin{thm}
\label{thm:FLDeadline_GMSLogSquared}There exists an $O(\log^{2}n)$-competitive
randomized algorithm for online facility location with deadlines for
any metric space of $n$ points.
\end{thm}

\subsection{Algorithm for HSTs}

We present an algorithm for facility location with deadlines on a
$\left(\ge2\right)$-HST $T$ of depth $D$. We denote the root of
the tree by $r$.

We make the assumption that the total weight of any path from the
root to a leaf is at most $f$. In a $\left(\ge2\right)$-HST, the
total weight of such a path is at most twice the weight of the top
edge, which is at most $f$. Thus, this assumption only costs us a
constant factor of $2$ in competitiveness.

Without loss of generality, we allow the algorithm to open facilities
on internal nodes of the tree. Indeed, any algorithm that opens facilities
on internal nodes can be converted to an algorithm that only opens
facilities on leaves in the following manner. Consider a facility
opened by the original algorithm on the internal node $u$, and denote
by $Q$ the set of requests connected to that facility. The modified
algorithm would open the facility at $v_{q^{\ast}}$ instead, where
$q^{\ast}=\arg\min_{q\in Q}\delta(u,v_{q})$, and connect the original
requests. Through triangle inequality, the connection cost of the
modified algorithm is at most twice larger.

\textbf{Algorithm's description. }The algorithm for facility location
with deadlines on a $\left(\ge2\right)$-HST is given in Algorithm
\ref{alg:FLDeadline}. The algorithm waits until the deadline of a
pending request. It then begins exploring the root node. An exploration
of a node $u$ consists of considering the pending requests in $T_{u}$
by order of increasing deadline. The exploration has a budget of exactly
$f$ to spend on raising counters of child nodes -- it maintains
that budget in the variable $b_{u}$. When considering a request $q$,
the algorithm raises the counter of the child node $v$, denoted $c_{v}$
in the algorithm, for the child node $v$ in the request's direction.
The counter is raised by the smallest of $\delta(v_{q},u)$, the amount
required to fill $c_{v}$, and the remaining budget $b_{u}$. If this
fills the counter of $v$, the exploration of $u$ is paused, and
a new exploration of $v$ is started, in a DFS manner. We claim, in
the analysis, that this exploration of $v$ connects $q.$ Otherwise,
the request $q$ is connected to $u$. 

The operation of the algorithm is visualized in Figure \ref{fig:FLDeadline_AlgorithmVisualization}
of Appendix \ref{appendix:AdditionalFigures}.

\LinesNumbered \RestyleAlgo{boxruled}\DontPrintSemicolon\newcommand\mycommfont[1]{\emph{#1}}
\SetCommentSty{mycommfont}
\begin{algorithm}[tb]
\caption{\label{alg:FLDeadline}Facility Location with Deadlines}

\SetKwProg{Fn}{Function}{}{end}
\SetKwProg{EFn}{Event Function}{}{end}
\SetKwFunction{UponDeadline}{UponDeadline}
\SetKwFunction{Open}{Open}
\SetKwFunction{Invest}{Invest}
\SetKwFunction{Explore}{Explore}

\textbf{Initialization.}

Initialize $c_{u}\leftarrow0$ for any node $u\in T\backslash\{r\}$.

Declare $b_{u}$ for any node $u\in T$.

\;

\EFn(\tcp*[h]{Upon expired deadline of pending request at time $t$}){\UponDeadline{}}{

\Explore{$r$}

}

\;

\Fn{\Explore{$u$}}{

\Open{$u$}

\tcp*[h]{Spend a budget of $f$ on charging child node counters}

set $b_{u}\leftarrow f$

\While{$b_{u}\neq0$ \normalfont{\textbf{and}} there remain pending
requests in $T_{u}$}{

\tcp*[h]{Consider pending requests by increasing deadline}

let $q$ be the pending request with earliest deadline in $T_{u}$

let $v$ be the child of $u$ on the path to $v_{q}$

call \Invest{$u$,$v$,$\delta(u,v_{q})$}

\lIf{$c_{v}=f$}{set $c_{v}\leftarrow0$ $\texttt{\textbf{;}}$
call \Explore{v}.} 

\lIf{$q$ is still pending}{connect $q$ to facility at $u$}

}

}
\end{algorithm}
\LinesNumbered \RestyleAlgo{boxruled}\DontPrintSemicolon
\begin{algorithm}[tb]
\caption{Facility Location with Deadlines}

\ContinuedFloat   \caption*{Facility Location with Deadlines (cont.)}
\SetKwProg{Fn}{Function}{}{end}
\SetKwFunction{UponDeadline}{UponDeadline}
\SetKwFunction{Open}{Open}
\SetKwFunction{Invest}{Invest}
\SetKwFunction{Explore}{Explore}

\Fn{\Open{$u$}}{

open facility at $u$.

\lIf{$u$ is a leaf node}{connect to facility all pending requests
on $u$}

}

\;

\Fn(\tcp*[h]{%
\parbox[t]{0.6\textwidth}{%
Invests in $v$'s counter either $x$, or until $v$'s counter \\
is full, or until $u$ is out of budget.%
}}){\Invest{$u$,$v$,$x$}}{

let $y\leftarrow\min(x,b_{u},f-c_{v})$

increase $c_{v}$ by $y$

decrease $b_{u}$ by $y$

\Return $y$

}
\end{algorithm}

\subsection{Analysis}

Fix any instance of online facility location with deadlines on a $\left(\ge2\right)$-HST.
Let $\opt$ be any solution to the instance. We denote by $\opt^{B}$
the total buying cost of $\opt$, and by $\opt^{C}$ the total connection
cost of $\opt$. Denote by $\alg$ the total cost of the solution
of Algorithm \ref{alg:FLDeadline} for this problem. In this subsection,
we prove the following theorem.
\begin{thm}
\label{thm:FLDeadline_HSTTheorem}$\alg\le O(D^{2})\cdot\opt^{B}+O(D)\cdot\opt^{C}$.
\end{thm}

To prove Theorem \ref{thm:FLDeadline_HSTTheorem}, we show validity
of the algorithm, an upper bound for $\alg$ and a lower bound for
$\opt$.

Throughout the analysis, we denote by $k$ the number of calls to
$\UponDeadline$ made by the algorithm. We also denote by $t_{1},...,t_{k}$
the times of these $k$ calls, by increasing order.

\subsubsection{Validity of the Algorithm}

The following proposition and its corollary show that the algorithm
is valid. 
\begin{prop}
\label{prop:FLDeadline_ConsideredMeansServed}Let $q$ be a request
considered in a call to $\Explore(u)$. Then $q$ is served when $\Explore(u)$
returns.
\end{prop}

\begin{proof}
This is guaranteed by the condition check at the end of the main loop
in $\Explore$.
\end{proof}
\begin{cor}
Every request is served by its deadline. That is, the algorithm is
valid.
\end{cor}

\begin{proof}
Observe that upon the deadline of a request $q$, $\Explore(r)$ is
called, and immediately considers $q$. Proposition \ref{prop:FLDeadline_ConsideredMeansServed}
concludes the proof.
\end{proof}

\subsubsection{Upper Bounding $\protect\alg$}

We now proceed to bound $\alg$ by proving the following lemma.
\begin{lem}
\label{lem:FLDeadline_ALG}$\alg\le3\cdot(D+1)\cdot kf$.
\end{lem}

The proof of Lemma \ref{lem:FLDeadline_ALG} is through providing
an upper bound for the cumulative amount by which counters are raised
in the algorithm, then bounding the cost of the algorithm by that
cumulative amount.
\begin{obs}
\label{obs:FLDeadline_RaisesCounterAtMostF}Observe any node $u$,
and consider a call to $\Explore(u)$. Denote by $x$ the total amount
by which $\Explore(u)$ increases the counters of its children nodes
through calls to $\Invest$. Then we have that $x\le f$. Moreover,
if there exists a pending request in $T_{u}$ after the return of
$\Explore(u)$, then $x=f$.
\end{obs}

From the previous observation, the following observation follows.
\begin{obs}
\label{obs:FLDeadline_SingleService}For any $u$, $\Explore(u)$
is called at most once at any time $t$.
\end{obs}

Using the last observation, we refer to a call to $\Explore(u)$ at
time $t$ by $\Explore_{t}(u)$. 

Observe the state of each counter in the algorithm over time. The
counter undergoes phases, such that in the start of each phase its
value is $0$. The counter increases in value during the phase until
it reaches $f$, and is then reset to $0$, triggering a service and
the end of the phase. 

We define a virtual counter $\bar{c}_{u}$ which contains the cumulative
value of $c_{u}$. That is, whenever $c_{u}$ increases, $\bar{c}_{u}$
increases by the same amount, but $\bar{c}_{u}$ is never reset when
$c_{u}$ is reset. For the sake of analysis, we also consider a virtual
counter $\bar{c}_{r}$, which is raised by $f$ whenever $\Explore(r)$
is called.

We define $\bar{C}_{j}=\sum_{\text{node \ensuremath{u} at depth \ensuremath{j}}}\bar{c}_{u}$.
Observe that $\bar{C}_{0}=\bar{c}_{r}$. 
\begin{prop}
\label{prop:FLDeadline_CounterDecreaseWithLevel}For every $j\in[D]$,
$\bar{C}_{j}\le\bar{C}_{j-1}$.
\end{prop}

\begin{proof}
Observe that the counters at depth $j$ are raised only upon a call
to $\Explore(u)$ for a node $u$ at depth $j-1$. $\texttt{\ensuremath{\Explore}}(u)$
is only called after $\bar{c}_{u}$ is raised by $f$, and every such
call raises counters at depth $j$ by at most $f$ (using Observation
\ref{obs:FLDeadline_RaisesCounterAtMostF}).
\end{proof}
\begin{cor}
\label{cor:FLDeadline_CountersBoundedByDKF}$\sum_{u\in T}\bar{c}_{u}\le(D+1)kf$.
\end{cor}

\begin{proof}
Observe that $\bar{C}_{0}=\bar{c}_{r}=kf$. Using Proposition \ref{prop:FLDeadline_CounterDecreaseWithLevel},
we have that 
\[
\sum_{u\in T}\bar{c}_{u}=\sum_{j=0}^{D}\bar{C}_{j}\le\sum_{j=0}^{D}\bar{C}_{0}=(D+1)kf
\]
\end{proof}
\begin{prop}
\label{prop:FLDeadline_InvestServes}Suppose the function $\Explore(u)$
calls $\Invest(u,v,x)$ when considering request $q$. Then at least
one of the following holds:
\begin{enumerate}
\item $\Invest(u,v,x)$ returns $x$.
\item $b_{u}=0$ after the return of $\Invest$.
\item The condition check in $\Explore$ of whether $q$ is still pending
fails. 
\end{enumerate}
\end{prop}

\begin{proof}
If $\Invest(u,v,x)$ does not return $x$, and $b_{u}\neq0$ after
its return, then it must be that $c_{v}=f$. In this case, $\Explore(u)$
calls $\Explore(v)$ when checking the condition after the return
of $\Invest$. Request $q$ is the pending request with earliest deadline
under $T_{u}$, and thus also under $T_{v}$. Hence, $q$ is immediately
considered by $\Explore(v)$, and is thus served by the end of $\Explore(v)$
by Proposition \ref{prop:FLDeadline_ConsideredMeansServed}.
\end{proof}
\begin{prop}
\label{prop:FLDeadline_ALGBoundedByCounters}$\alg\le3\cdot\sum_{u\in T}\bar{c}_{u}$
\end{prop}

\begin{proof}
The costs of the algorithm (both opening and connection) are contained
in calls to the function $\texttt{\ensuremath{\Explore}}$ (where
we associate the opening costs in $\Open$ to the $\Explore$ call
that invoked it). In each call to $\texttt{\ensuremath{\Explore}}(u)$,
the algorithm has a cost of $f$ in opening a facility at $u$. 

In addition, the algorithm incurs connection costs, as $\Explore(u)$
connects any considered request if it is still pending at the end
of the loop's iteration. From Proposition \ref{prop:FLDeadline_InvestServes},
if $\Explore(u)$ connects a request $q$, then either the preceding
$\Invest(u,v,\delta(u,v_{q}))$ returned $\delta(u,v_{q})$, or $b_{u}=0$
after the return of that call to $\Invest$. 

Observe the calls to $\Invest(u,v,\delta(u,v_{q}))$ that return $\delta(u,v_{q})$.
For those requests, the connection cost of $q$ is exactly the return
value of $\Invest$. But the return values of $\Invest$ sum to at
most the initial value of $b_{u}$, which is $f$. Thus, connection
costs for those requests sum to at most $f$.

As for calls to $\Invest$ after which we have that $b_{u}=0$, observe
that there is at most one such call, after which the loop in $\Explore$
ends. The connection cost for the request considered in this iteration
is $\delta(u,v_{q})\le f$.

Overall, the connection costs in $\Explore(u)$ sum to at most $2f.$

Thus, in each call to $\texttt{\ensuremath{\Explore}}(u)$, the total
cost of the algorithm (buying and connection) is at most $3f$. Observing
that $\texttt{\ensuremath{\Explore}}(u)$ is called only upon raising
$\bar{c}_{u}$ by $f$ concludes the proof.
\end{proof}
\begin{proof}[Proof of Lemma \ref{lem:FLDeadline_ALG}]
The lemma results directly from Proposition \ref{prop:FLDeadline_ALGBoundedByCounters}
and Corollary \ref{cor:FLDeadline_CountersBoundedByDKF}.
\end{proof}

\subsubsection{Lower Bounding $\protect\opt$}

We now lower bound the cost of $\opt$. 

\paragraph*{Charging nodes and incurred costs.}

We define a charging node to be a tuple $(u,[\tau_{1},\tau_{2}])$
such that $u\in T$, and $\tau_{1},\tau_{2}$ are two subsequent times
in which $\Explore(u)$ is called. We allow the charging nodes of
the form $(u,[\tau_{1},\tau_{2}])$ in which $\tau_{1}=-\infty$ and
$\tau_{2}$ is the first time in which $\Explore(u)$ is called. Similarly,
we allow the charging nodes $(u,[\tau_{1},\tau_{2}])$ in which $\tau_{1}$
is the last time $\Explore(u)$ is called, and $\tau_{2}=\infty$.
We denote by $M$ the set of charging nodes.

For a charging node $\mu=(u,[\tau_{1},\tau_{2}])$, we define the
following.
\begin{enumerate}
\item Let \emph{$c_{b}(\mu)$ }be the \emph{buying cost incurred by} $\opt$
\emph{in $\mu$}, defined to be the total cost at which $\opt$ opened
facilities in $T_{u}$ during $[\tau_{1},\tau_{2}]$.
\item Let $c_{c}(\mu)$ be the \emph{connection cost incurred by $\opt$
in $\mu$, }defined to be $\sum_{q\in Q}\delta(p(u),v_{q})$, where
$Q$ is the set of requests $q$ such that $v_{q}\in T_{u}$, $r_{q}\in[\tau_{1},\tau_{2}]$
and $\opt$ connected $q$ to a facility outside $T_{u}$.
\end{enumerate}
Let $c(\mu)=c_{b}(\mu)+c_{c}(\mu)$ be the \emph{total cost $\opt$
incurred in $\mu$}.

\begin{figure}[tb]
\subfloat[\label{subfig:ChargingNodeVisualization_SingleNode}Charging Nodes
for Single Tree Node]{\includegraphics[width=0.45\columnwidth]{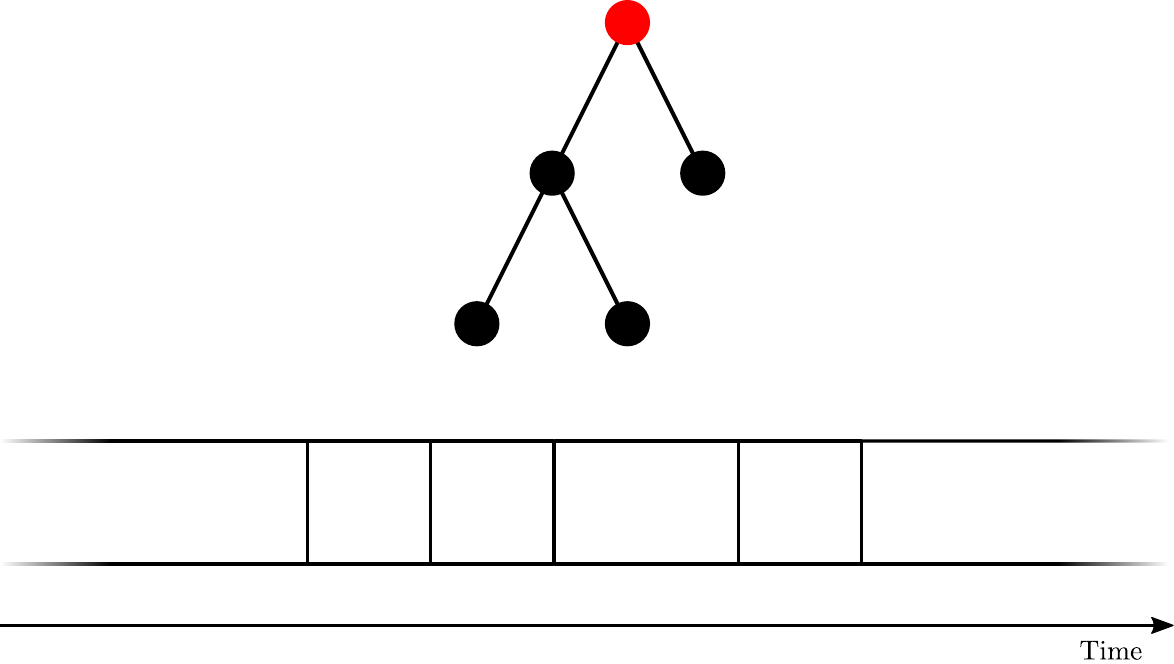}

}\hfill{}\subfloat[\label{subfig:ChargingNodeVisualization_Branch}Charging Nodes for
a Branch]{\includegraphics[width=0.45\columnwidth]{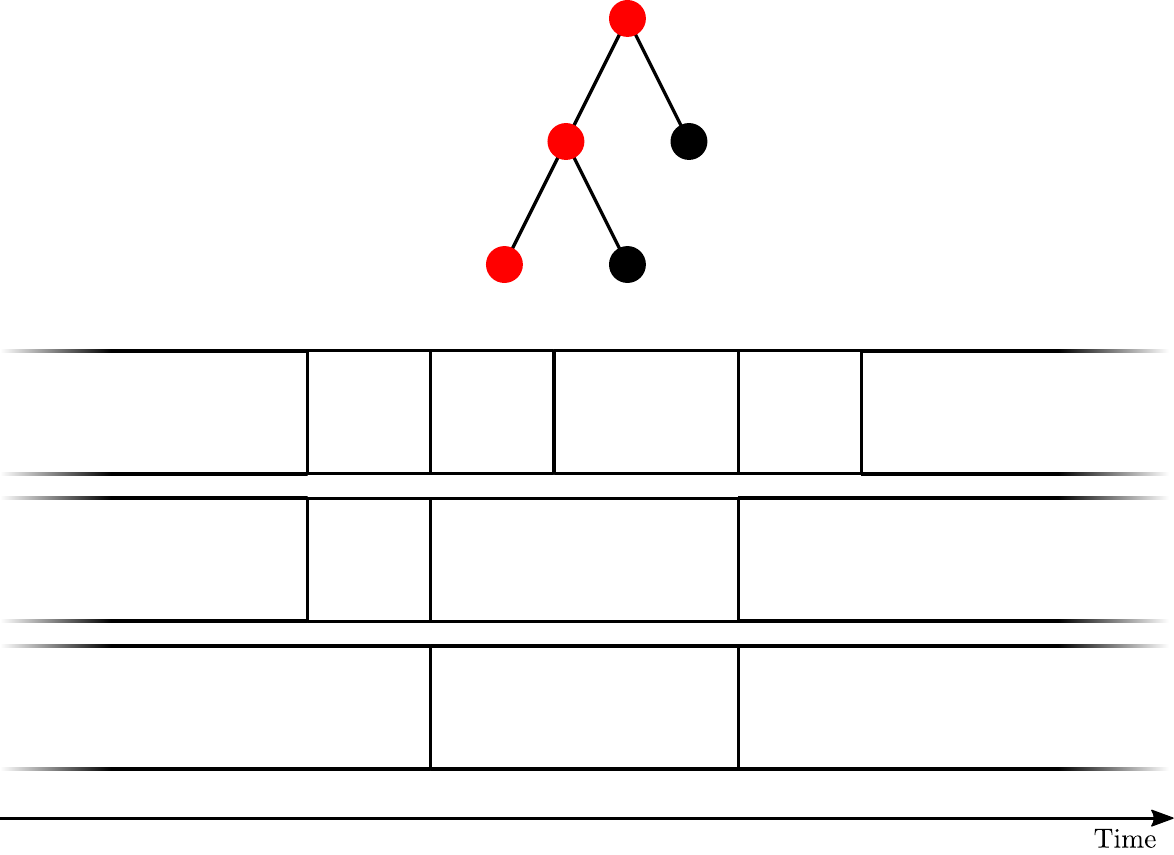}

}\\

Sub-figure \ref{subfig:ChargingNodeVisualization_SingleNode} is a
visualization of the charging nodes corresponding to a single tree
node, displayed over a timeline. Each rectangle is a charging node.
Note the charging node from $-\infty$ and the charging node to $\infty$.

Sub-figure \ref{subfig:ChargingNodeVisualization_Branch} shows the
charging nodes corresponding to a branch in the tree. Observe the
containment of charging node intervals for a certain tree node in
the intervals of descendant tree nodes.

\caption{Visualization of Charging Nodes}
\end{figure}

\begin{lem}
\label{lem:FLDeadline_ChargeLBOpt}$\sum_{\mu}c(\mu)\le2(D+1)\cdot\opt^{B}+4\cdot\opt^{C}$.
\end{lem}

\begin{proof}
We consider each action of $\opt$ and how it affects the incurred
cost at various charging nodes.

For a facility that is opened by $\opt$ at node $u$ at time $t$
to participate in $c_{b}((u^{\prime},[\tau_{1},\tau_{2}]))$, we must
have that $u\in T_{u^{\prime}}$. Hence, $u^{\prime}$ is on the branch
from the root to $u$, and thus $u^{\prime}$ is one of at most $D+1$
possible nodes. We also have that $t\in[\tau_{1},\tau_{2}]$. Using
Observation \ref{obs:FLDeadline_SingleService}, we have that $\tau_{2}>\tau_{1}$,
and therefore $t$ can belong to at most two such intervals, for every
choice of $u^{\prime}$. This yields that the cost of each facility
opened by $\opt$ is counted in \textbf{$\sum_{\mu}c_{b}(\mu)$ }at
most $2(D+1)$ times.

As for connection costs, consider a request $q$ that $\opt$ connects
to a facility at node $v$. Denote by $u$ the least common ancestor
of $v$ and $v_{q}$. If $\opt$ incurs connection cost due to $q$
in charging node $(u^{\prime},[\tau_{1},\tau_{2}])$, then $v_{q}\in T_{u^{\prime}}$
and $v\notin T_{u^{\prime}}$. Therefore, $u^{\prime}$ must be on
the path from $v_{q}$ to $u$ (including $v_{q}$, and not including
$u$). Let $u=u^{(0)},u^{(1)},u^{(2)},...,u^{(m)}=v_{q}$ be the path
from $u$ to $v_{q}$. As with the buying cost, for every $l\in[m]$,
$\opt$ may incur connection cost due to $q$ in at most $2$ charging
nodes of the form $(u^{(l)},[\tau_{1},\tau_{2}])$ for some $\tau_{1},\tau_{2}$.
Therefore, denoting the total connection cost of $\opt$ due to connecting
$q$ by $X$, we have that 
\[
X\le2\cdot\sum_{l=1}^{m}\delta(u^{(l-1)},v_{q})
\]

Now observe that since the tree is a $\left(\ge2\right)$-HST, we
have that the total weight of any path from a node to a descendant
leaf is at most the weight of the node's ancestor edge. Therefore,
for every $l\in[m]$:
\begin{align*}
\delta(u^{(l-1)},v_{q}) & =w((u^{(l-1)},u^{(l)}))+\delta(u^{(l)},v_{q})\\
 & \ge2\delta(u^{(l)},v_{q})
\end{align*}
Therefore, we have that $\delta(u^{(l)},v_{q})\le\frac{1}{2^{l}}\cdot\delta(u,v_{q})$.
Hence:
\[
X\le2\cdot\sum_{l=1}^{m}\frac{1}{2^{l-1}}\cdot\delta(u,v_{q})\le4\delta(u,v_{q})
\]

Since $u$ is on the path from $v$ to $v_{q}$, we have that $\delta(u,v_{q})\le\delta(v,v_{q})$.
Hence $X$ is at most $4$ times the connection cost of $\opt$ for
$q$. This concludes the lemma.
\end{proof}
\begin{defn}[excess]
Let $G=(V^{\prime},E)$ be a directed multigraph, with a non-negative
weight function $\alpha:E\rightarrow R^{+}$ defined on its edges.
We denote by $E_{v}^{+}\subseteq E$ the set of edges entering node
$v$, and by $E_{v}^{-}\subseteq E$ the set of edges leaving $v$.
We define the \emph{excess at a node $v\in V^{\prime}$ }to be $\chi_{v}=\sum_{\sigma\in E_{v}^{+}}\alpha(\sigma)-\sum_{\sigma\in E_{v}^{-}}\alpha(\sigma)$. 
\end{defn}

Note that every edge $\sigma\in E$ from $u$ to $v$ is counted in
$\chi_{u}$ and $\chi_{v}$ with opposite signs. The following observation
follows.
\begin{obs}
\label{obs:Flow_ExcessSumToZero}For any $G=(V^{\prime},E)$ and weights
$\alpha:E\to\R^{+}$, we have $\sum_{v\in V^{\prime}}\chi_{v}=0$.
\end{obs}

\begin{defn}
For a graph $G=(V^{\prime}=V\cup\{s\},E)$ and non-negative weights
$\alpha:E\rightarrow\R^{+}$, We say that $Z=(G,s,\alpha)$ is a \emph{preflow
}if for every node $v\neq s$ we have that $\chi_{v}\ge0$. We call
$s$ the \emph{source node }of the preflow.

Observation $\ref{obs:Flow_ExcessSumToZero}$ yields that $\chi_{s}\le0$
for every preflow $Z=(G,s,\alpha)$. We write $\omega_{Z}=-\chi_{s}$.
\end{defn}

\begin{prop}
\label{prop:Flow_SubsetLBSource}For $G=(V\cup\{s\},E)$ a directed
graph, for weights $\alpha:E\to\R^{+}$ such that $Z=(G,s,\alpha)$
is a preflow, and for every $S\subseteq V$, we have $\sum_{v\in S}\chi_{v}\le\omega_{Z}$.
\end{prop}

\begin{proof}
Observation \ref{obs:Flow_ExcessSumToZero} and the definition of
a preflow, we get $\sum_{v\in S}\chi_{v}\le\sum_{v\in V}\chi_{v}=-\chi_{s}=\omega_{Z}$.
\end{proof}

We now construct a preflow to lower bound $\opt$. The graph $G$
underlying the preflow has the set of nodes $M\cup\{s\}$, where $M$
is the set of charging nodes and $s$ is a source node.

Consider a charging node $\mu=(u,[\tau_{1},\tau_{2}])$. We have that
$[\tau_{1},\tau_{2}]$ corresponds to a phase of the counter $c_{u}$,
since $c_{u}$ was empty at $\tau_{1}$ and was filled and emptied
again until time $\tau_{2}$.
\begin{defn}[Investing]
Observe two charging nodes $\mu=(u,[\tau_{1},\tau_{2}])$ and $\mu^{\prime}=(u^{\prime}=p(u),[\tau_{1}^{\prime},\tau_{2}^{\prime}])$.
We say that $\mu^{\prime}$ \emph{invested $x$ in $\mu$ }if the
function call $\Explore_{\tau_{1}^{\prime}}(u^{\prime})$ increased
$c_{u}$ by $x$, through calls to $\Invest$, during the phase of
$c_{u}$ between $\tau_{1}$ and $\tau_{2}$.
\end{defn}

\begin{defn}[$\lambda_{u}^{t}$ and $\lambda_{\mu}$]
For every function call $\Explore_{t}(u)$ for some $u\in T$ and
time $t$, we denote by $\lambda_{u}^{t}$ the earliest deadline of
a pending request in $T_{u}$ immediately after the return of $\Explore_{t}(u)$
(if there are no pending requests in $T_{u}$, we write $\lambda_{u}^{t}=\infty$).

In addition, for a charging node $\mu=(u,[\tau_{1},\tau_{2}])$ with
$\tau_{1}\neq-\infty$, we write $\lambda_{\mu}=\lambda_{u}^{\tau_{1}}$.
\end{defn}

\paragraph*{Possible edges. }

We describe the set of possible edges in $G$ from nodes in $M$ to
other nodes in $M$, denoted by $\bar{E}$, and the weight function
$\alpha:\bar{E}\rightarrow\R^{+}$. The final set of edges added to
$G$ by Procedure \ref{proc:FLDeadline_PreflowBuilder} from the nodes
of $M$ to themselves is a subset of $\bar{E}$. The set $\bar{E}$
contains an edge $\sigma$ from any charging node $\mu_{1}=(u_{1},[\tau_{1}^{1},\tau_{2}^{1}])$
to any charging node $\mu_{2}=(u_{2},[\tau_{1}^{2},\tau_{2}^{2}])$
if $\mu_{1}$ invested in $\mu_{2}$. We set the weight $\alpha(\sigma)$
to be the amount that $\mu_{1}$ invested in $\mu_{2}$.

\noindent \LinesNumbered \RestyleAlgo{boxruled}\renewcommand{\algorithmcfname}{Procedure}\DontPrintSemicolon
\begin{algorithm}[tb]
\caption{\label{proc:FLDeadline_PreflowBuilder}PreflowBuilder - Facility Location
with Deadlines}

\SetKwProg{Fn}{Function}{}{end}
\SetKwFunction{PreflowBuilder}{PreflowBuilder}
\SetKwFunction{SetColor}{SetColor}
\SetKw{Break}{break}

\textbf{Initialization.}

Let the set of vertices of $G$ be $M\cup\{s\}$, and initialize $G$'s
edge set to be $E=\emptyset$.

Initialize dictionary $\texttt{Color}[\mu]=\None$ for every $\mu\in M$.

\ForEach{$\mu=(u,[\tau_{1},\tau_{2}])\in M$ such that $\opt$ opened
a facility in $T_{u}$ during $[\tau_{1},\tau_{2}]$}{

set $\Color[\mu]\leftarrow\Special$

}

\ForEach{$\mu\in M$ such that $c(\mu)>0$}{

add a new edge $\sigma=(s,\mu)$ to $E$, and set $\alpha(\sigma)=c(\mu)$

}

\;

\Fn{\PreflowBuilder{}}{

\tcp*[h]{Create new colors for root charging nodes}

\For{$i$ from $1$ to $k$}{

let $\mu\leftarrow(r,[t_{i-1}^{r},t_{i}^{r}])$

\SetColor{$\mu$,$\mu$}

}

\tcp*[h]{Propagate colors to other charging nodes}

\For{$j$ from $1$ to $D$}{

\ForEach{$\mu=(u,[\tau_{1},\tau_{2}])\in M$ such that $u$ is of
depth $j$}{

\ForEach{edge $\sigma\in E_{\mu}^{-}$ incoming to a node $\mu^{\prime}$}{

\lIf{\SetColor{$\mu,\texttt{Color}[\mu^{\prime}]$}$\neq None$}{\Break}

}

}

}

}

\end{algorithm}
\LinesNumbered \RestyleAlgo{boxruled}\renewcommand{\algorithmcfname}{Procedure}\DontPrintSemicolon
\begin{algorithm}[tb]
\caption{PreflowBuilder - Facility Location with Deadlines}

\ContinuedFloat   \caption*{PreflowBuilder - Facility Location with Deadlines (cont.)}
\SetKwProg{Fn}{Function}{}{end}
\SetKwFunction{SetColor}{SetColor}
\SetKw{Break}{break}

\Fn(\tcp*[h]{If $\Color[\mu]=\None$, tries to set it to $\mu^{\star}$}){\SetColor{$\mu,\mu^{\star}$}}{

let $[\tau_{1},\tau_{2}]$ be the interval of $\mu$ and let $[\tau_{1}^{\star},\tau_{2}^{\star}]$
be the interval of $\mu^{\star}$.

\If{\normalfont{(} $\Color[\mu]=\None$ \normalfont{\textbf{and}}
$\tau_{1}\neq-\infty$ \normalfont{\textbf{and}} $\lambda_{\mu}\le\tau_{2}^{\star}$
\normalfont{)}}{

add the edges in $\bar{E}$ incoming to $\mu$ to $E$.

set $\texttt{Color}[\mu]\leftarrow\mu^{\star}$

}

\Return $\texttt{Color}[\mu]$

}
\end{algorithm}
\begin{figure}[tb]
\subfloat[\label{subfig:FLDeadline_BeforeColoring}Initial state before coloring]{\includegraphics[width=0.45\columnwidth]{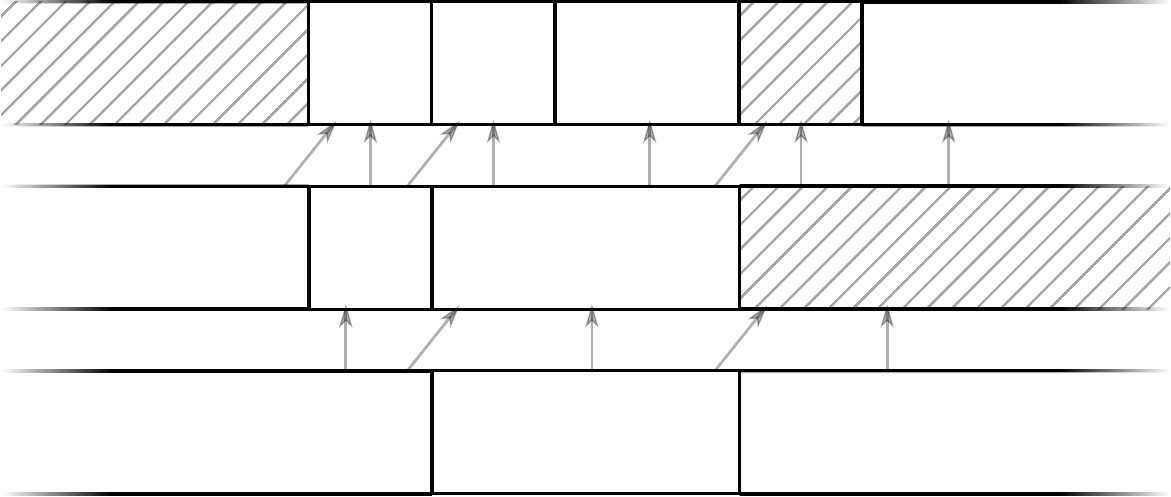}

}\hfill{}\subfloat[\label{subfig:FLDeadline_AfterColorCreation}After assigning colors
to root charging nodes]{\includegraphics[width=0.45\columnwidth]{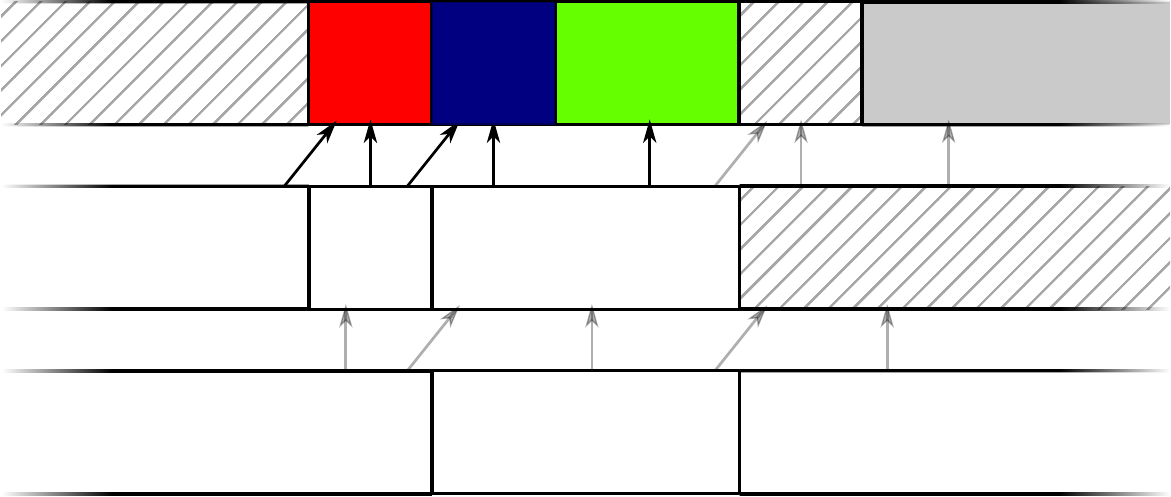}}

\subfloat[\label{subfig:FLDeadline_FirstLevelPropagation}After propagating
color to depth 1]{\includegraphics[width=0.45\columnwidth]{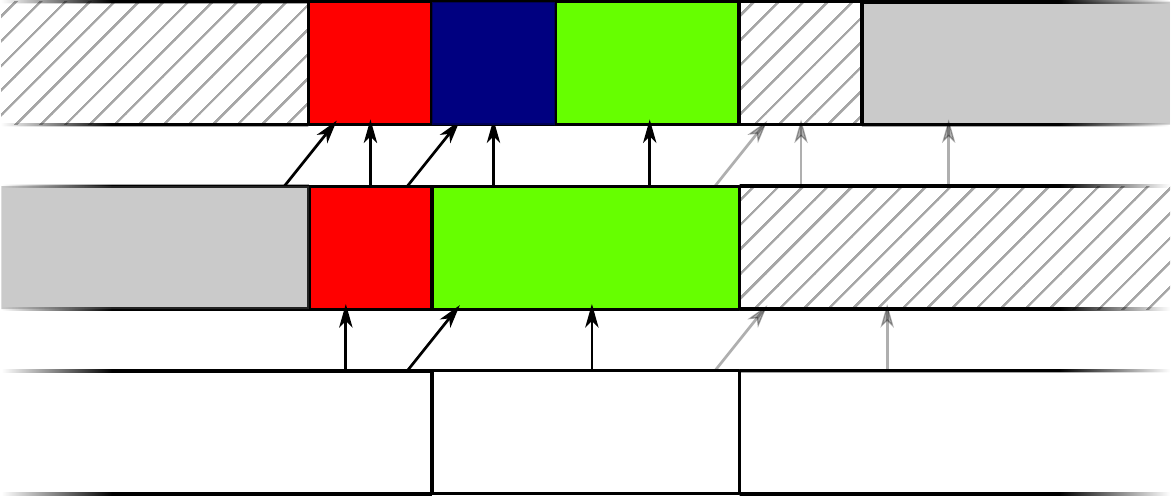}

}\hfill{}\subfloat[\label{subfig:FLDeadline_FinalPreflowState}Final state after preflow
construction]{\includegraphics[width=0.45\columnwidth]{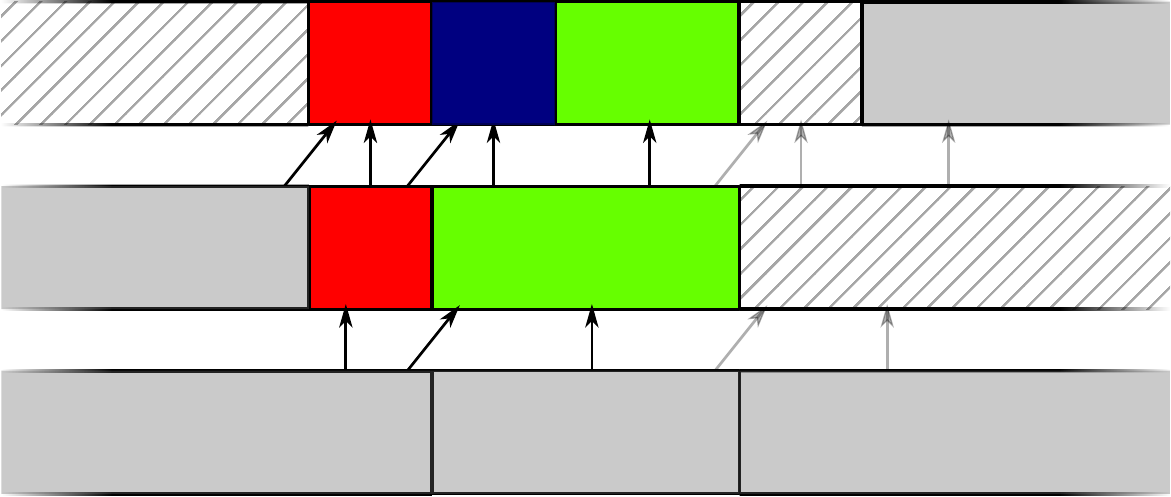}}~\\
\\

In this figure, we see the stages of the preflow construction in Procedure
\ref{proc:FLDeadline_PreflowBuilder}, visualized on a set of charging
nodes corresponding to a branch in the tree. Sub-Figure \ref{subfig:FLDeadline_BeforeColoring}
shows the state after the initialization, where the $\Special$ color
has been assigned. The node with the $\Special$ color appears as
striped. The gray edges are the edges of $\bar{E}$, not yet added
to the edge set $E$. 

Sub-Figure \ref{subfig:FLDeadline_AfterColorCreation} shows the state
after the creation of the colors at the root charging nodes. Note
that when a node is colored, the edges of $\bar{E}$ incoming to that
node are added to $E$. In the figure, nodes with the $\None$ color
for which the procedure will not call $\SetColor$ again (i.e. $\None$
is their final color) are colored gray. Note that the procedure does
\emph{not} call $\SetColor$ for the final root node (the interval
of which ends at $\infty$), and thus its color remains $\None$.

Sub-Figure \ref{subfig:FLDeadline_FirstLevelPropagation} shows the
state after the propagation of colors to the nodes at depth $1$.
Sub-Figure \ref{subfig:FLDeadline_FinalPreflowState} shows the final
state of the preflow, after coloring the nodes at the maximum depth.
Note that the color of a node $\mu$ at the maximum depth is either
$\Special$ or $\None$, as $\lambda_{\mu}=\infty$.

\caption{Visualization of preflow construction}
\end{figure}
In the analysis of the preflow $Z=(G,s,\alpha)$ resulting from this
procedure, we refer to the values of the variables used in the procedure
in their final state.
\begin{prop}
For every charging node $\mu=(u,[\tau_{1},\tau_{2}])\in M$, we have
that $\sum_{\sigma\in E_{\mu}^{-}}\alpha(\sigma)\le f$.
\end{prop}

\begin{proof}
Observe that $E_{\mu}^{-}$ is a subset of $\bar{E}$. In $\bar{E}$,
the sum of $\alpha(\sigma)$ over edges $\sigma$ outgoing from $\mu$
is exactly the amount $\mu$ invested in other charging nodes, which
is at most $f$.
\end{proof}
\begin{cor}
\label{cor:FLDeadline_FIsEnough}For a charging node $\mu\in M$ in
which $\sum_{\sigma\in E_{\mu}^{+}}\alpha(\sigma)\ge f$ we have $\chi_{\mu}\ge0$.
\end{cor}

\begin{prop}
\label{prop:FLDeadline_NotBought}Let $\mu=(u,[\tau_{1},\tau_{2}])$
be such that $\texttt{Color}[\mu]=\mu^{\star}$ for some charging
node $\mu^{\star}=(r,[\tau_{1}^{\star},\tau_{2}^{\star}])$. Then
$\opt$ did not open a facility in $T_{u}$ during $[\tau_{1},\tau_{2}^{\star}]$.
\end{prop}

\begin{proof}
Since $\texttt{Color}[\mu]\neq\Special$, we have that $\opt$ did
not open a facility in $T_{u}$ during $[\tau_{1},\tau_{2}]$.

The proof is by induction on the depth of $u$. If $u=r$, then it
must be that $\mu=\mu^{\star}$, completing the proof. Otherwise,
observe the node $\mu^{\prime}=(p(u),[\tau_{1}^{\prime},\tau_{2}^{\prime}])$
from which $\mu$ inherited its color. By the induction hypothesis,
$\opt$ did not open a facility in $T_{p(u)}$ during $[\tau_{1}^{\prime},\tau_{2}^{\star}]$.
Since there exists an edge from $\mu$ to $\mu^{\prime}$, we must
have that $\tau_{1}^{\prime}\in[\tau_{1},\tau_{2}]$, which completes
the proof.
\end{proof}
\begin{lem}
\label{lem:FLDeadline_ValidPreflow}$Z=(G,s,\alpha)$ is a preflow.
That is, for every charging node $\mu=(u,[\tau_{1},\tau_{2}])\in M$
we have $\chi_{\mu}\ge0$.
\end{lem}

\begin{proof}
We observe the following cases according to the final values of the
variables in the graph construction procedure.

\textbf{Case 1: }$\texttt{\ensuremath{\Color}}[\mu]=\Special$. In
this case, $\opt$ opened a facility in $T_{u}$ during $[\tau_{1},\tau_{2}]$,
implying that $c(\mu)\ge f$. In the initialization of Procedure \ref{proc:FLDeadline_PreflowBuilder},
an edge $\sigma$ from $s$ to $\mu$ with $\alpha(\sigma)=c(\mu)\ge f$
is created, and thus Corollary \ref{cor:FLDeadline_FIsEnough} implies
that $\chi_{\mu}\ge0$.

\textbf{Case 2: }$\texttt{Color}[\mu]=\mu^{\star}$ for some charging
node $\mu^{\star}$. In this case, incoming edges to $\mu$ were added,
with a total weight which is the total amount invested by $\mu$.
Since $\texttt{Color}[\mu]=\mu^{\star}$ was set by $\SetColor$,
we must have that $\tau_{1}\neq-\infty$ and that $\lambda_{\mu}<\infty$.
Using Observation \ref{obs:FLDeadline_RaisesCounterAtMostF}, we have
that $\Explore_{\tau_{1}}(u)$ raised counters by a total of exactly
$f$. Corollary \ref{cor:FLDeadline_FIsEnough} therefore proves the
lemma for this case.

\textbf{Case 3:} $\texttt{Color}[\mu]=\None$. Observe any edge $\sigma\in E_{\mu}^{-}$,
incoming to some node $\mu^{\prime}=(u^{\prime},[\tau_{1}^{\prime},\tau_{2}^{\prime}])$.
It must be that $\texttt{\texttt{Color}}[\mu^{\prime}]=\mu^{\star}$,
for some charging node $\mu^{\star}=(r,[\tau_{1}^{\star},\tau_{2}^{\star}])$.
Note that $\mu^{\prime}$ invested in $\mu$, and thus $\tau_{1}^{\prime}\in[\tau_{1},\tau_{2}]$.
Combining Proposition \ref{prop:FLDeadline_NotBought} for $\mu^{\prime}$
and the fact that $\Color[\mu]\neq\Special$, we have that $\opt$
did not open a facility in $T_{u}$ during $[\tau_{1},\tau_{2}^{\star}]$. 

We therefore have that for every request $q$ such that $v_{q}\in T_{u}$
and $[r_{q},d_{q}]\subseteq[\tau_{1},\tau_{2}^{\star}]$, $\opt$
must connect $q$ to some facility at a node $v\notin T_{u}$. If
it also holds that $r_{q}\le\tau_{2}$, then $\opt$ incurs a connection
cost of $\delta(v_{q},u^{\prime})$ in $\mu$ on connecting $q$.
We proceed to find such requests $q$. 

Now observe that $\mu^{\prime}$ invested $\alpha(\sigma)$ in $\mu$.
Thus, there exists a set of requests $L_{\sigma}$ that are considered
in $\Explore_{\tau_{1}^{\prime}}(u^{\prime})$ such that $\alpha(\sigma)\le\sum_{q\in L_{\sigma}}\delta(a_{q},u^{\prime})$
and $a_{q}\in T_{u}$ for every $q\in L_{\sigma}$. Since the requests
of $L_{\sigma}$ are considered in $\Explore_{\tau_{1}^{\prime}}(u^{\prime})$,
we have that $\lambda_{\mu^{\prime}}\ge d_{q}$ for every $q\in L_{\sigma}$.
Since $\texttt{Color}[\mu^{\prime}]=\mu^{\star}$, we have that $\lambda_{\mu^{\prime}}\le\tau_{2}^{\star}$,
and thus $d_{q}\le\tau_{2}^{\star}$.

Observe any $q\in L_{\sigma}$. It holds that $r_{q}\le\tau_{1}^{\prime}\le\tau_{2}$,
since $q$ is considered in $\Explore_{\tau_{1}^{\prime}}(u^{\prime})$.
Now, observe that $\texttt{Color}[\mu]=\None$ even though $\texttt{Color}[\mu^{\prime}]=\mu^{\star}$.
Hence, either $\tau_{1}=-\infty$ or $\lambda_{\mu}>\tau_{2}^{\star}$.
If $\tau_{1}=-\infty$, then $r_{q}\ge\tau_{1}$. Otherwise, $\tau_{1}\neq-\infty$
and $\lambda_{\mu}>\tau_{2}^{\star}$. Since $d_{q}\le\tau_{2}^{\star}$,
it must be that $q$ was not pending immediately after the return
of $\Explore_{\tau_{1}}(u)$. However, $\Explore_{\tau_{1}^{\prime}}(u^{\prime})$
considered $q$ when raising $c_{u}$ toward $\Explore_{\tau_{2}}(u)$.
Thus, $q$ was released after $\tau_{1}$.

Overall, for every $q\in L_{\sigma}$ we have that $r_{q}\in[\tau_{1},\tau_{2}]$
and $d_{q}\le\tau_{2}^{\star}$. Thus, $\opt$ incurs a connection
cost of at least $\sum_{q\in L_{\sigma}}\delta(v_{q},u^{(1)})$ in
the charging node $\mu$ due to $L_{\sigma}$, which is at least $\alpha(\sigma)$.

Now, if for every distinct $\sigma_{1},\sigma_{2}\in E_{\mu}^{-}$
we have that $L_{\sigma_{1}}\cap L_{\sigma_{2}}=\emptyset$, then
the connection cost $\opt$ incurs in $\mu$ is at least $\sum_{\sigma\in E_{\mu}^{-}}\alpha(\sigma)$.
Indeed, observe that $\sigma_{1}$ and $\sigma_{2}$ enter two distinct
charging nodes $\mu^{(1)}=(p(u),[\tau_{1}^{(1)},\tau_{2}^{(1)}])$
and $\mu^{(2)}=(p(u),[\tau_{1}^{(2)},\tau_{2}^{(2)}])$. Lemma \ref{obs:FLDeadline_SingleService}
implies that $\tau_{1}^{(1)}\neq\tau_{1}^{(2)}$. It is enough to
observe that for $b\in\{1,2\}$, each request $q\in L_{\sigma_{b}}$
is pending before $\tau_{1}^{(b)}$ and is served after $\tau_{1}^{(b)}$. 

Overall, we have that $c(\mu)\ge\sum_{\sigma\in E_{\mu}^{-}}\alpha(\sigma)$.
In the initialization of Procedure \ref{proc:FLDeadline_PreflowBuilder},
an edge $\sigma$ from $s$ to $\mu$ with $\alpha(\sigma)=c(\mu)$
is created, and thus $\chi_{\mu}\ge0$ as required. This concludes
the proof of the current case, and the lemma.
\end{proof}
\begin{figure}[tb]
\subfloat[Case 1: $\protect\Color=\protect\Special$]{\includegraphics[width=0.3\textwidth]{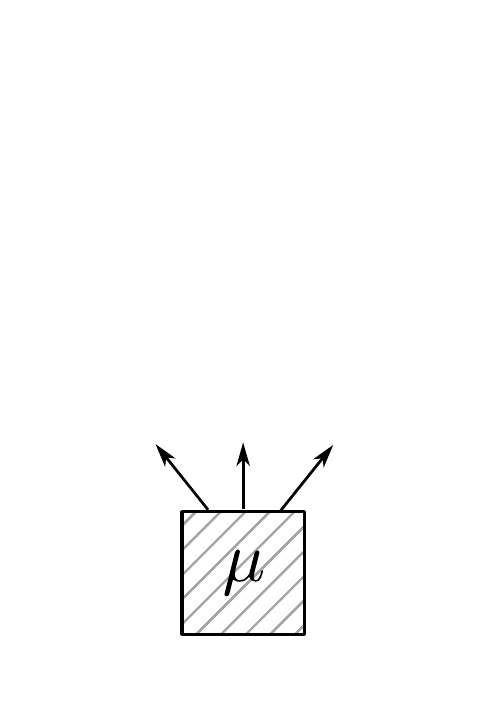}

}\hfill{}\subfloat[Case 2: $\protect\Color=\mu^{\star}$]{\includegraphics[width=0.3\textwidth]{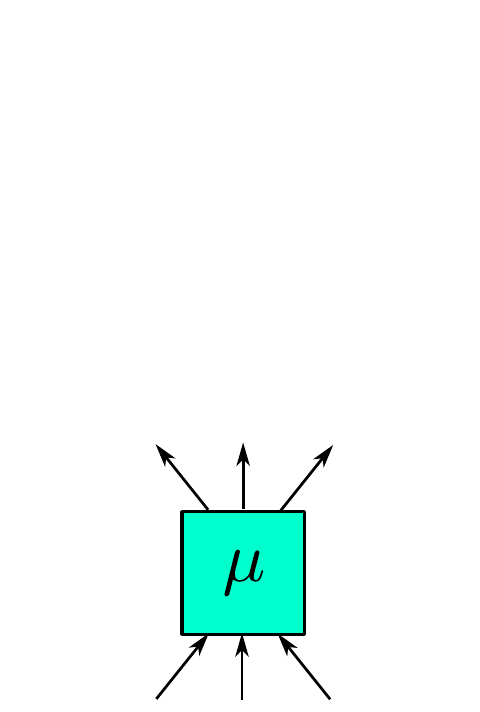}

}\hfill{}\subfloat[Case 3: $\protect\Color=\protect\None$]{\includegraphics[width=0.3\textwidth]{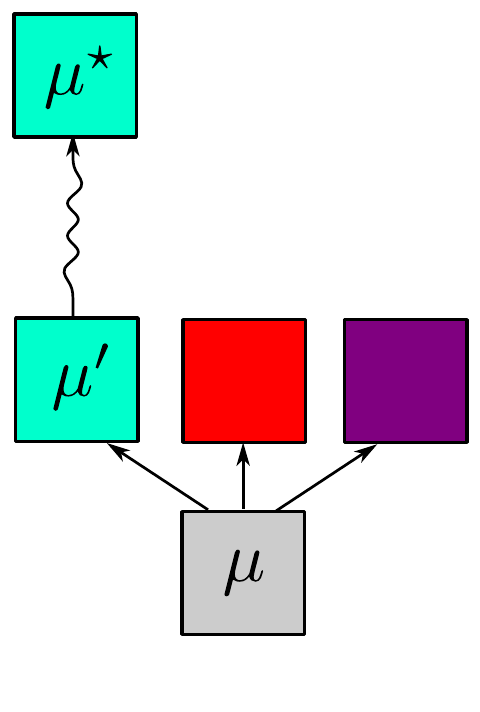}

}

\caption{Cases of Lemma \ref{lem:FLDeadline_ValidPreflow}}
\end{figure}

\begin{lem}
\label{lem:FLDeadline_RootExcesses}For each $i\in[k]$ and charging
node $\mu=(r,[t_{i-1},t_{i}])$, we have $\chi_{\mu}\ge f$.
\end{lem}

\begin{proof}
Observe that $E_{\mu}^{-}=\emptyset$. It remains to see that $\sum_{\sigma\in E_{\mu}^{+}}\alpha(\sigma)\ge f$.

If $\Color[\mu]\neq\None$, it holds that $\sum_{\sigma\in E_{\mu}^{+}}\alpha(\sigma)\ge f$
identically to Cases 1 and 2 in Lemma \ref{lem:FLDeadline_ValidPreflow}.

Otherwise, we must have that $\SetColor(\mu,\mu)$ returned $\None$.
Thus, it must be that either $t_{i-1}=-\infty$ or $\lambda_{\mu}>t_{i}$.
We claim that either of these cases contradicts $\Color[\mu]\neq\Special$.
Indeed, observe that $\Explore_{t_{i}}(r)$ is called upon a deadline
of a request $q$. If $t_{i-1}=-\infty$, it holds that $r_{q}\ge t_{i-1}$.
If $\lambda_{\mu}>t_{i}$, it must be that $r_{q}\ge t_{i-1}$ as
well. Overall, $[r_{q},d_{q}]\in[t_{i-1},t_{i}]$, and thus $\opt$
must open a facility in that interval, in contradiction.
\end{proof}
\begin{lem}
\label{lem:FLDeadline_OPT}$kf\le2(D+1)\cdot\opt^{B}+4\cdot\opt^{C}$
\end{lem}

\begin{proof}
Lemma \ref{lem:FLDeadline_ValidPreflow} yields that $Z$ is a valid
preflow. For $i\in[k]$, let $\mu_{i}=(r,[t_{i-1},t_{i}])$. Using
Lemma \ref{lem:FLDeadline_RootExcesses} and Proposition \ref{prop:Flow_SubsetLBSource},
we have that 
\[
kf\le\sum_{i=1}^{k}\chi_{\mu_{i}}\le\omega_{Z}
\]

Now observe that $E_{s}^{+}=\emptyset$, and that $\sum_{\sigma\in E_{s}^{-}}\alpha(\sigma)=\sum_{\mu\in M}c(\mu)$.
Using Lemma \ref{lem:FLDeadline_ChargeLBOpt}, we obtain
\[
kf\le\omega_{Z}=\sum_{\sigma\in E_{s}^{-}}\alpha(\sigma)=\sum_{\mu\in M}c(\mu)\le2(D+1)\cdot\opt^{B}+4\cdot\opt^{C}
\]

as required.
\end{proof}
\begin{proof}[Proof of Theorem \ref{thm:FLDeadline_HSTTheorem}]
Combining Lemmas \ref{lem:FLDeadline_ALG} and \ref{lem:FLDeadline_OPT},
we have that 
\[
\alg\le Dkf\le O(D^{2})\cdot\opt^{B}+O(D)\cdot\opt^{C}
\]
\end{proof}
\begin{rem}
\label{rem:FLDeadline_WorksForNondecreasing}Our algorithm and its
analysis also work in the case that the cost of opening a facility
is different between nodes in the tree, as long as the cost of opening
a facility at a node is at least the cost of opening a facility at
its parent node. If this is not the case, observe that the analysis
of Case 1 of Lemma \ref{lem:FLDeadline_ValidPreflow} would no longer
hold.
\end{rem}

\subsection{From HST to General Metric Space}

In this subsection, we show how the deterministic Algorithm \ref{alg:FLDeadline}
for $\left(\ge2\right)$-HSTs yields a randomized $O(\log^{2}n)$-competitive
algorithm for facility location with deadlines on a general metric
space with $n$ points, thus proving Theorem \ref{thm:FLDeadline_GMSLogSquared}.
To do this, we consider a standard probabilistic embedding of a metric
space to an HST.
\begin{thm}
\label{thm:HSTEmbedding}For any metric space $\mathcal{X}=(X,\delta_{\mathcal{X}})$
such that $|\mathcal{X}|=n$, there exists a distribution $\mathcal{D}$
over $\left(\ge2\right)$-HSTs of depth $O(\log n)$ such that $X$
are the leaves of the HST, such that the expected distortion is $O(\log n)$.
That is, for every $x_{1},x_{2}\in X$ we have that 
\[
\delta_{\mathcal{X}}(x_{1},x_{2})\le\E_{\mathcal{T}\sim\mathcal{D}}\left[\delta_{\mathcal{T}}(x_{1},x_{2})\right]\le O(\log n)
\]
where $\delta_{\mathcal{T}}$ is the distance in the tree $\mathcal{T}$.
\end{thm}

Theorem \ref{thm:HSTEmbedding} is a direct result of composing the
embeddings of Fakcharoenphol et al. \cite{DBLP:journals/jcss/FakcharoenpholRT04}
and Bansal et al. \cite{DBLP:journals/jacm/BansalBMN15}. 

We observe the following randomized algorithm for facility location
with delay on a general metric space:
\begin{enumerate}
\item Embed the metric space to a $\left(\ge2\right)$-HST according to
the distribution in Theorem \ref{thm:HSTEmbedding}.
\item Run Algorithm \ref{alg:FLDeadline} on the resulting $\left(\ge2\right)$-HST.
\end{enumerate}
\begin{proof}[Proof of Theorem \ref{thm:FLDeadline_GMSLogSquared}]
We show that the randomized algorithm described above is indeed $O(\log^{2}n)$-competitive.
Fix any instance of facility location with deadlines. We denote by
$\alg^{\T}$ the cost of the algorithm on the instance with regard
to distances on the chosen $\left(\ge2\right)$-HST $\T$. Since $\delta_{\T}(x_{1},x_{2})\ge\delta_{\mathcal{X}}(x_{1},x_{2})$,
we have that $\alg^{\mathcal{X}}\le\alg^{\T}$, where $\alg^{\mathcal{X}}$
is the actual cost incurred by the algorithm on this instance.

From Theorem \ref{thm:FLDeadline_HSTTheorem}, we know that for any
solution $\opt^{\T}$ for the instance on $\T$, it holds that $\alg^{\T}\le O(D^{2})\cdot\opt^{\T,B}+O(D)\cdot\opt^{\T,C}$
, where $D$ is the depth of $\T$ (and thus $D=O(\log n)$).

Now, denote by $\opt^{\mathcal{X}}$ the optimal solution for the
instance over $\mathcal{X}$. Observe that for every $\T$ in the
support of $\mathcal{D}$, $\opt$ yields a solution $\opt^{\T}$
by opening facilities at the same locations, at the same times, and
connecting the same requests. It holds that $\opt^{\T,B}=\opt^{\mathcal{X},B}$,
and that $\E_{\T\sim\mathcal{D}}\left[\opt^{\T,C}\right]\le O(\log n)\cdot\opt^{\mathcal{X},C}$.

Combining the above facts, we have that 
\begin{align*}
\E\left[\alg^{\mathcal{X}}\right] & \le\E_{\T\sim\mathcal{D}}\left[\alg^{\T}\right]\le\E_{\T\sim\mathcal{D}}\left[O(\log^{2}n)\cdot\opt^{\T,B}+O(\log n)\cdot\opt^{\T,C}\right]\\
 & \le O(\log^{2}n)\cdot\opt^{\mathcal{X},B}+O(\log^{2}n)\cdot\opt^{\mathcal{X},C}=O(\log^{2}n)\cdot\opt^{\mathcal{X}}
\end{align*}
proving the theorem.
\end{proof}
The reasoning behind the main theorem of this subsection is that the
connection cost is distorted upon HST embedding, while the buying
cost is not. Thus, the HST algorithm is allowed to lose a larger factor
over $\opt^{B}$ ($\Theta(\log^{2}n)$) compared to the factor it
loses over $\opt^{C}$ ($\Theta(\log n)$). This property is used
to analyze the other problems considered in this paper in a similar
manner.
\begin{rem}
For the case of different costs for opening facilities at different
points in the metric space, we obtain a $O(\log^{2}\Delta+\log\Delta\log n)$-competitive
randomized algorithm, with $\Delta$ the aspect ratio of the metric
space, in the following manner. We use the embedding of Fakcharoenphol
et al. \cite{DBLP:journals/jcss/FakcharoenpholRT04} without composing
it with the embedidng of \cite{DBLP:journals/jacm/BansalBMN15}. This
yields a 2-HST (rather than a $\left(\ge2\right)$-HST), which has
a depth of $O(\log\Delta)$, with $\Delta$ the aspect ratio of the
original metric space. This tree has an distortion of $O(\log n)$. 

In this $2$-HST, for each node $u$, the distances between $u$ and
the leaves in $T_{u}$ are equal. Thus, we can allow the algorithm
to open a facility at $u$, at a cost which is the minimal cost of
opening a facility at a leaf in $T_{u}$, at a loss of a factor of
$2$ in connection cost. The resulting tree has non-decreasing opening
costs from the root to any leaf, and is (in particular) a $\left(\ge2\right)$-HST
of depth $D=O(\log\Delta)$. Thus, using the algorithm for $\left(\ge2\right)$-HSTs
and Remark \ref{rem:FLDeadline_WorksForNondecreasing}, and applying
the distortion of $O(\log n)$ to the connection cost as in the proof
of Theorem \ref{thm:FLDelay_GMSLogSquared}, we obtain the $O(\log^{2}\Delta+\log\Delta\log n)$-competitive
algorithm.
\end{rem}

\section{\label{sec:FLDelay}Facility Location with Delay}

\subsection{Problem and Notation}

We now describe the facility location with delay problem. The problem
is an extension of the facility location with deadlines problem, in
which the deadline for each request $q$ is replaced with an arbitrary
delay function $d_{q}(t)$ associated with that request. Each delay
function is required to be continuous and monotonically non-decreasing.
This is indeed an extension of the deadline problem, as a deadline
can be described as a step function, which goes from $0$ to infinity
at the time of the deadline. Such a step function can be approximated
arbitrarily well by a continuous delay function.

A feasible solution for a facility location with delay instance consists
of opening facilities and connecting each request to some facility,
as in the deadline case. In addition to the opening costs and connection
costs incurred, the solution also pays $d_{q}(t)$ for each request
$q$ connected at time $t$. Overall, for an instance of the problem
with requests $Q$, the algorithm incurs the delay cost $\alg^{D}=\sum_{q\in Q}d_{q}(t_{q})$,
where $t_{q}$ is the time in which $q$ is served by the algorithm.
Thus, the algorithm's goal is to minimize the total cost
\[
\alg=\alg^{B}+\alg^{C}+\alg^{D}
\]

Without loss of generality, we assume that $d_{q}(r_{q})=0$. Indeed,
if this is not the case, observe that any solution (including the
optimal one) must pay this initial amount of $d_{q}(r_{q})$ in delay
for that request, which only reduces the competitive ratio of any
online algorithm.

In this section, we prove the following theorem.
\begin{thm}
\label{thm:FLDelay_GMSLogSquared}There exists an $O(\log^{2}n)$-competitive
randomized algorithm for facility location with delay for a general
metric space of size $n$.
\end{thm}

\subsection{Algorithm for HSTs}

In this subsection, we present a deterministic algorithm for facility
location with delay on a $\left(\ge2\right)$-HST. This algorithm
yields a randomized $O(\log^{2}n)$-competitive algorithm for general
metric spaces, in a similar way to the deadline case.

We require the following definitions.
\begin{defn}[Solution]
 Let $Q$ be a set of requests. For $S\subseteq X$, and a function
$\phi:Q\rightarrow S$ we say that $(S,\phi)$ is a \emph{solution
for }$Q$, with a cost $|S|\cdot f+\sum_{q\in Q}\delta(v_{q},\phi(q))$.
If $S\subseteq T_{u}$ for some node $u$, we write that $(S,\phi)$
is a \emph{solution for $Q$ under $u$.}
\end{defn}

\begin{defn}[Ancestor-closed solution]
 Let $Q$ be a set of requests, and let $(S,\phi)$ be a solution
for $Q$ under a node $u$. We say that $(S,\phi)$ is an \emph{ancestor-closed
solution for $Q$ under $u$ }if for every $s\in S$ such that $s\neq u$
we have that $p(s)\in S$.

If $u=r$, we simply write that $(S,\phi)$ is an ancestor-closed
solution for $Q$.
\end{defn}

\begin{defn}[$\psi(Q)$ and $\psi_{u}(Q)$]
We define $\psi(Q)$ to be the cost of the minimal-cost ancestor-closed
solution for $Q$. Similarly, we define $\psi_{u}(Q)$ to be the cost
of the minimal-cost ancestor-closed solution for $Q$ under $u$.
\end{defn}

\begin{defn}[Critical request set]
\label{def:FLDelay_CriticalSet}We say that a set of pending requests
$Q$ at time $t$ is \emph{critical }if $d_{Q}(t)\ge\psi(Q)$.
\end{defn}

\textbf{Algorithm's description. }Our algorithm is given in Algorithm
\ref{alg:FLDelay}. The algorithm calls $\texttt{\ensuremath{\UponCritical}}$
whenever a set of pending requests $Q$ becomes critical. It uses
the $\Open$ and $\Invest$ functions from Algorithm \ref{alg:FLDeadline},
but redefines the $\Explore$ function. The function call $\Explore(u)$
now forwards time until the first occurrence of one of two events. 

The first event is a pending request $q$ in $T_{u}$, the delay of
which exceeds the cost of connecting it to $u$. Handling this event
is similar to handling the event considered in Algorithm \ref{alg:FLDeadline}
-- we attempt to raise the counter of the child node $v$ in $q$'s
direction by $\delta(u,v_{q})$. If this fills the counter of $v$,
this triggers an immediate call to $\Explore(v)$. However, in contrast
to the deadline case, calling $\Explore(v)$ is not guaranteed to
connect the request $q$. For this reason, $\Explore(u)$ must invest
the remainder of $\delta(u,v_{q})$ (or until $b_{u}=0$) in $v$'s
counter before connecting $q$ to $u$.

The second event is not analogous to the deadline case. In this event,
for a child $v$ of $u$ and a ``coalition'' of pending requests
$Q$ in $T_{v}$, we have that the delay of $Q$ exceeds $\psi_{v}(Q)$.
In this case, the algorithm invests in $v$ until either it is out
of budget ($b_{u}=0$) or $v$'s counter is full $(c_{v}=f$). It
is important to note that in contrast to the first event, the fact
that $\Explore(u)$ considered $Q$ does not provide any guarantees
for connecting the requests of $Q$. 

Algorithm \ref{alg:FLDelay} changes $\Explore(u)$ so that time is
forwarded until one of two events happens, rather than the single
event in Algorithm \ref{alg:FLDeadline} (i.e. expired deadline).
These two events are shown in Figure \ref{fig:FLDelay_TimeForwarding}.

\noindent \LinesNumbered \RestyleAlgo{boxruled}\renewcommand{\algorithmcfname}{Algorithm}
\begin{algorithm}[tb]
\caption{\label{alg:FLDelay}Facility Location with Delay}

\SetKwProg{Fn}{Function}{}{end}
\SetKwProg{EFn}{Event Function}{}{end}
\SetKwFunction{UponCritical}{UponCritical}
\SetKwFunction{Open}{Open}
\SetKwFunction{Explore}{Explore}

\textbf{Initialization.}

Initialize $c_{u}\leftarrow0$ for any node $u\in T\backslash\{r\}$.

Declare $b_{u}$ for any node $u\in T$.

\;

\EFn(\tcp*[h]{Upon request set becoming critical as per Definition
\ref{def:FLDelay_CriticalSet}}){\UponCritical{}}{

\Explore{$r$}

}

\;

\Fn{\Explore{$u$}}{

\Open{$u$}

set $b_{u}\leftarrow f$

\While{$b_{u}\neq0$ \normalfont{\textbf{and}} there remain pending
requests in $T_{u}$}{

let $Q$ be the set of pending requests in $T_{u}$.

let $t_{1}^{\prime}\ge t$ be the earliest time in which a request
$q\in Q$ satisfies $d_{q}(t_{1}^{\prime})\ge\delta(v_{q},u)$.

let $t_{2}^{\prime}\ge t$ be the earliest time in which there exists
$v$ such that $p(v)=u$, and a set of requests $Q^{\prime}\subseteq Q\cap T_{v}$
such that $d_{Q^{\prime}}(t_{2}^{\prime})\ge\psi_{v}(Q^{\prime})$.

\eIf{$t_{1}^{\prime}\le t_{2}^{\prime}$}{

let $q\in Q$ be the request in the definition of $t_{1}^{\prime}$.
let $v$ be the child of $u$ on the path to $v_{q}$.

call \Invest{$u$,$v$,$\delta(u,v_{q})$}, and let $y$ be the
return value.

\lIf{$c_{v}=f$}{set $c_{v}\leftarrow0$ $\texttt{\textbf{;}}$
call \Explore{$v$}.}

\If{ $q$ is still pending }{

call \Invest{$u$,$v$,$\delta(u,v_{q})-y$}

connect $q$ to facility at $u$

}

} { 

let $v$ be as in the definition of $t_{2}^{\prime}$. 

call \Invest{$u$,$v$,$\infty$}.

\lIf{$c_{v}=f$}{set $c_{v}\leftarrow0$ $\texttt{\textbf{;}}$
call \Explore{$v$}.}

}

}

}

\;

\lFn(\tcp*[h]{As in Algorithm \ref{alg:FLDeadline}}){\Open{$u$}}{}

\lFn(\tcp*[h]{As in Algorithm \ref{alg:FLDeadline}}){\Invest{$u$,$v$,$x$}}{}
\end{algorithm}
\begin{figure}[h]
\subfloat[Event 1: single request's delay exceeds connection]{\includegraphics[scale=0.7]{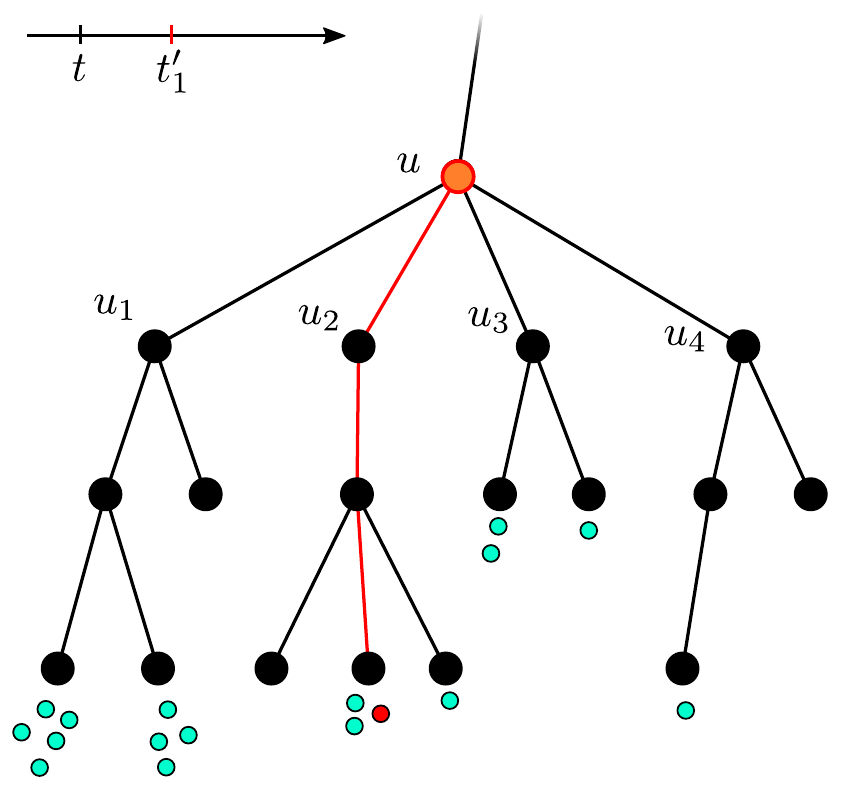}

}\hfill{}\subfloat[Event 2: coalition of low-delay requests exceeds ]{\includegraphics[scale=0.7]{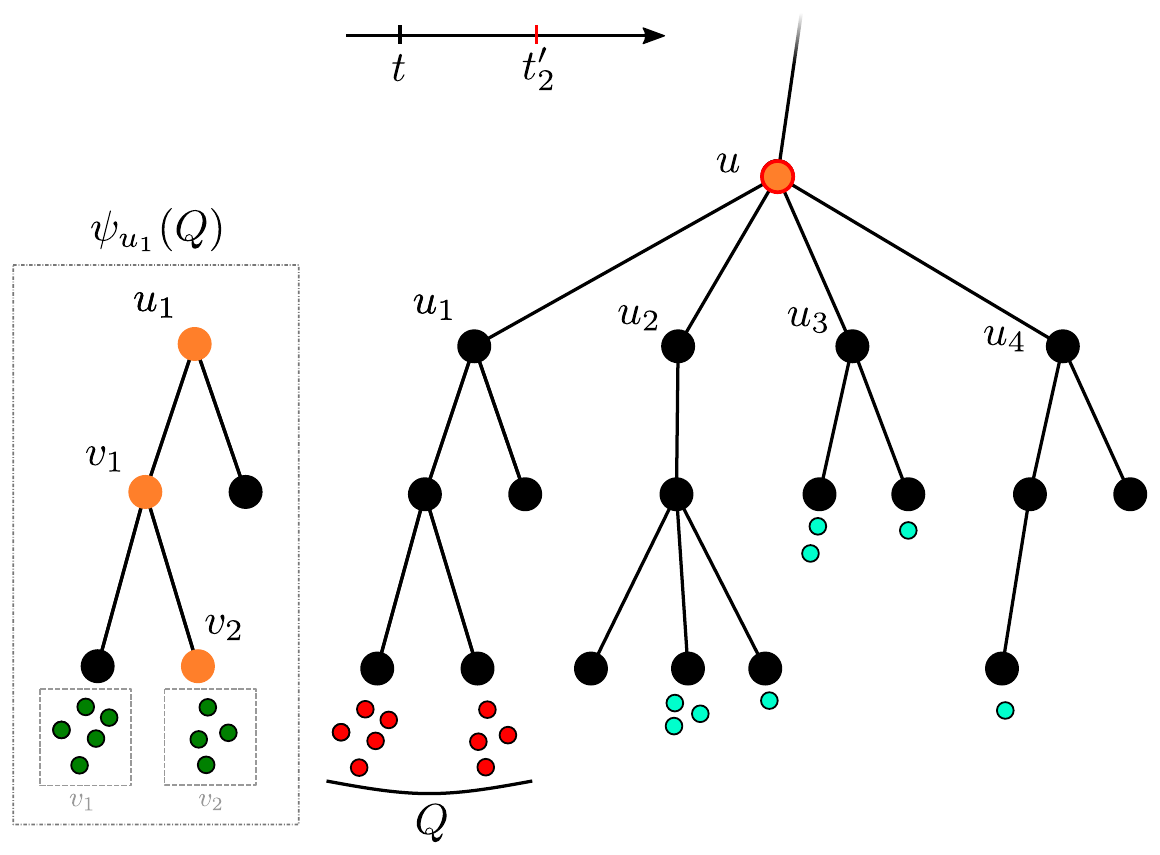}

}\\
\\

This figure shows the two events considered in the time-forwarding
stage of $\Explore(u)$. The first event is a request whose delay
exceeds the cost of connecting it to $u$. This event is very similar
to the single event in the deadline case. 

The second event is a ``coalition'' $Q$ of low-delay requests under
some child node of $u$, denoted $u_{1}$. Though the delay of no
single request in the coalition exceeds the cost of its connection
to $u$, the total delay of the coalition exceeds $\psi_{u_{1}}(Q)$
-- the cost of an optimal ancestor-closed solution under $u_{1}$
for that coalition (the optimal solution is shown in the dash-dotted
rectangle).

\caption{\label{fig:FLDelay_TimeForwarding}Time forwarding conditions of Algorithm
\ref{alg:FLDelay}}
\end{figure}

\subsection{Analysis}

Fix any instance of facility location with delay on a $\left(\ge2\right)$-HST. 
\begin{thm}
\label{thm:FLDelay_HSTTheorem}$\alg\le O(D^{2})\cdot\opt^{B}+O(D)\cdot\opt^{C}+O(D^{2})\cdot\opt^{D}$.
\end{thm}

Observe that the connection cost is distorted by the embedding, while
the buying and delay costs are not. Thus, using an identical argument
to the proof of Theorem \ref{thm:FLDeadline_GMSLogSquared} of the
deadline problem, Theorem \ref{thm:FLDelay_HSTTheorem} implies Theorem
\ref{thm:FLDelay_GMSLogSquared}.

We devote this subsection to prove Theorem \ref{thm:FLDelay_HSTTheorem}.

\subsubsection{Upper Bounding $\protect\alg$}

To upper bound the cost of the algorithm, we show the following Lemma. 
\begin{lem}
\label{lem:FLDelay_ALG}$\alg\le6\cdot(D+1)\cdot kf$
\end{lem}

\begin{prop}
\label{prop:FLDelay_DelayBoundedByBuyingAndConnection}$\alg^{D}\le\alg^{B}+\alg^{C}$
\end{prop}

\begin{proof}
The algorithm explicitly maintains that for every set of pending requests
$Q$ at any time $t$ we have that $\psi(Q)\ge d_{Q}(t)$. Now, consider
that since the delay of a pending request goes to infinity, the algorithm
ultimately serves every request. Consider a specific service made
by the algorithm, described by a solution $(S,\phi)$ to some set
of requests $Q$, and note that $(S,\phi)$ is an ancestor-closed
solution to $Q$. Thus, its total cost is at least $\psi(Q)$, completing
the proof.
\end{proof}
\begin{lem}
\label{lem:FLDelay_BuyingAndConnectionBounded}$\alg^{B}+\alg^{C}\le3\cdot(D+1)\cdot kf$.
\end{lem}

Proving Lemma \ref{lem:FLDelay_BuyingAndConnectionBounded} is very
similar to proving Lemma \ref{lem:FLDeadline_ALG} of the deadline
case. Defining cumulative counters as in the deadline case, we can
prove Corollary \ref{cor:FLDeadline_CountersBoundedByDKF} holds in
the delay case using an identical proof. It remains to show and prove
analogues to Propositions \ref{prop:FLDeadline_InvestServes} and
\ref{prop:FLDeadline_ALGBoundedByCounters}. 

Note that the connection costs in $\Explore(u)$ only occur during
iterations of the main loop in which the main \textbf{if} condition
is entered.
\begin{prop}[analogue of Proposition \ref{prop:FLDeadline_InvestServes} ]
\label{prop:FLDelay_InvestServes} Suppose the function $\Explore(u)$
enters the main \textbf{if }condition in an iteration, and let $q$
be the pending request under consideration. Then at least one of the
following holds:
\begin{enumerate}
\item The sum of return values of calls to $\Invest$ in that iteration
is $\delta(u,v_{q})$.
\item $b_{u}=0$ at the end of the iteration.
\item $q$ is no longer pending after the first call to $\Invest$.
\end{enumerate}
\end{prop}

\begin{proof}
Consider the state after the return of the first call to $\Invest$.
Either $\Invest$ returned $\delta(u,v_{q})$, or $b_{u}=0$, or $c_{v}=f$.
In the first two cases, we are done. In the third case, $\Explore(u)$
then calls $\Explore(v)$. If $q$ is connected during $\Explore(v)$,
we are done. Otherwise, $\Explore(u)$ enters the nested \textbf{if
}condition upon observing that $q$ is still pending. 

Denote by $y$ the return value of the first call to $\Invest$, and
consider the return value of the second call to $\Invest$ made in
the nested \textbf{if}. If the return value is $\delta(u,v_{q})-y$,
we are done. Otherwise, consider that $c_{v}=0$ before the call to
$\Invest$, and since $\delta(u,v_{q})-y\le f$ it must thus be that
$b_{u}=0$ after the return of the second call to $\Invest$. This
concludes the proof.
\end{proof}
\begin{prop}[analogue of Proposition \ref{prop:FLDeadline_ALGBoundedByCounters}]
\label{prop:FLDelay_ALGBoundedByCounters}$\alg^{B}+\alg^{C}\le3\cdot\sum_{u\in T}\bar{c}_{u}$
\end{prop}

\begin{proof}
The costs of the algorithm (both opening and connection) are again
contained in calls to $\Explore$, as in the proof of Proposition
\ref{prop:FLDeadline_ALGBoundedByCounters}. Each call to $\Explore(u)$
has an opening cost of $f$. 

As for connection costs, note that they only occur in iterations of
the main loop in $\Explore(u)$ in which the main \textbf{if }condition
is entered, and not in the \textbf{else} condition. In each iteration
in which the main \textbf{if} condition is entered, a request $q$
is considered, which may be connected to $u$ at cost $\delta(u,v_{q})$.
Through Proposition \ref{prop:FLDelay_InvestServes}, in each such
iteration either the return values of calls to $\Invest$ sum to $\delta(u,v_{q})$
(and thus $b_{u}$ decreases by $\delta(u,v_{q})$), $b_{u}=0$ at
the end of the iteration (in which case this is the last iteration),
or $\Explore(u)$ does not connect $q$. 

Since $b_{u}$ can decrease by at most $f$, the connection cost of
the algorithm is bounded by $f+\delta(u,v_{q})$ for $q$ the last
request considered, which is at most $2f$.

Noting that the total cost of $\Explore(u)$ is at most $3f$, and
that $\bar{c}_{u}$ is raised by $f$ before calling $\Explore(u)$
yields the proposition.
\end{proof}
\begin{proof}[Proof of Lemma \ref{lem:FLDelay_BuyingAndConnectionBounded}]
Results directly from Proposition \ref{prop:FLDelay_ALGBoundedByCounters}
and Corollary \ref{cor:FLDeadline_CountersBoundedByDKF} (which holds
for the delay case as well).
\end{proof}
\begin{proof}[Proof of Lemma \ref{lem:FLDelay_ALG}]
 The lemma results directly from Proposition \ref{prop:FLDelay_DelayBoundedByBuyingAndConnection}
and Lemma \ref{lem:FLDelay_BuyingAndConnectionBounded}.
\end{proof}

\subsubsection{Lower Bounding $\protect\opt$}

To lower bound the cost of the optimum, we prove the following lemma,
which is analogous to Lemma \ref{lem:FLDeadline_OPT} of the deadline
case. 
\begin{lem}
\label{lem:FLDelay_OPT}$kf\le(D+1)\cdot\opt^{B}+2\cdot\opt^{C}+(D+1)\cdot\opt^{D}$.
\end{lem}

\paragraph*{Charging nodes and incurred costs.}

We again use charging nodes, defined as in the deadline case. However,
the charging nodes for the delay case use half-closed intervals instead
of the closed intervals of the deadline case. The reason for this
is that we do not have the guarantee that only one call to $\UponCritical$
is made at a given time, so that using closed intervals would break
the analogue of Lemma \ref{lem:FLDeadline_ChargeLBOpt} for our case.

Let $M$ be the set of charging nodes. The definitions of $c_{b}(\mu)$
(buying costs) and $c_{c}(\mu)$ (connection costs) are identical
to the definition in the deadline case. For the delay case, we also
define incurring \emph{delay }costs, 
\begin{defn}[$c_{d}(\mu)$]
Let $\mu=(u,[\tau_{1},\tau_{2}))$ be a charging node. Let $c_{d}(\mu)$
be the \emph{delay cost incurred by $\opt$ on }$\mu$, defined to
be the total delay cost incurred by $\opt$ on requests in $T_{u}$
released in $[\tau_{1},\tau_{2})$.
\end{defn}

We write $c(\mu)=c_{b}(\mu)+c_{c}(\mu)+c_{d}(\mu)$.

We use Procedure \ref{proc:FLDeadline_PreflowBuilder} to create the
preflow. However, we give a different definition to $\lambda$ than
in the deadline case. The definition follows.
\begin{defn}[$\lambda_{u}^{t}$ and $\lambda_{\mu}$]
For every function call $\Explore_{t}(u)$ for some $u\in T$ and
time $t$, let $Q$ be the set of requests pending in $T_{u}$ immediately
after the return of $\Explore(u)$. We define $\lambda_{u}^{t}$ to
be the first time $t^{\prime}\ge t$ in which one of the following
conditions occurs:
\begin{enumerate}
\item There is a request $q\in Q$ such that $d_{q}(t^{\prime})\ge\delta(v_{q},u)$.
\item There exists a set of requests $Q^{\prime}\subseteq Q$ such that
$Q^{\prime}\subseteq T_{u^{\prime}}$, for some $u^{\prime}$ a child
of $u$, and also $d_{Q^{\prime}}(t^{\prime})\ge\psi_{u^{\prime}}(Q^{\prime})$.
\end{enumerate}
Like in the deadline case, we write $\lambda_{\mu}=\lambda_{u}^{\tau_{1}}$
where $\mu=(u,[\tau_{1},\tau_{2}))$.
\end{defn}

\noindent 
\begin{lem}
\label{lem:FLDelay_ChargeLBOPT}$\sum_{\mu}c(\mu)\le(D+1)\cdot\opt^{B}+2\cdot\opt^{C}+(D+1)\cdot\opt^{D}$.
\end{lem}

\begin{proof}
$\sum_{\mu}c_{b}(\mu)$ can be charged to $(D+1)\cdot\opt^{B}$ and
$\sum_{\mu}c_{c}(\mu)$ can be charged to $2\cdot\opt^{C}$ as in
Lemma \ref{lem:FLDeadline_ChargeLBOpt} (since the intervals are not
closed, this improves by a factor of $2$). It remains to charge $\sum_{\mu}c_{d}(\mu)$
to $(D+1)\cdot\opt^{D}$. To do so, observe that the delay incurred
by $\opt$ on a request $q$ can only be counted in charging nodes
with intervals containing $r_{q}$, and defined by a node which is
an ancestor of $v_{q}$. There are at most $D+1$ such nodes.
\end{proof}
\begin{obs}
\label{obs:FLDelay_AncestorClosedOptimalConnectsToAncestor}Let $(S,\phi)$
be a minimal-cost ancestor-closed solution for $Q$ under $u$. Then
it holds for every $q\in Q$ that $\phi(q)$ the least ancestor of
$v_{q}$ in $S$. 
\end{obs}

\begin{obs}
\label{obs:FLDelay_AncestorClosedOptimalSubOptimal}Let $(S,\phi)$
be a minimal-cost ancestor-closed solution for $Q$ under $u$. Let
$u^{\prime}\in S$ be a descendant of $u$. Observing the set $Q^{\prime}=Q\cap T_{u^{\prime}}$,
we have that $(S\cap T_{u^{\prime}},\phi\restriction_{Q^{\prime}})$
is a minimal-cost ancestor-closed solution for $Q^{\prime}$ under
$u^{\prime}$.
\end{obs}

\begin{prop}[Decomposition of minimum-cost ancestor-closed solutions]
\label{prop:FLDelay_MinimumPrefixDecomposition}Let $(S,\phi)$ be
a minimum-cost ancestor-closed solution for $Q\subseteq T_{u}$ under
$u$, and let $\bar{S}\subseteq S$ be the children of $u$ in $S$.
Define $Q_{1}^{u^{\prime}}=Q\cap T_{u^{\prime}}$, and define $Q_{2}=Q\backslash\left(\bigcup_{u^{\prime}\in\bar{S}}Q_{1}^{u^{\prime}}\right)$.
Then 
\[
\psi_{u}(Q)=\left(\sum_{u^{\prime}\in\bar{S}}\psi_{u^{\prime}}(Q_{1}^{u^{\prime}})\right)+f+\sum_{q\in Q_{2}}\delta(u,v_{q})
\]
\end{prop}

\begin{proof}
For every $u^{\prime}\in\bar{S}$, Observation \ref{obs:FLDelay_AncestorClosedOptimalConnectsToAncestor}
implies that the requests of $Q_{1}^{u^{\prime}}$ only connect to
facilities in $S\cap T_{u^{\prime}}$. The opening costs of $S\cap T_{u^{\prime}}$,
plus the connection costs of $Q_{1}^{u^{\prime}}$, are exactly $\psi_{u^{\prime}}(Q_{1}^{u^{\prime}})$
according to Observation \ref{obs:FLDelay_AncestorClosedOptimalSubOptimal},
for a total of $\sum_{u^{\prime}\in\bar{S}}\psi_{u^{\prime}}(Q_{1}^{u^{\prime}})$.

In addition, opening the facility at $u$ costs $f$. Observation
\ref{obs:FLDelay_AncestorClosedOptimalConnectsToAncestor} implies
that the requests of $Q_{2}$ are connected to the facility at $u$,
at a total cost of $\sum_{q\in Q_{2}}\delta(u,v_{q})$. This finishes
the proof of the proposition.
\end{proof}
\begin{lem}
\label{lem:FLDelay_ValidPreflow}$\chi_{\mu}\ge0$ for every $\mu=(u,[\tau_{1},\tau_{2}))\in M$.
That is, the preflow $Z=(G,s,\alpha)$ defined in Procedure \ref{proc:FLDeadline_PreflowBuilder}
is valid.
\end{lem}

\begin{proof}
We observe the following cases for $\mu$.

\textbf{Case 1: }$\Color[\mu]=\Special$. This case is identical to
Case 1 in Lemma \ref{lem:FLDeadline_ValidPreflow}.

\textbf{Case 2: }$\texttt{Color}[\mu]=\mu^{\star}$ for a charging
node $\mu^{\star}$. Again, this case is similar to Case 2 in Lemma
\ref{lem:FLDeadline_ValidPreflow}.

From now on, assume we are not in the previous two cases, and thus
$\Color[\mu]=\None$. Every outgoing edge from $\mu$ to some charging
node $\mu^{\prime}=(u^{\prime},[\tau_{1}^{\prime},\tau_{2}^{\prime}))$
is created from $\mu^{\prime}$ investing in $\mu$, which means that
$\Explore_{\tau_{1}^{\prime}}(u^{\prime})$ raised the counter $c_{u}$
towards $\Explore_{\tau_{2}}(u)$. 

\textbf{Case 3: }For every such $\mu^{\prime}$, we have that $\Explore_{\tau_{1}^{\prime}}(u^{\prime})$
raised $c_{u}$ towards $\Explore_{\tau_{2}}(u)$ only through calls
to $\Invest$ inside the main \textbf{if }condition of $\Explore$,
and \emph{not} through the main \textbf{else }condition. In this case,
we show that $c_{c}(\mu)+c_{d}(\mu)\ge\sum_{\sigma\in E_{\mu}^{-}}\alpha(\sigma)$,
proving the lemma for this case. The proof is almost identical to
the proof of Case 3 of Lemma \ref{lem:FLDeadline_ValidPreflow}, in
which we showed for the deadline case that $c_{c}(\mu)\ge\sum_{\sigma\in E_{\mu}^{-}}\alpha(\sigma)$.
The argument for the deadline case consisted of finding a set of requests
which the optimum had to connect, all released in $[\tau_{1},\tau_{2})$.
The difference between our delay case and the deadline case is that
$\opt$ might choose not to connect some of those requests, in which
case it must incur a delay cost which is at least its connection cost.

\textbf{Case 4: }There exists an outgoing edge from $\mu$ to a charging
node $\mu^{\prime}=(u^{\prime},[\tau_{1}^{\prime},\tau_{2}^{\prime}))$,
such that $\Explore_{\tau_{1}^{\prime}}(u^{\prime})$ raised $c_{u}$
towards $\Explore_{\tau_{2}}(u)$ through calls to $\Invest$ inside
the main \textbf{else }condition of $\Explore$. Let $\texttt{Color}[\mu^{\prime}]=\mu^{\star}=(r,[\tau_{1}^{\star},\tau_{2}^{\star}))$.
Observing that Proposition \ref{prop:FLDeadline_NotBought} holds
for the delay problem as well, and using $\Color[\mu]\neq\Special$,
we have that $\opt$ did not open a facility in $T_{u}$ during $[\tau_{1},\tau_{2}^{\star})$. 

Since $\Explore_{\tau_{1}^{\prime}}(u^{\prime})$ raised $c_{u}$
towards $\Explore_{\tau_{2}}(u)$ inside the main \textbf{else }condition,
there was a set $Q$ of requests pending at $\tau_{1}^{\prime}$ such
that there exists a time $\hat{t}\le\lambda_{\mu^{\prime}}\le t_{i}$
in which $d_{Q}(\hat{t})\ge\psi_{u}(Q)$. In addition, the main \textbf{else
}condition is only reached if $t_{1}^{\prime}>t_{2}^{\prime}=\hat{t}$.
Thus, for every request $q\in Q$ we have that $d_{q}(\hat{t})<\delta(u^{\prime},v_{q})$.

Observe that every $q\in Q$ is pending at $\tau_{1}^{\prime}\le\tau_{2}$,
and thus released prior to $\tau_{2}$. Showing that $r_{q}\ge\tau_{1}$,
together with the fact that $\opt$ did not open a server in $T_{u}$
during $[\tau_{1},\tau_{2}^{\star})$, would yield that $\opt$ either:
\begin{itemize}
\item connected $q$ to a facility outside $T_{u}$ at a cost of at least
$\delta(v_{q},p(u)=u^{\prime})$, which is at least $d_{q}(\hat{t})$,
or
\item did not connect $q$ until time $\tau_{2}^{\star}$, in which case
it paid a delay cost of $d_{q}(\tau_{2}^{\star})\ge d_{q}(\hat{t})$.
\end{itemize}
In either case, $\opt$ paid at least $d_{q}(\hat{t})$ in delay and
connection costs on $q$. Since we have that $r_{q}\in[\tau_{1},\tau_{2})$,
we have that $\opt$ incurred a cost of $d_{q}(\hat{t})$ in $u$
due to $q$. It remains to find a set of such requests $Q^{\prime}\subseteq Q$
such that $r_{q}\ge\tau_{1}$ for every $q\in Q^{\prime}$, and such
that $d_{Q^{\prime}}(\hat{t})\ge f$.

\textbf{Claim: }there exists a set of requests $Q^{\prime}\subseteq Q$
such that $r_{q}\ge\tau_{1}$ for every $q\in Q^{\prime}$, and such
that $d_{Q^{\prime}}(\hat{t})\ge f$.

Now, since $\texttt{Color}[u]=\None$, we have that either $\tau_{1}=-\infty$
or $\lambda_{\mu}>\tau_{2}^{\star}$. If $\tau_{1}=-\infty$, then
$r_{q}\ge\tau_{1}$ for every $q\in Q$. Since $d_{Q}(\hat{t})\ge\psi_{u}(Q)\ge f$,
choosing $Q^{\prime}=Q$ completes the proof of the claim.

Otherwise, $\tau_{1}\neq-\infty$ and $\lambda_{\mu}>t_{i}$. Let
$(S,\phi)$ be the minimal-cost ancestor-closed solution for $Q$
under $u$. Defining $\bar{S},Q_{2}$, and $Q_{1}^{u^{\prime}}$ for
every $u^{\prime}\in\bar{S}$ as in Proposition \ref{prop:FLDelay_MinimumPrefixDecomposition},
we have that 
\[
d_{Q}(\hat{t})\ge\psi_{u}(Q)=\left(\sum_{u^{\prime}\in\bar{S}}\psi_{u^{\prime}}(Q_{1}^{u^{\prime}})\right)+f+\sum_{q\in Q_{2}}\delta(u,v_{q})
\]

Now, denote by $\hat{Q}\subseteq Q$ the subset of $Q$ that was pending
immediately after $\Explore_{\tau_{1}}(u)$. We make the following
observations.
\begin{enumerate}
\item For every $q\in Q_{2}$, we have that $d_{q}(\hat{t})\ge\delta(v_{q},u)$.
Otherwise, $Q\backslash\{q\}$ would become critical before $\hat{t}$.
But since $\lambda_{\mu}>\tau_{2}^{\star}\ge\hat{t}$, we must have
that $q\notin\hat{Q}$. Thus, $Q_{2}\cap\hat{Q}=\emptyset$.
\item Writing $\hat{Q}_{1}^{u^{\prime}}=\hat{Q}\cap Q_{1}^{u^{\prime}}$,
we observe that since $\lambda_{\mu}>\hat{t}$, we have that $d_{\hat{Q}_{1}^{u^{\prime}}}(\hat{t})\le\psi_{u^{\prime}}(\hat{Q}_{1}^{u^{\prime}})\le\psi_{u^{\prime}}(Q_{1}^{u^{\prime}})$
\end{enumerate}
Overall, we get that 
\[
d_{\hat{Q}}(\hat{t})=\sum_{u^{\prime}\in\bar{S}}d_{\hat{Q}_{1}^{u^{\prime}}}(\hat{t})\le\sum_{u^{\prime}\in\bar{S}}\psi_{u^{\prime}}(Q_{1}^{u^{\prime}})
\]

Thus, we have that 
\[
d_{Q\backslash\hat{Q}}(\hat{t})=d_{Q}(\hat{t})-d_{\hat{Q}}(\hat{t})\ge f+\sum_{q\in Q_{2}}\delta(u,v_{q})\ge f
\]
Observing that $r_{q}\ge\tau_{1}$ for each $q\in Q\backslash\hat{Q}$
yields the claim, and thus the lemma.
\end{proof}
\begin{lem}
\label{lem:FLDelay_RootExcesses}For every $i\in[k]$, the charging
node $\mu=(r,[t_{i-1},t_{i}))$ has $\chi_{\mu}\ge f$.
\end{lem}

\begin{proof}
Observe that $E_{\mu}^{-}=\emptyset$. It remains to see that $\sum_{\sigma\in E_{\mu}^{+}}\alpha(\sigma)=f$.

If $\texttt{Color}[\mu]\neq\None$, this holds similarly to Lemma
\ref{lem:FLDeadline_RootExcesses}. 

Otherwise, assume that $\texttt{Color}[\mu]=\None$. Since $\Color[\mu]\neq\Special$,
$\opt$ did not open a facility in $[t_{i-1},t_{i})$. We find a set
of requests $Q^{\prime}$ released in $[t_{i-1},t_{i})$ on which
$\opt$ incurs at least $f$ delay. The argument that follows is similar
to that of Case 4 of Lemma \ref{lem:FLDelay_ValidPreflow}, the structure
of which we repeat for clarity.

We must have that either $t_{i-1}=-\infty$ or $\lambda_{\mu}>t_{i}$.
Denote by $Q$ the set of requests that triggered the service at $t_{i}$.
We have that $d_{Q}(t_{i})\ge\psi(Q)$. Observe that $r_{q}<t_{i}$
for every $q\in Q$. If $t_{i-1}=-\infty$, then $r_{q}\ge t_{i-1}$
for every $q\in Q$, and since $\psi(Q)\ge f$ the proof is complete.

Otherwise, $\lambda_{\mu}>t_{i}$. Denoting by $(S,\phi)$ the minimum-cost
ancestor-closed solution for $Q$, we define $\bar{S}$, $Q_{2}$
and $Q_{1}^{u^{\prime}}$ for every $u^{\prime}\in S$ as in Proposition
\ref{prop:FLDelay_MinimumPrefixDecomposition}. Proposition \ref{prop:FLDelay_MinimumPrefixDecomposition}
yields
\[
d_{Q}(t_{i})\ge\left(\sum_{u^{\prime}\in\bar{S}}\psi_{u^{\prime}}(Q_{1}^{u^{\prime}})\right)+f+\sum_{q\in Q_{2}}\delta(u,v_{q})
\]
Define $\hat{Q}\subseteq Q$ to be the subset of $Q$ alive immediately
after the return of $\Explore_{t_{i-1}}(r)$. Using $\lambda_{\mu}>t_{i}$,
and choosing $\hat{t}=t_{i}$, we use an identical argument to Case
4 of Lemma \ref{lem:FLDelay_ValidPreflow} to show that 
\[
d_{\hat{Q}}(t_{i})\le\left(\sum_{u^{\prime}\in\bar{S}}\psi_{u^{\prime}}(Q_{1}^{u^{\prime}})\right)
\]
Choosing $Q^{\prime}=Q\backslash\hat{Q}$ completes the proof of lemma,
identically to Case 4 of Lemma \ref{lem:FLDelay_ValidPreflow}.
\end{proof}
We can now prove Lemma \ref{lem:FLDelay_OPT}.
\begin{proof}[of Lemma \ref{lem:FLDelay_OPT}]
Lemma \ref{lem:FLDelay_ValidPreflow} yields that $Z$ is a valid
preflow. For $i\in[k]$, let $\mu_{i}=(r,[t_{i-1},t_{i}))$. Using
Lemma \ref{lem:FLDelay_RootExcesses} and Proposition \ref{prop:Flow_SubsetLBSource},
we have that 
\[
kf\le\sum_{i=1}^{k}\chi_{\mu_{i}}\le\omega_{Z}
\]

Now observe that $E_{s}^{+}=\emptyset$, and that $\sum_{\sigma\in E_{s}^{-}}\alpha(\sigma)=\sum_{\mu\in M}c(\mu)$.
Using Lemma \ref{lem:FLDelay_ChargeLBOPT}, we obtain
\[
kf\le\omega_{Z}=\sum_{\sigma\in E_{s}^{-}}\alpha(\sigma)=\sum_{\mu\in M}c(\mu)\le(D+1)\cdot\opt^{B}+2\cdot\opt^{C}+(D+1)\cdot\opt^{D}
\]

as required.
\end{proof}
\begin{proof}[Proof of Theorem \ref{thm:FLDelay_HSTTheorem}]
 Using Lemmas \ref{lem:FLDelay_ALG} and \ref{lem:FLDelay_OPT} completes
the proof.
\end{proof}

\section{\label{sec:MAD}Online Multilevel Aggregation with Delay}

\subsection{Problem and Notation}

In the online multilevel aggregation with delay problem, requests
arrive on the leaves of a rooted tree over time. Each such request
accumulates delay until served. At any point in time, an algorithm
for this problem may choose to transmit a subtree which contains the
root, at a cost which is the weight of that subtree. Any pending requests
on a leaf in the transmitted subtree are served by the transmission.

Formally, as in the facility location with delay problem, a request
is a tuple $(v_{q},r_{q},d_{q}(t))$ where the leaf of the request
is $v_{q}$, the arrival time of the request is $r_{q}$ and $d_{q}(t)$
is the request's delay function. The function $d_{q}(t)$ is again
required to be non-decreasing and continuous.

We observe online multilevel aggregation with delay on a $\left(\ge2\right)$-HST.
We assume, without loss of generality, that only a single edge exits
the root node, called the root edge. Otherwise, we operate on each
edge that exits the root node separately, as there is no interaction
between the subtrees rooted at those edges. We denote the tree by
$T$, and its root edge by $r$. 

For a request $q$, and a set of edges $E$ we write that $q\in E$
if the leaf edge on which $q$ is released is in $E$. In accordance,
we write $Q\subseteq E$ if $q\in E$ for every $q\in Q$. For a set
of pending requests $Q$ at time $t$, we denote by $d_{Q}(t)$ the
total delay incurred by the requests of $Q$ until time $t$. We denote
by $w(e)$ the weight of an edge, and by $w(E)=\sum_{e\in E}w(e)$
the total weight of a set of edges.

We assume that each request would gather infinite delay if it remains
pending forever.

The following notations are similar to those for facility location,
but refer to edges instead of nodes.
\begin{defn}[Similar to Definition \ref{def:FLDeadline_TreeDefinitions}]
For every tree edge $e\in T$, we use the following notations:
\begin{itemize}
\item For $e\neq r$, we denote by $p(e)$ the parent edge of $e$ in the
tree. 
\item We denote by $T_{e}$ the subtree rooted at $e$.
\item For a set of requests $Q\subseteq T_{e}$, we denote by $T_{e}^{Q}\subseteq T_{e}$
the subtree spanned by $e$ and the leaves of $Q$. We denote $T^{Q}=T_{r}^{Q}$.
\item We define the \emph{height} of $e$, denoted $h_{e}$, to be the depth
of $T_{e}$.
\end{itemize}
\end{defn}

In this section, we prove the following theorem.
\begin{thm}
\label{thm:MAD_GeneralTree}There exists a $O(D^{2})$-competitive
deterministic algorithm for online multilevel aggregation with delay
on any tree of depth $D$.
\end{thm}

\subsection{Algorithm for HSTs}

We now present an algorithm for the online multilevel aggregation
with delay problem over a $\left(\ge2\right)$-HST of depth $D$. 
\begin{defn}[saturation and critical sets]
\label{def:MAD_CriticalSet}For any edge $e$, we say that a set
of pending requests $Q\subseteq T_{e}$ \emph{saturates }$T_{e}$
if $d_{Q}(t)\ge w(T_{e}^{Q})$. We say that a set of pending requests
$Q$ is \emph{critical} at time $t$ if $Q$ saturates the root edge
$r$. 
\end{defn}

Upon a set of critical requests, the algorithm starts a service. In
every service, the algorithm maintains a tree $\T$, which it expands
and ultimately transmits. 
\begin{defn}[live cut]
At any time during the construction of $\T$, we define the \emph{live
cut under $e\in\T$ }to be the set of edges $E=\{e^{\prime}|e^{\prime}\in T_{e}\backslash\T\wedge p(e^{\prime})\in\T\}$. 
\end{defn}

\textbf{Algorithm's description. }The algorithm is given in Algorithm
\ref{alg:MAD_Algorithm}. When a set of requests is critical, a call
is made to $\UponCritical$, which resets the tree to transmit $\T$,
calls $\Explore(r)$ to expand $\T$, then transmits $\T$. 

The exploration of an edge $e$ adds $e$ to $\T$. It then considers
the live cut underneath $e$, which is the set of potential candidates
for expanding $\T$. The exploration forwards time until a set of
pending requests saturates $T_{e^{\prime}}$ for an edge $e^{\prime}$
in the the live cut. It then invests in raising the counter of $e^{\prime}$,
until either the counter is full (which triggers $\Explore(e^{\prime})$
immediately) or $\Explore(e)$ is out of budget. The counter of $e$,
as well as the budget of $\Explore(e)$, is equal to $w(e)$. 

Note that the live cut under $e$ can change significantly after every
iteration of the loop in $\Explore(e)$, as making a recursive call
to $\Explore(e^{\prime})$ can add many additional edges to $\T$.

\noindent \LinesNumbered \RestyleAlgo{boxruled}

\noindent \renewcommand{\algorithmcfname}{Algorithm}

\noindent 
\begin{algorithm}[tb]
\caption{\label{alg:MAD_Algorithm}Online Multilevel Aggregation with Delay}

\SetKwProg{Fn}{Function}{}{end}
\SetKwProg{EFn}{Event Function}{}{end}
\SetKwFunction{UponCritical}{UponCritical}
\SetKwFunction{Explore}{Explore}
\SetKwFunction{Add}{Add}

\textbf{Initialization.}

Initialize $c_{e}\leftarrow0$ for any edge $e\in T\backslash\{r\}$

Declare $b_{e}$ for every edge $e\in T$.

Declare $\T$.

\;

\EFn(\tcp*[h]{Upon request set becoming critical as per Definition
\ref{def:MAD_CriticalSet}}){\UponCritical{}}{

set $\T\leftarrow\emptyset$

\Explore{$r$}

transmit $\T$

}

\;

\Fn{\Explore{$e$}}{

\Add{$e$}

set $b_{e}\leftarrow w(e)$

\While{$b_{e}\ne0$ \normalfont{\textbf{and}} there remain pending
requests in $T_{e}$}{

let $H$ be the live cut under $e$.

let $Q$ be the set of pending requests in $T_{e}$.

let $t^{\prime}$ be the earliest time such that there exists a set
of requests $Q^{\prime}\subseteq Q$ that saturates $T_{e^{\prime}}$
for some $e^{\prime}\in H$.

call \Invest{$e$,$e^{\prime}$} 

\lIf{$c_{e^{\prime}}=w(e^{\prime})$}{set $c_{e^{\prime}}=0$ $\texttt{\textbf{;}}$
call \Explore{$e^{\prime}$}.}

}

}
\end{algorithm}

\noindent \LinesNumbered \RestyleAlgo{boxruled}

\noindent \renewcommand{\algorithmcfname}{Algorithm}
\begin{algorithm}[tb]
\caption{Online Multilevel Aggregation with Delay (cont.)}

\ContinuedFloat   \caption*{Online Multilevel Aggregation with Delay (cont.)}
\SetKwProg{Fn}{Function}{}{end}
\SetKwProg{EFn}{Event Function}{}{end}
\SetKwFunction{UponCritical}{UponCritical}
\SetKwFunction{Explore}{Explore}
\SetKwFunction{Add}{Add}

\Fn{\Add{$e$}}{

$\T\leftarrow\T\cup\{e\}$

\lIf{$e$ is a leaf edge}{mark all pending requests on $e$ as served}

}

\;

\Fn{\Invest{$e$,$e^{\prime}$}}{

let $y\leftarrow\min(b_{e},w(e^{\prime})-c_{e^{\prime}})$

increase $c_{e^{\prime}}$ by $y$

decrease $b_{e}$ by $y$

\Return{$y$}.

}

\end{algorithm}

\subsection{Analysis}

Fix any instance of online multilevel aggregation with delay, and
observe the behavior of Algorithm \ref{alg:MAD_Algorithm} for that
instance. We denote by $\alg$ the algorithm's total cost. We also
define $\alg^{B}$ to be the algorithm's buying cost, and $\alg^{D}$
to be the algorithm's delay cost, such that $\alg=\alg^{B}+\alg^{D}$.
We similarly define $\opt,\opt^{B}$ and $\opt^{D}$ for the optimal
solution for the instance.

In this subsection, we prove the following theorem.
\begin{thm}
\label{thm:MAD_HSTTheorem}$\alg\le O(D)\cdot\opt^{B}+O(D^{2})\cdot\opt^{D}$
\end{thm}

In the following analysis, we denote by $k$ the number of times that
the algorithm transmits a tree. We also denote the times of the $k$
transmissions by $t_{1},...,t_{k}$ in increasing order.

\subsubsection{Upper Bounding $\protect\alg$}

We upper bound the cost of the algorithm by proving the following
lemma.
\begin{lem}
\label{lem:MAD_ALG}$\alg\le2kDw(r)$
\end{lem}

The main technique used in proving Lemma \ref{lem:MAD_ALG} is constructing
a preflow to provide an upper bound for $\alg^{B}$. Bounding $\alg^{D}$
by $\alg^{B}$ then yields the lemma.

\begin{obs}
\label{obs:MAD_EveryServiceServes}Every call to $\Explore(r)$ serves
at least one pending request.
\end{obs}

\begin{prop}
\label{prop:MAD_EveryRequestServed}Every request is eventually served.
\end{prop}

\begin{proof}
Consider a request $q$. As assumed in the model, the delay of $q$
goes to infinity as $q$ remains pending. But at some point, the delay
of $q$ would exceed $T_{r}^{\{q\}}$, making $\{q\}$ critical, and
triggering calls to $\Explore(r)$ until $q$ is served. Each such
call serves at least one pending request due to Observation \ref{obs:MAD_EveryServiceServes},
and thus $q$ will eventually be served.
\end{proof}
The following observation follows from the fact that a tree is transmitted
whenever a set of requests becomes critical.
\begin{obs}
\label{obs:MAD_NoCriticalInvariant}At any time $t$ during the algorithm,
and for any set of requests $Q$ pending at $t$, it holds that $d_{Q}(t)\le w(T^{Q})$.
\end{obs}

\begin{lem}
\label{lem:MAD_DelayLessThanBuy}$\alg^{D}\le\alg^{B}$.
\end{lem}

\begin{proof}
Denote by $\mathcal{Q}$ the set of all requests released in the instance.
Through Proposition \ref{prop:MAD_EveryRequestServed}, we can partition
$\mathcal{Q}$ into the sets of requests $Q_{i}$, for $i\in k$,
such that $Q_{i}$ is served in the $i$'th service. Denote by $T_{i}$
the tree bought by the algorithm in the $i$'th service, and denote
by $d(Q_{i})$ the total delay incurred by the algorithm on the requests
of $Q_{i}$. To prove the lemma, it is enough to show that $d(Q_{i})\le w(T_{i})$
for every $i\in[k]$.

Now, observe that since all of $Q_{i}$ are served in $t_{i}$. Therefore,
$d(Q_{i})=d_{Q_{i}}(t_{i})$. Since transmitting $T_{i}$ serves $Q_{i}$,
we have that $T^{Q_{i}}\subseteq T_{i}$. Using Observation \ref{obs:MAD_NoCriticalInvariant},
we have that $d(Q_{i})\le w(T_{i})$ as required.
\end{proof}
It remains to bound $\alg^{B}$. 

Let $V$ be the set of calls to \Explore made by the algorithm. Observe
that in the algorithm, whenever an edge $e$ is bought, a call to
\Explore{$e$} is made immediately afterwards. Therefore, we have
that $\alg^{B}=\sum_{\Explore_{\tau}(e)\in V}w(e)$. 

In addition, immediately prior to calling \Explore{$e$} the counter
$c_{e}$ is zeroed. We say that $\Explore_{\tau_{1}}(e_{1})$ \emph{invested
}$x$ in $\Explore_{\tau_{2}}(e_{2})$ if $\Explore_{\tau_{1}}(e_{1})$
raised $c_{e_{2}}$ by $x$, such that the next zeroing of $c_{e_{2}}$
triggers $\Explore_{\tau_{2}}(e_{2})$.

We now construct a graph $G=(V\cup\{s\},E)$ and a weight function
$\alpha:E\to\R^{+}$, such that $Z=(G,s,\alpha)$ is a preflow. We
construct $E$ and $\alpha$ in the following way:
\begin{enumerate}
\item For every $j\in[k]$, and for every root function call $\Explore_{\tau}(r)$,
add to $E$ an edge $\sigma$ from $s$ to $\Explore_{\tau}(r)$,
and set $\alpha(\sigma)=D\cdot w(r)$.
\item For every function call $\Explore_{\tau}(e)\in V$, and for each function
call $\Explore_{\tau^{\prime}}(e^{\prime})\in V$ that invested some
amount $x$ in $\Explore_{\tau}(e)$, we add to $E$ an edge $\sigma$
from $\Explore_{\tau^{\prime}}(e^{\prime})$ to $\Explore_{\tau}(e)$,
and set $\alpha(\sigma)=h_{e}\cdot x$.
\end{enumerate}
\begin{lem}
\label{lem:MAD_AlgChiGood}For every $v=\Explore_{\tau}(e)\in V$
we have that $\chi_{v}\ge w(e)$, implying that $Z$ is a valid preflow.
\end{lem}

\begin{proof}
We first claim that $\sum_{\sigma\in E_{v}^{+}}\alpha(\sigma)\ge h_{e}\cdot w(e)$.
If $e=r$, this is true since there exists an edge $\sigma$ from
$s$ to $\Explore_{\tau}(e)$ such that $\alpha(\sigma)=Dw(r)\ge h_{e}\cdot w(e)$. 

Otherwise, observe that the total amount invested in $\Explore_{\tau}(e)$
is exactly $w(e)$, and thus $\sum_{\sigma\in E_{v}^{+}}\alpha(\sigma)\ge h_{e}\cdot w(e)$.

Now, observe that $\Explore_{\tau}(e)$ invests at most $w(e)$ in
counters for edges of height at most $h_{e}-1$, and thus $\sum_{\sigma\in E_{v}^{-}}\alpha(\sigma)\le(h_{e}-1)\cdot w(e)$.
Combining this with the previous claim, we get $\chi_{v}\ge w(e)$
as required.
\end{proof}
We can now prove Lemma \ref{lem:MAD_ALG}.
\begin{proof}[of Lemma \ref{lem:MAD_ALG}]
 Observe the preflow $Z$. Note that 
\[
\omega_{Z}=\sum_{\sigma\in E_{s}^{-}}\alpha(\sigma)=kDw(r)
\]

Using Lemmas \ref{lem:MAD_AlgChiGood} and \ref{prop:Flow_SubsetLBSource},
we have 
\[
\alg^{B}=\sum_{\Explore_{\tau}(e)\in V}w(e)\le\sum_{v\in V}\chi_{v}\le\omega_{Z}=kDw(r)
\]
Using Lemma \ref{lem:MAD_DelayLessThanBuy}, we get that $\alg\le2kDw(r)$
as required.
\end{proof}

\subsubsection{Lower Bounding $\protect\opt$}

The following lemma provides a lower bound on the cost of the optimum.
\begin{lem}
\label{lem:MAD_OPT}$kw(r)\le\opt^{B}+D\cdot\opt^{D}$.
\end{lem}

\subparagraph*{Charging nodes and incurred costs.}

We now define charging nodes for the analysis of our algorithm. The
charging nodes are tuples of the form $(e,[\tau_{1},\tau_{2}))$,
such that $\tau_{1}$ and $\tau_{2}$ are two subsequent times in
which the edge $e$ is bought. As in the facility location case, we
allow $\tau_{1}=-\infty$ and $\tau_{2}=\infty$.

For a charging node $\mu=(e,[\tau_{1},\tau_{2}))$ we say that:
\begin{itemize}
\item $\opt$ incurs a \emph{buying cost }of $w(e)$ in $\mu$ if $\opt$
bought the edge $e$ during $[\tau_{1},\tau_{2})$. We denote the
buying cost that $\opt$ incurs in $\mu$ by $c_{b}(\mu)$.
\item $\opt$ incurs a \emph{delay cost }in $\mu$ equal to the delay incurred
by $\opt$ on the set of requests $Q=\{q\in T_{e}|r_{q}\in[\tau_{1},\tau_{2})\}$.
We denote the delay cost that $\opt$ incurs in $\mu$ by $c_{d}(\mu)$.
\end{itemize}
We denote the total cost that $\opt$ incurs in $\mu$ by $c(\mu)=c_{b}(\mu)+c_{d}(\mu)$.

Denote by $M$ the set of all charging nodes. To prove Lemma \ref{lem:MAD_OPT},
we show a preflow on the set of vertices $M\cup\{s\}$, where $s$
is the source node.

The following definition of charging node investment is very similar
to the definition for the facility location case.
\begin{defn}[Investing]
For two charging nodes $\mu_{1}=(e_{1},[\tau_{1}^{1},\tau_{2}^{1}))$
and $\mu_{2}=(e_{2},[\tau_{1}^{2},\tau_{2}^{2}))$, such that $e_{1}$
is an ancestor of $e_{2}$, we say that $\mu_{1}$ \emph{invested
$x$ in $\mu_{2}$ }if $\Explore_{\tau_{1}^{1}}(e_{1})$ raised the
counter $c_{e_{2}}$ by $x$, through calls to $\Invest$, during
the counter phase of $c_{e_{2}}$ between $\tau_{1}^{2}$ and $\tau_{2}^{2}$.
\end{defn}

\begin{defn}[$\lambda_{e}^{t}$ and $\lambda_{\mu}$]
For every function call $\Explore_{t}(e)$ for some edge $e\in T$
and time $t$, let $Q$ be the set of requests pending in $T_{e}$
immediately after the return of $\Explore_{t}(e)$. We define $\lambda_{e}^{t}$
to be the first time $t^{\prime}\ge t$ such that there exists $Q^{\prime}\subseteq Q$
such that $d_{t}(Q^{\prime})\ge w(T_{e}^{Q^{\prime}})-w(e)$. 

For a charging node $\mu=(e,[\tau_{1},\tau_{2}))$ such that $\tau_{1}\neq-\infty$,
we write $\lambda_{\mu}=\lambda_{e}^{\tau_{1}}$.
\end{defn}

\paragraph*{Possible edges. }

We describe the set of possible edges in $G$ from nodes in $M$ to
other nodes in $M$, denoted by $\bar{E}$, and the weight function
$\alpha:\bar{E}\rightarrow\R^{+}$. The final set of edges added to
$G$ by Procedure \ref{proc:MAD_PreflowBuilder} from the nodes of
$M$ to themselves is a subset of $\bar{E}$. The set $\bar{E}$ contains
an edge $\sigma$ from any charging node $\mu_{1}\in M$ to any charging
node $\mu_{2}\in M$ if $\mu_{1}$ invested in $\mu_{2}$. We set
the weight $\alpha(\sigma)$ to be the amount that $\mu_{1}$ invested
in $\mu_{2}$.

We can now construct the preflow required for the analysis using Procedure
\ref{proc:MAD_PreflowBuilder}. This procedure is very similar to
Procedure \ref{proc:FLDeadline_PreflowBuilder}, used for analysis
of our algorithms for facility location. It uses the function $\SetColor$
as defined in Procedure \ref{proc:FLDeadline_PreflowBuilder}.

The procedure for the construction is given in Procedure \ref{proc:MAD_PreflowBuilder}
very similar to that given in Procedure \ref{proc:FLDeadline_PreflowBuilder}.

\noindent \LinesNumbered \RestyleAlgo{boxruled}\renewcommand{\algorithmcfname}{Procedure}\DontPrintSemicolon
\begin{algorithm}[tb]
\caption{\label{proc:MAD_PreflowBuilder}PreflowBuilder - Online Multilevel
Aggregation with Delay}

\SetKwProg{Fn}{Function}{}{end}
\SetKwFunction{PreflowBuilder}{PreflowBuilder}
\SetKwFunction{SetColor}{SetColor}
\SetKw{Break}{break}

\textbf{Initialization.}

Let the set of vertices be $M\cup\{s\}$, and initialize the edge
set to be $E=\emptyset$.

Initialize dictionary $\texttt{Color}[w]=None$ for every $\mu\in M$.

\ForEach{$\mu=(e,[\tau_{1},\tau_{2}))\in M$ such that $\opt$ transmitted
edge $e$ during $[\tau_{1},\tau_{2})$}{

set $\Color[\mu]\leftarrow\Special$

}

\ForEach{$\mu\in M$ such that $c(\mu)>0$}{

add a new edge $\sigma=(s,\mu)$ to $E$, and set $\alpha(\sigma)=c(\mu)$

}

\;

\Fn{\PreflowBuilder{}}{

\For{$i$ from $1$ to $k$}{

let $\mu\leftarrow(r,[t_{i-1},t_{i}))$

\SetColor{$\mu$,$\mu$}

}

\For{$j$ from $2$ to $D$}{

\ForEach{$\mu=(e,[\tau_{1},\tau_{2}))\in M$ such that $e$ is of
depth $j$}{

\ForEach{edge $\sigma\in E_{\mu}^{-}$ incoming to a node $\mu^{\prime}$}{

\lIf{\SetColor{$\mu,\texttt{Color}[\mu^{\prime}]$}$\neq\None$}{\Break}

}

}

}

}

\lFn(\tcp*[h]{As in Procedure \ref{proc:FLDeadline_PreflowBuilder}}){\SetColor{$\mu$,$\mu^{\star}$}}{}
\end{algorithm}

\begin{defn}[Cut]
 We say that a set of edges $H\subseteq T$ is a \emph{cut }if no
edge in $H$ is an ancestor of another edge in $H$. 
\end{defn}

It is easy to verify that any live cut is a cut.
\begin{prop}
\label{prop:MAD_LambdaLargerThanCutCrit}Let $e$ be an edge, and
let $H$ be a cut in $T_{e}$ that does not include $e$. Let $Q\subseteq\bigcup_{h\in H}T_{h}$
be a set of pending requests and $t$ be a time such that $d_{Q}(t)\ge w(T_{e}^{Q})-w(e)$.
Then there exists an $h\in H$ and a subset $Q_{h}\subseteq Q$ such
that $Q_{h}\subseteq T_{h}$ and $d_{Q_{h}}(t)\ge w(T_{h}^{Q_{h}})$.
\end{prop}

\begin{proof}
Partition $Q$ into $|H|$ disjoint sets $Q_{h}$ for every $h\in H$,
according to the subtree $T_{h}$ in which the requests are. Now,
observe that $w(T_{e}^{Q})\ge w(e)+\sum_{h\in H}w(T_{h}^{Q_{h}})$.
We thus have that 
\[
\sum_{h\in H}d_{Q_{h}}(t)=d_{Q}(t)\ge w(T_{e}^{Q})-w(e)\ge\sum_{h\in H}w(T_{h}^{Q_{h}})
\]
and thus there exists $h\in H$ such that $d_{Q_{h}}(t)\ge w(T_{h}^{Q_{h}})$,
as required.
\end{proof}
\begin{prop}
\label{prop:MAD_ForwardSmallerThanLambda}Observe the function call
$\Explore_{t}(e)$, and let $P$ be the set of times chosen as $t^{\prime}$
in $\Explore_{t}(e)$. Then for every $t^{\prime}\in P$ we have that
$t^{\prime}\le\lambda_{e}^{t}$.
\end{prop}

\begin{proof}
Fix some point during the execution of $\Explore_{t}(e)$. Denote
by $Q$ the set of pending requests in $T_{e}$, and let $\lambda\ge t$
be the first time such that there exists $Q^{\prime}\subseteq Q$
for which $d_{Q^{\prime}}(\lambda)\ge w(T_{e}^{Q^{\prime}})-w(e)$.
Observe the next time chosen as $t^{\prime}$ in $\Explore_{t}(e)$.
Observe that the subtrees rooted in edges of the current live cut
contain all of $Q$. Thus, using Proposition \ref{prop:MAD_LambdaLargerThanCutCrit},
we obtain that $t^{\prime}\le\lambda$.

Since during $\Explore_{t}(e)$ requests are being served but do not
arrive, we have that $\lambda$ only increases during $\Explore_{t}(e)$.
Since the final value of $\lambda$ is $\lambda_{e}^{t}$, for every
$t^{\prime}\in P$ we have $t^{\prime}\le\lambda_{e}^{t}$ as required.
\end{proof}
\begin{prop}
For every charging node $\mu=(e,[\tau_{1},\tau_{2}))\in M$, it holds
that $\sum_{\sigma\in E_{\mu}^{-}}\alpha(\sigma)\le w(e)$.
\end{prop}

\begin{proof}
Observe that an outgoing edge $\sigma$ from $\mu$ only goes to a
node $\mu^{\prime}$ that invested in $\mu$, and is labeled $\alpha(\sigma)=x$
where $x$ is the amount that $\mu^{\prime}$ invested in $\mu$.
These amount sum to at most $w(e)$, since the counter $c_{e}$ can
only reach $w(e)$ before it is zeroed and $e$ is bought (thus ending
the counter phase $[\tau_{1},\tau_{2})$).
\end{proof}
\begin{cor}
\label{cor:MAD_w_e_is_enough}Every charging node $\mu\in M$ such
that $\sum_{\sigma\in E_{\mu}^{+}}\alpha(\sigma)\ge w(e)$ has that
$\chi_{\mu}\ge0$.
\end{cor}

The following observation results from the condition checks in $\SetColor$.
\begin{obs}
\label{obs:MAD_ColoredNotMinusInfinity}For any node $\mu=(e,[\tau_{1},\tau_{2}))$
such that $\texttt{Color}[\mu]=\mu^{\star}$ for some charging node
$\mu^{\star}$, we have that $\tau_{1}\neq-\infty$, and also $\lambda_{\mu}<\infty$.
\end{obs}

The following Proposition is analogous to Proposition \ref{prop:FLDeadline_NotBought},
and its proof is identical.
\begin{prop}
\label{prop:MAD_NotBought}Let $\mu=(e,[\tau_{1},\tau_{2}))$ such
that $\texttt{Color}[\mu]=\mu^{\star}$ for some charging node $\mu^{\star}=(r,[\tau_{1}^{\star},\tau_{2}^{\star}))$.
Then $\opt$ did not transmit $e$ during $[\tau_{1},\tau_{2}^{\star})$.
\end{prop}

\begin{lem}
\label{lem:MAD_ValidPreflow}The preflow defined by Procedure \ref{proc:MAD_PreflowBuilder}
is valid.
\end{lem}

\begin{proof}
We need to show that $\chi_{\mu}\ge0$ for every $\mu=(e,[\tau_{1},\tau_{2}))\in M$.

We consider the following cases:

\textbf{Case 1: }$\texttt{Color}[\mu]=\Special$. In this case, we
have that $c(\mu)\ge c_{b}(\mu)=w(e)$. Observe that an edge $\sigma$
from $s$ to $\mu$ is created with $\alpha(\sigma)=c(\mu)$, and
thus from Corollary \ref{cor:MAD_w_e_is_enough} we have that $\chi_{\mu}\ge0$.

\textbf{Case 2:} $\texttt{Color}[\mu]=\mu^{\star}$, for some charging
node $\mu^{\star}$. Using Observation \ref{obs:MAD_ColoredNotMinusInfinity},
observe that $\tau_{1}\neq-\infty$ and $\lambda_{\mu}<\infty$, and
thus the call $\Explore_{\tau_{1}}(e)$ has raised  counters by exactly
$w(e)$, and thus $\mu$ has invested a total of $w(e)$ in other
charging nodes. Thus, in $\sum_{\sigma\in\bar{E}_{\mu}^{+}}\alpha(\sigma)=w(e)$.
Observe that $\SetColor$ added the edges of $\bar{E}_{\mu}^{+}$
to $E$, and thus $\sum_{\sigma\in E_{\mu}^{+}}\alpha(\sigma)\ge w(e)$.
Using Corollary \ref{cor:MAD_w_e_is_enough}, we have that $\chi_{\mu}\ge0$.

\textbf{Case 3: }$\Color[\mu]=\None$. If there are no outgoing edges
from $\mu$, then clearly $\chi_{\mu}\ge0$ and we are done. Otherwise,
there exists an outgoing edge $\sigma$ to some node $\mu^{\prime}=(e^{\prime},[\tau_{1}^{\prime},\tau_{2}^{\prime}))$
with $\Color[\mu^{\prime}]=\mu^{\star}$, for some charging node $\mu^{\star}$.
Denote $\mu^{\star}=(r,[\tau_{1}^{\star},\tau_{2}^{\star}))$, and
observe that since $\mu^{\prime}$ invested in $\mu$, we must have
that $\tau_{1}^{\prime}\le\tau_{2}$. Using Proposition \ref{prop:MAD_NotBought},
and the fact that $\Color[\mu]\neq\Special$, we have that $\opt$
did not transmit $e$ during $[\tau_{1},\tau_{2}^{\star})$.

\textbf{Claim -- }There exists a set of requests $Q^{\prime}\subseteq T_{e}$
such that $r_{q}\in[\tau_{1},\tau_{2})$ such that $d_{Q^{\prime}}(\tau_{2}^{\star})\ge w(e)$.
\begin{proof}[Proof of Claim]
Since $\Color[\mu^{\prime}]=\mu^{\star}$, it must be that $\tau_{1}^{\prime}\neq-\infty$
and $\lambda_{\mu^{\prime}}\le\tau_{2}^{\star}$. Since $\mu^{\prime}$
invested in $\mu$, we have that at some point during $\Explore_{\tau_{1}^{\prime}}(e^{\prime})$,
$e$ was in the live cut under $e^{\prime}$, and the algorithm detected
a set of pending requests $Q\subseteq T_{e}$ such that $d_{Q}(\hat{t})\ge w(T_{e}^{Q})$
for some time $\hat{t}\ge\tau_{1}^{\prime}$. From Proposition \ref{prop:MAD_ForwardSmallerThanLambda},
we have that $\hat{t}\le\lambda_{\mu}\le\tau_{2}^{\star}$. Note also
that since $Q$ is pending at $\tau_{1}^{\prime}$, we have that $r_{q}<\tau_{1}^{\prime}\le\tau_{2}$.
for every $q\in Q$.

Now observe that since $\Color[\mu^{\prime}]=\mu^{\star}$, $\Color[\mu]=\None$,
and there exists an edge from $\mu$ to $\mu^{\prime}$, we must have
that $\SetColor(\mu,\mu^{\star})$ was called. Since $\Color[\mu]=\None$,
it must be that either $\tau_{1}=-\infty$ or $\lambda_{\mu}>\tau_{2}^{\star}$. 

If $\tau_{1}=-\infty$, then $r_{q}\ge\tau_{1}$ for every $q\in Q$.
Combining this with $r_{q}<\tau_{2}$, we have that $r_{q}\in[\tau_{1},\tau_{2})$
for every $q\in Q$. We also have that $d_{Q}(\tau_{2}^{\star})\ge d_{Q}(\hat{t})\ge w(T_{e}^{Q})\ge w(e)$,
by $Q$'s definition. Thus choosing $Q=Q^{\prime}$ proves the claim
for this case.

Otherwise, we have that $\tau_{1}\neq-\infty$ and $\lambda_{\mu}>\tau_{2}^{\star}$.
We denote by $\hat{Q}\subseteq Q$ the subset of requests pending
immediately after the return of $\Explore_{\tau_{1}}(e)$. By the
definition of $\lambda_{\mu}$, and since $\hat{t}<\lambda_{\mu}$,
we have that $d_{\hat{Q}}(\hat{t})<w(T_{e}^{\hat{Q}})-w(e)$. Thus,
\begin{align*}
d_{Q\backslash\hat{Q}}(\tau_{2}^{\star}) & \ge d_{Q\backslash\hat{Q}}(\hat{t})=d_{Q}(\hat{t})-d_{\hat{Q}}(\hat{t})\\
 & \ge w(T_{e}^{\hat{Q}})-\left(w(T_{e}^{\hat{Q}})-w(e)\right)=w(e)
\end{align*}
Denote $Q^{\prime}=Q\backslash\hat{Q}$. The requests of $Q^{\prime}$
were not pending immediately during $\Explore_{\tau_{1}}(e)$, and
therefore $r_{q}\ge\tau_{1}$ for any $q\in Q^{\prime}$. As seen
before, for every $q\in Q$ we have that $r_{q}<\tau_{2}$, and thus
for any $q\in Q^{\prime}$ we have $r_{q}\in[\tau_{1},\tau_{2})$.
$Q^{\prime}$ therefore proves the claim. 
\end{proof}
We now use the claim. As shown before, $\opt$ did not buy $e$ during
$[\tau_{1},\tau_{2}^{\star})$, and has therefore did not serve any
request from $Q^{\prime}$ until time $\tau_{2}^{\star}$. Therefore,
$\opt$ incurs a delay cost of $w(e)$ at $\mu$ on the requests of
$Q^{\prime}$, and thus $c(\mu)\ge w(e)$. Observe that an edge $\sigma$
from $s$ to $\mu$ is created with $\alpha(\sigma)=c(\mu)$, and
thus Corollary \ref{cor:MAD_w_e_is_enough} implies that $\chi_{\mu}\ge0$.
This concludes the proof of Lemma \ref{lem:MAD_ValidPreflow}.
\end{proof}
\begin{lem}
\label{lem:MAD_RootExcesses}For every $j\in[k]$, the charging node
$\mu=(r,[t_{j-1},t_{j}))$ has that $\chi_{\mu}\ge w(r)$.
\end{lem}

\begin{proof}
We denote $\tau_{1}=t_{j-1},\tau_{2}=t_{j}$. Observe that no other
nodes invest in $\mu$, and thus $E_{\mu}^{-}=\emptyset$. It remains
to show that $\sum_{e\in E_{\mu}^{+}}\alpha(e)\ge w(e)$.

If $\texttt{Color}[\mu]\neq\None$, then identically to Cases 1 and
2 of Lemma \ref{lem:MAD_ValidPreflow}, we have that $\sum_{e\in E_{\mu}^{+}}\alpha(e)\ge w(e)$.
This completes the proof for these cases.

If $\Color[\mu]=\None$, then we have a very similar proof to case
3 of Lemma \ref{lem:MAD_ValidPreflow}. Observe the pending requests
$Q$ that became critical at $\tau_{2}$, triggering the service.
Clearly, $r_{q}<\tau_{2}$ for every $q\in Q$. Observe that $\SetColor(\mu,\mu)$
was called, yet $\Color[\mu]=\None$. Thus, it must be that either
$\tau_{1}=-\infty$ or $\lambda_{\mu}>\tau_{2}$. To complete the
proof, we need the following claim.

\textbf{Claim -- }there exists a set of requests $Q^{\prime}$ such
that $r_{q}\in[\tau_{1},\tau_{2})$ for every $q\in Q^{\prime}$,
and $d_{Q^{\prime}}(\tau_{2})\ge w(r)$.
\begin{proof}[Proof of Claim]
 Observe the two cases of the claim. If $\tau_{1}=-\infty$, then
$r_{q}\in[\tau_{1},\tau_{2})$ for every $q\in Q$. Together with
the fact that $d_{Q}(\tau_{2})\ge w(T^{Q})\ge w(r)$, choosing $Q^{\prime}=Q$
proves the claim.

Otherwise, $\tau_{1}\neq-\infty$ and $\lambda_{\mu}>\tau_{2}$. In
this case, denote by $\hat{Q}\subseteq Q$ the requests of $Q$ pending
immediately after $\Explore_{\tau_{1}}\{r\}$ and observe, as in Case
3 of Lemma \ref{lem:MAD_ValidPreflow}, that $d_{\hat{Q}}(\tau_{2})\le w(T^{\hat{Q}})-w(r)\le w(T^{Q})-w(r)$.
Thus, we have that $d_{Q\backslash\hat{Q}}(\tau_{2})\ge w(r)$. Observe
that $r_{q}\ge\tau_{1}$ for every $q\in Q\backslash\hat{Q}$, and
thus $r_{q}\in[\tau_{1},\tau_{2})$ for every $q\in Q\backslash\hat{Q}$.
Thus choosing $Q^{\prime}=Q\backslash\hat{Q}$ yields the claim.
\end{proof}
Now, observe that $\Color[\mu]\neq\Special$ and thus $\opt$ did
not transmit $e$ during $[\tau_{1},\tau_{2})$. Using the claim,
$\opt$ incurred delay cost of at least $w(r)$ on $\mu$ due to $Q^{\prime}$.
Thus $c(\mu)\ge w(r)$, and thus $\sum_{e\in E_{\mu}^{+}}\alpha(e)\ge w(r)$,
completing the proof of the lemma.
\end{proof}
\begin{prop}
\label{prop:MAD_ChiSBoundsOpt}$\omega_{Z}\le\opt^{B}+D\cdot\opt^{D}$
\end{prop}

\begin{proof}
Observe that $E_{s}^{+}=\emptyset$, and that for every $\sigma\in E_{s}^{-}$
to a node $\mu\in M$ we have that $\alpha(\sigma)=c(\mu)$. Therefore,
$\omega_{Z}=\sum_{\mu\in M}c(\mu)$.

Observe that for buying an edge $e$ at time $t$, $\opt$ incurs
buying cost only at the unique charging node $(e,[\tau_{1},\tau_{2}))$
such that $t\in[\tau_{1},\tau_{2})$. 

In addition, when $\opt$ incurs delay for a request $q$ released
on leaf edge $e$, it incurs delay cost in at most $D$ charging nodes,
of the form $(e^{\prime},[\tau_{1},\tau_{2}))$ such that $r_{q}\in[\tau_{1},\tau_{2})$
and $e^{\prime}$ is an ancestor of $e$.

Thus, $\sum_{\mu\in M}c(\mu)\le\opt^{B}+D\cdot\opt^{D}$, proving
the proposition.
\end{proof}
We can now prove Lemma \ref{lem:MAD_OPT}.
\begin{proof}[Proof (of Lemma \ref{lem:MAD_OPT})]
 Observe the set of charging nodes $N=\{(r,[t_{j-1},t_{j})|j\in[k]\}$.
Using Lemma \ref{lem:MAD_RootExcesses}, we have that $\sum_{\mu\in N}\chi_{\mu}\ge kw(r)$. 

We now use Propositions \ref{prop:Flow_SubsetLBSource} and \ref{prop:MAD_ChiSBoundsOpt}
to obtain 
\[
kw(r)\le\omega_{Z}\le\opt^{B}+D\cdot\opt^{D}
\]
 proving the lemma.
\end{proof}
We now prove the main theorem for this subsection.
\begin{proof}[Proof of Theorem \ref{thm:MAD_HSTTheorem}]
The theorem results immediately from Lemmas \ref{lem:MAD_ALG} and
\ref{lem:MAD_OPT}.
\end{proof}

\subsection{From HSTs to General Trees}

In this subsection, we show how to extend our result for multilevel
aggregation on $\left(\ge2\right)$-HSTs to general trees, thus proving
Theorem \ref{thm:MAD_GeneralTree}. To do so, we use a similar method
to that used in \cite{DBLP:conf/soda/BuchbinderFNT17} to form a virtual
forest of $\left(\ge2\right)$-HSTs, based on the edges of the original
tree. 

\paragraph*{The decomposition.}

Let $T$ be the tree, with general weights, rooted at root edge $r$.
We create a forest, the edges of which are the edges of $T$. 

\begin{figure}[tb]
\begin{centering}
\includegraphics{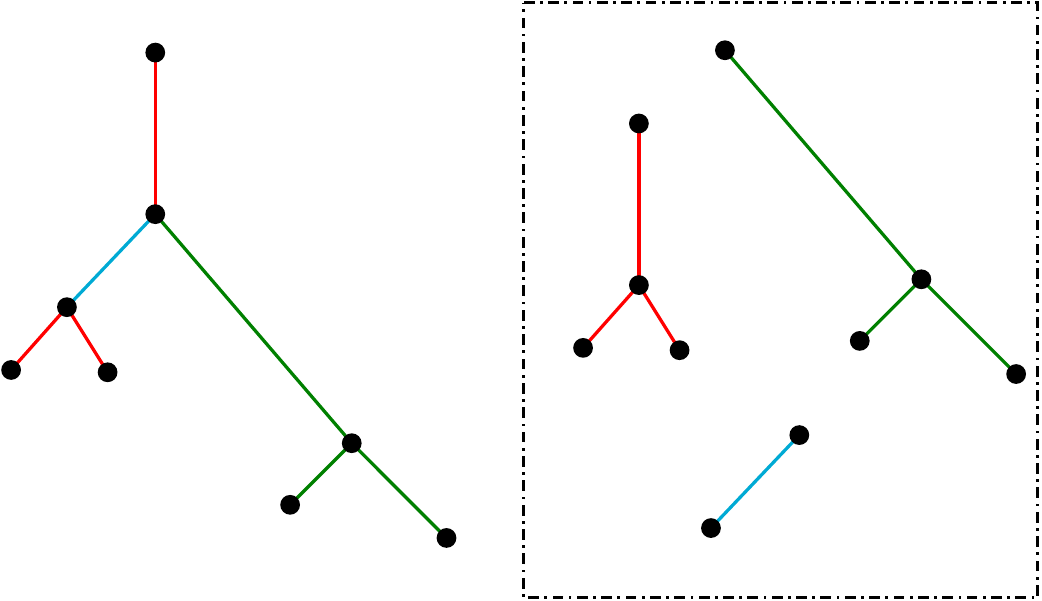}
\par\end{centering}
This figure shows an example of a forest decomposition. The original
tree is on the left, and the virtual forest of $\left(\ge2\right)$-HSTs
resulting from the decomposition in given in the dash-dotted box on
the right.

\caption{Forest Decomposition}

\end{figure}

\begin{defn}[parenthood in virtual $\left(\ge2\right)$-HST]
For every edge $e$, we define $p^{\prime}(e)$, the \emph{virtual
parent }of $e$, to be the least ancestor $e^{\prime}$ of $e$ in
$T$ such that $w(e)\le2w(e^{\prime})$. If there is no such $e^{\prime}$,
then $e$ is \emph{the root edge }of a virtual tree in the forest.
\end{defn}

We define the forest according to the function $p^{\prime}$. Observe
that each connected component is indeed a tree, and specifically a
$\left(\ge2\right)$-HST. Denote by $T^{1},....,T^{m}$ the virtual
trees formed from $T$, and denote by $r^{i}$ the root edge of $T^{i}$

Let $I$ be an instance of online multilevel aggregation with delay.
We partition the requests of $I$ to $I^{1},...,I^{m}$, such that
a request belongs to $I^{i}$ if the leaf edge $v_{q}\in I^{i}$.

We denote by $\opt_{i}$ the optimal solution for the multilevel aggregation
instance $I_{i}$ in the virtual tree $T_{i}$. Using an identical
argument to Observation 4.2 in \cite{DBLP:conf/soda/BuchbinderFNT17},
we have the following observation.
\begin{obs}
\label{obs:MAD_GeneralTreeBigOptAtLeastSmallOpts}$\opt\ge\sum_{i=1}^{m}\opt_{i}$
\end{obs}

\begin{defn}
Let $e\in T_{i}$. We define $B_{e}$ to be the set of edges in $T$
on the path from $e$ to $p^{\prime}(e)$ (including $e$, not including
$p^{\prime}(e)$). If $e=r^{i}$, then let $B_{e}$ be all the edges
from $e$ to $r$, including $r$.
\end{defn}

\begin{defn}
Let $\T_{i}$ be some transmittable subtree in $T_{i}$ for any $i$.
We define $\bar{T}_{i}=\bigcup_{e\in\T_{i}}B_{e}$ to be the \emph{concretization
}of $T_{i}$.
\end{defn}

\paragraph*{The algorithm.}

We now describe the algorithm for online multilevel aggregation with
delay on a general tree. The algorithm is:
\begin{enumerate}
\item Run Algorithm \ref{alg:MAD_Algorithm} for each of $T_{1},....,T_{m}$
separately.
\item Whenever the instance of Algorithm \ref{alg:MAD_Algorithm} for $T_{i}$
transmits the virtual subtree $\T_{i}$, transmit its concretization
$\bar{\T}_{i}$.
\end{enumerate}
Observe that any transmission made by the main algorithm indeed serves
the same requests as the original, virtual transmission. We denote
by $\alg_{i}$ the virtual cost of the $\left(\ge2\right)$-HST algorithm
for $T_{i}$ -- that is, the delay of the requests of $I_{i}$ plus
the sum of the costs of virtual transmissions triggered by the $\left(\ge2\right)$-HST
algorithm for $T_{i}$.

We denote by $k_{i}$ for $i\in[m]$ the number of transmissions caused
by the algorithm for $T_{i}$. The following lemma is a restatement
of Lemma \ref{lem:MAD_OPT}.
\begin{lem}
\textup{\label{lem:MAD_OPTGeneralTree}$k_{i}w(r^{i})\le\opt_{i}^{B}+D\cdot\opt_{i}^{D}\le D\cdot\opt_{i}$}
\end{lem}

It remains to bound the cost of the algorithm.
\begin{prop}
\label{prop:MAD_GeneralTreeDelayBoundedByBuying}$\alg^{D}\le\alg^{B}$
\end{prop}

\begin{proof}
Observe that $\alg^{D}=\sum_{i=1}^{m}\alg_{i}^{D}$ and that $\alg^{B}\ge\sum_{i=1}^{m}\alg_{i}^{B}$.
Thus, we have that 
\[
\alg^{D}=\sum_{i=1}^{m}\alg_{i}^{D}\le\sum_{i=1}^{m}\alg_{i}^{B}\le\alg^{B}
\]
where the second inequality is from Lemma \ref{lem:MAD_DelayLessThanBuy}.
\end{proof}
We denote by $\overline{\alg}_{i}^{B}=\sum_{j=1}^{k_{i}}w(\bar{\T}_{i}^{j})$
where $\T_{i}^{j}$ is the $j$'th transmission made by the $T_{i}$
algorithm. Observe that $\alg^{B}=\sum_{i=1}^{m}\overline{\alg}_{i}^{B}$.

The following lemma bounds the cost of the algorithm, and provides
the final component for Theorem \ref{thm:MAD_GeneralTree}. 
\begin{lem}
\label{lem:MAD_GeneralTreeALG}For every $i$, we have that $\overline{\alg}_{i}^{B}\le2Dk_{i}\cdot w(r^{i})$.
\end{lem}

Fix $i\in[m]$. We denote by $\T_{j}$ for $j\in[k_{i}]$ the $j$'th
virtual transmission made by the $T_{i}$-algorithm. For $j\in[k_{i}]$,
we denote by $t_{j}$ the time of $\T_{j}$'s transmission.

To prove Lemma \ref{lem:MAD_GeneralTreeALG}, we construct a preflow,
in a similar manner to the proof of Lemma \ref{lem:MAD_ALG}. However,
in this case we also have nodes that correspond to edges that for
which $\Explore$ is not called.

We now describe the construction of the graph $G=(V\cup\{s\},E)$,
and the weight function $\alpha$, such that $Z=(G,s,\alpha)$ is
a preflow. Each vertex in $V$ is of the form $(e,j)$ where $e\in\bar{\T}_{j}$.
To describe the edge set $E$, we require the following definition.
\begin{defn}[$x$-route]
Let $(e,j)$, $(e^{\prime},j^{\prime})$ be two edges such that $e$
is an ancestor of $e^{\prime}$, and $j^{\prime}\ge j$. Denote by
$e=e_{0},e_{1},...,e_{l}=e^{\prime}$ the path from $e$ to $e^{\prime}$
in $T$. We define an $x$-route from $(e_{1},j_{1})$ to $(e_{2},j_{2})$
to be the set of the following charging node edges.
\begin{enumerate}
\item An edge $\sigma$ from $(e,j)$ to $(e_{1},j^{\prime})$ with $\alpha(\sigma)=x\cdot h_{e_{1}}$.
\item For each $\beta\in[l-1]$, an edge $\sigma$ from $(e_{\beta},j^{\prime})$
to $(e_{\beta+1},j^{\prime})$ with $\alpha(\sigma)=x\cdot h_{e_{\beta+1}}$.
\end{enumerate}
We also define an $x$-route from $s$ to $(e^{\prime},j^{\prime})$
in a similar manner. Let $r=e_{1},e_{2},...,e_{l}=e^{\prime}$ the
path from the root of $T$ to $e^{\prime}$. The edges of this $x$-route
are:
\begin{enumerate}
\item An edge $\sigma$ from $s$ to $(r,j^{\prime})$ with $\alpha(\sigma)=x\cdot D$.
\item For each $\beta\in[l-1]$, and edge $\sigma$ from $(e_{\beta},j^{\prime})$
to $(e_{\beta+1},j^{\prime})$ with $\alpha(\sigma)=x\cdot h_{e_{\beta+1}}$.
\end{enumerate}
\end{defn}

We can now describe $E$. The edges of $E$ are constructed in the
following way:
\begin{enumerate}
\item For each $j\in[k_{i}]$, add to $E$ the edges of a $w(r_{i})$-route
from $s$ to $(r_{i},j)$.
\item For two charging nodes $(e_{1},j_{1})$, $(e_{2},j_{2})$ such that
$e_{1},e_{2}\in T_{i}$, $e_{1}$ is an ancestor of $e_{2}$ and $\Explore_{t_{j_{1}}}(e_{1})$
invested $x$ in $\Explore_{t_{j_{2}}}(e_{2})$, add to $E$ the edges
of an $x$-route from $(e_{1},j_{1})$ to $(e_{2},j_{2})$.
\end{enumerate}
\begin{obs}
\label{obs:MAD_GeneralTreeExternalAtMostTwice}For every two edges
$e\in T$, $e^{\prime}\in T_{i}$ such that $e\in B_{e^{\prime}}$
it holds that $w(e)\le2w(e^{\prime})$.
\end{obs}

\begin{lem}
\label{lem:MAD_GeneralTreeAlgChiGood}For every charging node $\mu=(e,j)$
it holds that $\chi_{\mu}\ge\frac{w(e)}{2}$.
\end{lem}

\begin{proof}
Observe that $x$-routes do not 

It must be that $e\in\bar{\T}_{j}$. Hence, there exists an edge $e^{\prime}\in\T_{j}$
such that $e\in B_{e^{\prime}}$. Since $e^{\prime}\in\T_{j}$, then
we are in one of the following cases. 

\textbf{Case 1: }$e^{\prime}=e$, and thus $e^{\prime}\in\T_{j}$.
It can be shown that $\sum_{\sigma\in E_{\mu}^{+}}\alpha(\sigma)\ge w(e)\cdot h_{e}$,
and that $\sum_{\sigma\in E_{\mu}^{-}}\alpha(\sigma)\le w(e)\cdot(h_{e}-1)$,
similarly to the proof of Lemma \ref{lem:MAD_AlgChiGood}.

\textbf{Case 2: }$e\neq e^{\prime}$ and $e^{\prime}\ne r_{i}$. Thus,
$e\notin T_{i}$. Observe that since $e\notin T_{i}$, adding any
$x$-route cannot decrease $\chi_{\mu}$. Indeed, adding an $x$-route
can only create an outgoing edge from $\mu$ when creating an incoming
edge with greater $\alpha$. Thus, we locate a set of $x$-routes
that increases $\chi_{\mu}$ to at least $w(e^{\prime})$. From Observation
\ref{obs:MAD_GeneralTreeExternalAtMostTwice}, we get that $w(e^{\prime})\ge\frac{w(e)}{2}$,
proving the lemma.

If $e^{\prime}=r_{i}$, then a $w(r_{i})$-route is created from $s$
to $(r_{i},j)$. Since $e$ is on the path from $r$ to $r_{i}$,
it must be that the route adds:
\begin{itemize}
\item An incoming edge $\sigma$ to $(e,j)$ with $\alpha(\sigma)\ge h_{e}\cdot w(r_{i})$.
\item An outgoing edge $\sigma^{-}$ from $(e,j)$ with $\alpha(\sigma)\le(h_{e}-1)\cdot w(r_{i})$.
\end{itemize}
showing that $\chi_{\mu}\ge w(r_{i})\ge\frac{w(e)}{2}$.

Otherwise, $e^{\prime}\neq r_{i}$. Observe that any $x$-route to
$(e^{\prime},j)$ contains $\mu$, and increases $\chi_{\mu}$ by
at least $x$ (using the same argument as the case for $e^{\prime}=r_{i}$).
In this case, observe that a total of $w(e^{\prime})$ has been invested
has been invested in $(e^{\prime},j)$ to trigger $\Explore_{t_{j}}(e^{\prime})$.
This completes the proof.
\end{proof}
\begin{proof}[Proof of Lemma \ref{lem:MAD_GeneralTreeALG}]
We have that $\overline{\alg}_{i}^{B}\le2Dk_{i}\cdot w(r_{i})$.

Observe the preflow $Z$ as constructed. We have that $\omega_{Z}=Dk_{i}w(r_{i})$.
From Lemma \ref{lem:MAD_GeneralTreeAlgChiGood}, and using Proposition
\ref{prop:Flow_SubsetLBSource}, we have
\begin{align*}
\overline{\alg}_{i}^{B} & \le\sum_{j=1}^{k_{i}}w(\bar{\T_{j}})=\sum_{(e,j)\in V}w(e)=2\cdot\sum_{\mu=(e,j)\in V}\chi_{\mu}\le2\cdot\omega_{Z}=2Dk_{i}w(r_{i})
\end{align*}
\end{proof}
\begin{proof}[Proof of Theorem \ref{thm:MAD_GeneralTree}]
 From Lemmas \ref{lem:MAD_OPTGeneralTree} and \ref{lem:MAD_GeneralTreeALG},
we have that for every $i\in[m]$ 
\[
\overline{\alg}_{i}^{B}\le2D^{2}\opt_{i}
\]

From Observation \ref{obs:MAD_GeneralTreeBigOptAtLeastSmallOpts},
we have that 
\[
\alg^{B}=\sum_{i=1}^{m}\overline{\alg}_{i}^{B}\le2D^{2}\cdot\sum_{i=1}^{m}\opt_{i}\le2D^{2}\cdot\opt
\]

Using Lemma \ref{prop:MAD_GeneralTreeDelayBoundedByBuying}, we have
that 
\[
\alg\le2\alg^{B}\le4D^{2}\cdot\opt
\]
 as required.
\end{proof}

\section{\label{sec:OSD}Online Service with Delay}

\subsection{Problem and Notation}

In the online service with delay (OSD) problem, a single server exists
on a point in a metric space. Requests arrive on points of the metric
space over time, and accumulate delay until served, where serving
a request requires moving the server to that request. The cost of
moving the server from one point to another is the distance between
those two points in the metric space. The goal is to minimize the
sum of the moving cost and the delay cost. 

Formally, a request is a tuple $q=(v_{q},r_{q},d_{q}(t))$ such that
$v_{q}$ is the point on which $q$ arrives, the request arrives at
time $r_{q}$, and $d_{q}(t)$ is an arbitrary non-decreasing continuous
delay function. We also assume that $d_{q}(t)$ tends infinity as
time progresses. For any instance of OSD $I$, denote by $\alg^{B}$
the total cost of moving the algorithm's server. We also denote by
$\alg^{D}=\sum_{q\in Q}d_{q}(t_{q})$, where $t_{q}$ is the time
in which the request $q$ is served. Then the algorithm's goal is
to minimize the total cost 
\[
\alg=\alg^{B}+\alg^{D}
\]
As in the previous problems in this paper, we also consider the special
case in which the metric space is the leaves of a $\left(\ge2\right)$-HST.
Without loss of generality, we allow an algorithm to move its server
to the internal nodes of the tree, even though they are not a part
of the original metric space. This is implemented by lazy moving of
the server -- that is, the server never really moves to those internal
nodes, but its virtual location in an internal node is kept in the
algorithm's memory for the sake of calculations. 

In this section, we prove the following theorem.
\begin{thm}
\label{thm:OSD_GMSLogSquared}There exists a randomized $O(\log^{2}n)$-competitive
algorithm for online service with delay on a general metric space
of $n$ points.
\end{thm}

\subsection{Algorithm for HSTs}

In this subsection, we present an algorithm for online service with
delay on $\left(\ge2\right)$-HSTs. We assume that the weight of each
edge is a power of $2$ -- this can be enforced, at a loss factor
of $2$ to competitiveness. This algorithm encapsulates our algorithm
for online multilevel aggregation with delay, while using similar
mechanisms to those in \cite{DBLP:conf/stoc/AzarGGP17}.

For an edge $e$, denote $\mathcal{C}(e)=\{e^{\prime}|p(e^{\prime})=p(e)\wedge w(e^{\prime})<w(e)\}$,
the set of sibling edges of $e$ with smaller weight. Note that for
every $e^{\prime}\in\mathcal{C}(e)$ we have $w(e^{\prime})\le\frac{1}{2}w(e)$,
since edge weights are powers of 2. We define the following.
\begin{defn}[Top and bottom nodes]
For an edge $e$, we define $v_{e}^{\top}$ to be the top node of
$e$, and $v_{e}^{\bot}$ to be the bottom node of $e$. 
\end{defn}

\begin{defn}[Relative subtree $R_{e}$]
 For an edge $e$, we define the \emph{relative subtree of $e$}
to be $\{e\}\cup\bigcup_{e^{\prime}\in\mathcal{C}(e)}T_{e^{\prime}}$.
\end{defn}

The following definition is required for defining exactly what we
mean when referring to locations of servers and requests.
\begin{defn}[Locations of servers and requests]
 Consider the location of a server (either the algorithm's or the
optimum's).
\begin{itemize}
\item For $T_{e}$, we say that \emph{the server is internal to $T_{e}$
}if the server is in one of the nodes of $T_{e}$ \emph{other than
$v_{e}^{\top}$}.
\item For $R_{e}=\{e\}\cup\left(\bigcup_{e^{\prime}\in\mathcal{C}(e)}T_{e^{\prime}}\right)$,
we say that \emph{the server is internal to $R_{e}$ }if the server
in one of the nodes of $R_{e}$ \emph{other than $v_{e}^{\bot}$}.
\end{itemize}
The same applies for saying that a \emph{request $q$ is internal
to $T_{e}$ (or $R_{e}$)}, and writing $q\in T_{e}$ (or $q\in R_{e}$).

Let $Q\subseteq R_{e}$ be a set of requests, and denote by $Q\restriction_{T_{e^{\prime}}}=\{q\in T_{e^{\prime}}|q\in Q\}$.
Then we define $R_{e}^{Q}$ to be 
\[
\left(\bigcup_{e^{\prime}\in\mathcal{C}(e)}T_{e^{\prime}}^{Q\restriction_{T_{e^{\prime}}}}\right)\cup\{e\}
\]
\end{defn}

We sometimes write $Y_{e}$ to make claims that refer to either $R_{e}$
or $T_{e}$.
\begin{defn}[Saturation]
We say that a set of requests $Q\subseteq Y_{e}$ \emph{saturates
$Y_{e}$ at time $t$ if $d_{Q}(t)\ge w(Y_{e}^{Q})$.}
\end{defn}

\begin{defn}[Major edges]
We say that an edge $e$ is \emph{major }at a time $t$ if every
edge $e^{\prime}$ on the path from the algorithm's server to $e$
has that $w(e^{\prime})\le w(e)$.
\end{defn}

\begin{defn}[Critical set]
\label{def:OSD_CriticalSet}We say that a set of requests $Q$ is
\emph{critical }at time $t$ if it saturates $Y_{e}$ at time $t$
for an edge $e$ which is major at time $t$.
\end{defn}

\begin{defn}
Let $e$ be an edge, and $Y_{e}$ be either $T_{e}$ or $R_{e}$.
We say that the algorithm's server is \emph{on the other side of $e$
than $Y_{e}$ }if:
\begin{itemize}
\item The server is internal to $T_{e}$ and $Y_{e}=R_{e}$.
\item The server is not internal to $T_{e}$ and $Y_{e}=T_{e}$.
\end{itemize}
\end{defn}

The following proposition allows us to assume that whenever a set
of requests is critical by saturating $Y_{e}$ for a major edge $e$,
we have that the algorithm's server is on the other side of $e$ than
$Y_{e}$.
\begin{prop}
\label{prop:OSD_CriticalTreeOppositeServer}Suppose there exists a
critical set of requests $Q$, saturating $Y_{e}$ for $e$ a major
edge, at some point in time. Then there exists another critical set
of requests $Q^{\prime}$, saturating $Y_{e^{\prime}}$ for another
major edge $e^{\prime}$, such that the algorithm's server is on the
other side of $e^{\prime}$ than $Y_{e^{\prime}}$.
\end{prop}

\begin{proof}
If the server is on the other side of $e$ than $Y_{e}$, we are done.
Suppose otherwise, and let $Q$ be the minimal set saturating $Y_{e}$. 

Consider the case that $Y_{e}=T_{e}$, and the algorithm's server
is internal to $T_{e}$. Note that $e$ cannot be a leaf edge --
otherwise, the server and all requests in $Q$ must be on $v_{e}^{\bot}$,
in contradiction to the requests of $Q$ being pending.
\begin{enumerate}
\item If the server is in $v_{e}^{\bot}$, we can thus choose $e^{\prime}$
to be any child edge of $e$ saturated by $Q\restriction_{T_{e^{\prime}}}$(such
an edge must exist, otherwise $Q$ would not have saturated $T_{e}$). 
\item If the server is internal to $T_{\hat{e}}$ for some $\hat{e}$ child
edge of $e$, then:
\begin{enumerate}
\item If there exists a sibling $\tilde{e}$ of $\hat{e}$ such that $w(\tilde{e})\ge w(\hat{e})$
such that $Q\restriction_{T_{\tilde{e}}}$ is saturated, then $\tilde{e}$
is major, and thus $Q\restriction_{T_{\tilde{e}}}$ is critical. The
server is on the other side of $\tilde{e}$ than $T_{\tilde{e}}$,
completing the proof.
\item If there is no such $\tilde{e}$, by the minimality of $Q$ we have
for any $\tilde{e}$ sibling of $\hat{e}$ such that $w(\tilde{e})\ge w(e)$
that $Q\restriction_{T_{\tilde{e}}}=\emptyset$.
\begin{enumerate}
\item If $Q\restriction_{T_{\hat{e}}}$ does not saturate $T_{\hat{e}}$,
then again from minimality of $Q$ we have $Q\restriction_{T_{\hat{e}}}=\emptyset$.
Thus $Q\subseteq R_{\hat{e}}$. Since $w(R_{\hat{e}}^{Q})=w(T_{e}^{Q})-w(e)+w(\hat{e})\le w(T_{e}^{Q})$,
it holds that $Q$ saturates $\hat{e}$, and is thus critical. 
\item Otherwise, $Q\restriction_{T_{\hat{e}}}$ saturates $T_{\hat{e}}$,
and is thus critical. Since the server is internal to $T_{\hat{e}}$,
induction on the height of $e$ yields the proof.
\end{enumerate}
\end{enumerate}
\end{enumerate}
The case that $Y_{e}=R_{e}$ and the server is not internal to $T_{e}$
is very similar. 
\end{proof}
\textbf{Algorithm's description. }The algorithm for service with delay
on a $\left(\ge2\right)$-HST is given in Algorithm \ref{alg:OSD_Algorithm}.
The algorithm triggers a service whenever a set of requests becomes
critical. We assume that the set of requests considered by the algorithm
is always on the other side of the major edge than the server. This
assumption uses Proposition \ref{prop:OSD_CriticalTreeOppositeServer}.

Whenever a set of requests becomes critical, saturating $Y_{e}$ for
a major edge $e$, the algorithm moves the server to the closer node
touching $e$ (denoted by $u_{1}$). It then calls the exploration
function of the multilevel aggregation algorithm for $\left(\ge2\right)$-HSTs,
given in Algorithm \ref{alg:MAD_Algorithm}. To make this well defined,
a call to $\texttt{MultilevelAggregationExplore}(Y_{e})$ observes
the $\left(\ge2\right)$-HST $Y_{e}$, in which $e$ is the root edge.
If $Y_{e}=R_{e}$, then $e$ is ``promoted'' to be the parent edge
of its siblings in $R_{e}$ for the sake of the multilevel aggregation
exploration (note that the resulting tree is indeed a $\left(\ge2\right)$-HST).
The counters used by the exploration are the same counters $c_{e}$
of the service with delay algorithm.

The exploration of the multilevel aggregation algorithm yields a tree
to transmit $\T$. In the case of service with delay, instead of transmitting
$\T$, we traverse it with the server, in DFS order, returning to
the node $u_{1}$. Note that the cost of this is exactly twice the
weight of $\T$. To conclude, the server crosses $e$, ending the
service on the other side of $e$ than before the service. Observe
that while this concludes the call to $\UponCritical$, it may immediately
trigger new calls to $\UponCritical$ due to new edges becoming major
in the server's new location.

\noindent \LinesNumbered \RestyleAlgo{boxruled}\renewcommand{\algorithmcfname}{Algorithm}
\begin{algorithm}[tb]
\caption{\label{alg:OSD_Algorithm}Online Service with Delay}

\SetKwProg{Fn}{Function}{}{end}
\SetKwProg{EFn}{Event Function}{}{end}
\SetKwFunction{UponCritical}{UponCritical}
\SetKwFunction{Explore}{Explore}
\SetKwFunction{MultilevelAggregationExplore}{MultilevelAggregationExplore}

\textbf{Initialization.}

Initialize $c_{e}\leftarrow0$ for any edge $e\in T\backslash\{r\}$

\;

\EFn(\tcp*[h]{Upon request set becoming critical as per Definition
\ref{def:OSD_CriticalSet}}){\UponCritical{}}{

Let $Y_{e}$ be the tree saturated by the critical requests, such
that $e$ is a major edge.

\tcp*[h]{The server is on the other side of $e$ than $Y_{e}$, by
Proposition \ref{prop:OSD_CriticalTreeOppositeServer}.}

let $e=(u_{1},u_{2})$, such that the server is on $u_{1}$'s side.

move server to $u_{1}$. 

let $\T\leftarrow$\MultilevelAggregationExplore{$Y_{e}$}

traverse $\T$ in DFS order, finishing at $u_{1}$.

traverse $e$ to reach $u_{2}$.

}
\end{algorithm}

\subsection{Analysis}

Fix any instance of online service with delay on the tree $T$. Define
$\alg^{B}$ and $\alg^{D}$ to be the total moving cost and the total
delay cost of the algorithm on the instance, respectively. Define
$\alg=\alg^{B}+\alg^{D}$. Define $\opt^{B},\opt^{D}$ and $\opt$
similarly for the optimum. 

In this subsection, we prove the following theorem.
\begin{thm}
\label{thm:OSD_HSTTheorem}$\alg\le O(D)\cdot\opt^{B}+O(D^{2})\cdot\opt^{D}$.
\end{thm}

Observe that upon embedding from a general metric space of $n$ points
to a $\left(\ge2\right)$-HST, the moving cost is distorted but the
delay cost is not. Thus, using similar arguments to the proof of Theorem
\ref{thm:FLDeadline_GMSLogSquared}, we have that Theorem \ref{thm:OSD_HSTTheorem}
implies Theorem \ref{thm:OSD_GMSLogSquared} for general metric spaces.

\subsubsection{Upper Bounding $\protect\alg$}

We again denote by $k$ the number of services made by the algorithm.
That is, $k$ is the number of calls to $\UponCritical$. We denote
by $e_{i}$ for $i\in[k]$ the major edge considered in the $i$'th
service. We also denote by $t_{i}$ the time of the $i$'th service.

We devote this part of the analysis to proving the following lemma.
\begin{lem}
\label{lem:OSD_ALG}$\alg\le O(D)\cdot\sum_{i=1}^{k}w(e_{i})$
\end{lem}

Observe the operation of the algorithm. Upon a critical set of requests
the algorithm calls $\UponCritical$ a few times consecutively, until
there is no critical set of requests with regard to the server's current
location. The algorithm then enters the waiting state. We call each
such instantaneous set of services a \emph{service phase. }We denote
by $k^{\prime}$ the number of these phases. We also assume that no
two sets of requests become critical at the same time, which can easily
be enforced by the algorithm by breaking ties arbitrarily.
\begin{prop}
\label{prop:OSD_ServerInPhaseInitialTree}Consider the service phase
which starts from a set of requests $Q$ becoming critical by saturating
$Y_{e}$, for a major edge $e$. Then during the entire phase, the
server only serves requests internal to $Y_{e}$.
\end{prop}

\begin{proof}
The first service in the phase only serves requests internal to $Y_{e}$,
and the server finishes the service in a point internal to $Y_{e}$.
We claim that during the rest of the phase, the server remains internal
to $Y_{e}$, which proves the proposition.

Assume otherwise. Then we must have that at some point during the
phase, a set of pending requests $Q^{\prime}$ is critical (with regards
to the server's location at that point in the phase) by saturating
$Y_{e^{\prime}}$ for an edge $e^{\prime}\notin Y_{e}$. Consider
the first such point during the phase. Due to our assumption that
no two sets of requests become critical at the same time, we have
that $e^{\prime}$ must not have been a major edge before the start
of the phase. But note that all edges in $Y_{e}$ have weight at most
$w(e)$, and thus the server only traversed edges of weight at most
$w(e)$ since the start of the phase. Thus, we must have that $w(e^{\prime})<w(e)$.
Now, note that the server cannot reach any edge $e^{\prime}$ such
that $w(e^{\prime})<w(e)$ and $e^{\prime}\notin Y_{e}$ from a position
which is internal to $Y_{e}$ without traversing an edge of weight
at least $w(e)$. This is a contradiction to $e^{\prime}$ being
a major edge.
\end{proof}
\begin{lem}
\label{lem:OSD_DelayBoundedByBuying}$\alg^{D}\le\alg^{B}$
\end{lem}

\begin{proof}
Let $\mathcal{Q}$ be the set of all requests in the instance. Divide
$\mathcal{Q}$ into $Q_{1},...,Q_{k^{\prime}}$ such that $Q_{i}$
are the requests served by the algorithm in the $i$'th phase. 

Fix the $i$'th phase, let $t$ be the time of the phase and let $Q=Q_{i}$.
Let $Y_{e}$ be the saturated tree triggering the phase, with $e$
a major edge. Due to Proposition \ref{prop:OSD_ServerInPhaseInitialTree},
we have that $Q\subseteq Y_{e}$. Since the algorithm's server is
outside $Y_{e}$, we have that $w(Y_{e}^{Q})$ is a lower bound for
the cost of moving the server to serve $Q$. Since $e$ is a major
edge immediately before the start of the phase, we have that $d_{Q}(t)\le w(Y_{e}^{Q})$.
Thus the delay incurred by the requests of $Q$ is bounded by the
buying cost incurred by the algorithm in the phase.

Summing this conclusion over all phases yields the lemma.
\end{proof}
\begin{prop}
\label{prop:OSD_MajorEdgeNotTooFar}Moving the server to touch a major
edge $e$ costs at most $2w(e)$.
\end{prop}

\begin{proof}
Since we are in a $\left(\ge2\right)$-HST, the path from any node
to another node consists of (at most) one upwards path followed by
one downwards path. Since $e$ is a major edge, each edge on the path
from the server to $e$ must have weight at most $w(e)$. Thus, the
downwards path must be of length $0$ -- otherwise, it would contain
$e$'s parent edge, which has weight larger than $w(e)$. Consider
that the weight of the upwards path is at most $2w(e)$.
\end{proof}
\begin{lem}
\label{lem:OSD_BoundedBuying}$\alg^{B}\le(2D+5)\cdot\sum_{i=1}^{k}w(e_{i})$
\end{lem}

\begin{proof}
Each service triggered by the saturation of a major edge $e$ causes
a multilevel aggregation service of either $T_{e}$ or $R_{e}$, plus
additional server movements required to reach and traverse $e$. The
additional movements are of at most $3w(e)$ (using Proposition \ref{prop:OSD_MajorEdgeNotTooFar}),
and thus $3\cdot\sum_{i=1}^{k}w(e_{i})$ over all services.

Using a very similar proof to the case for multilevel aggregation,
we can show that the sum of the weight of the trees to ``transmit''
yielded by the calls to the multilevel aggregation algorithm are at
most $(D+1)\sum_{i=1}^{k}w(e_{i})$. Since traversing a tree by DFS
is twice the cost of transmission, the buying cost incurred by the
OSD algorithm for that step is at most $2(D+1)\sum_{i=1}^{k}w(e_{i})$.

Overall, the buying cost of the algorithm is at most $(2D+5)\sum_{i=1}^{k}w(e_{i})$.
\end{proof}
\begin{proof}[of Lemma \ref{lem:OSD_ALG}]
 The lemma results directly from Lemmas \ref{lem:OSD_DelayBoundedByBuying}
and \ref{lem:OSD_BoundedBuying}.
\end{proof}

\subsubsection{Lower Bounding $\protect\opt$}
\begin{defn}[$\mathbb{I}_{i}$]
We define the indicator variable $\mathbb{I}_{i}$ for $i\in[k]$
to be $1$ if the optimum's server was on the same side of $e_{i}$
at $t_{i}$ as the algorithm's server (before the call to $\UponCritical$),
and $0$ otherwise.
\end{defn}

The following lemma provides a lower bound on the cost of the optimum.
\begin{lem}
\label{lem:OSD_OPT}$\sum_{i=1}^{k}\mathbb{I}_{i}\cdot w(e_{i})\le3\cdot\opt^{B}+3D\cdot\opt^{D}$
\end{lem}

\subparagraph*{Charging nodes and incurred costs.}

We first define the charging nodes for the analysis of this algorithm.
For every edge $e$, there exist three types of charging nodes:
\begin{enumerate}
\item \emph{Standard root charging nodes} (SRCN), which are nodes of the
form $(e,[\tau_{1},\tau_{2}))$ where $\tau_{1}$ and $\tau_{2}$
are two subsequent times in which $\Explore(e)$ is called due to
$e$ being a major edge and $T_{e}$ being saturated, triggering service.
\item \emph{Relative root charging nodes} (RRCN), which are nodes of the
form $(e,[\tau_{1},\tau_{2}))$ where $\tau_{1}$ and $\tau_{2}$
are two subsequent times in which $\Explore(e)$ is called due to
$e$ being a major edge and $R_{e}$ being saturated, triggering service.
\item \emph{Normal charging nodes} (NCN), which are nodes of the form $(e,[\tau_{1},\tau_{2}))$
where $\tau_{1}$ and $\tau_{2}$ are two subsequent times in which
$\Explore(e)$ is called due to the counter $c_{e}$\emph{ }reaching
$w(e)$.
\end{enumerate}
Nodes of types 1 and 2 correspond to root charging nodes in the multilevel
aggregation case, while nodes of type 3 correspond to non-root nodes. 

For a charging node $\mu=(e,[\tau_{1},\tau_{2}))$ we say that:
\begin{itemize}
\item $\opt$ incurs a \emph{buying cost }of $w(e)$ in $\mu$ if $\opt$
traversed the edge $e$ during $[\tau_{1},\tau_{2})$. We denote the
buying cost that $\opt$ incurs in $\mu$ by $c_{b}(\mu)$.
\item If $\mu$ is an SRCN or an NCN, $\opt$ incurs a \emph{delay cost}
in $\mu$ equal to the delay incurred by $\opt$ on the set of requests
$Q=\{q\in T_{e}|r_{q}\in[\tau_{1},\tau_{2})\}$
\item If $\mu$ is an RRCN, $\opt$ incurs a\emph{ delay cost} in $\mu$
equal to the delay incurred by $\opt$ on the set of requests $Q=\{q\in R_{e}|r_{q}\in[\tau_{1},\tau_{2})\}$
\uline{if} $\opt$'s server remained internal to $T_{e}$ during
$[\tau_{1},\tau_{2})$.
\end{itemize}
We denote the total delay cost incurred by $\opt$ in $\mu$ be $c_{d}(\mu)$.
We denote the total cost that $\opt$ incurs in $\mu$ by $c(\mu)=c_{b}(\mu)+c_{d}(\mu)$.
\begin{lem}
\label{lem:OSD_ChargeLBOpt}$\sum_{\mu\in M}c(\mu)\ge3\cdot\opt^{B}+3D\cdot\opt^{D}$
\end{lem}

\begin{proof}
Observe that any edge traversal by the optimum's server can be counted
in three charging nodes relating to that edge (one SRCN, one RRCN
and one NCN). 

Any delay cost incurred by the optimum due to a request $q$ can be
counted in NCNs and SRCNs along the depth of the tree, yielding $2D$
such charging nodes. In addition, the delay of $q$ can be counted
in at most $D$ RRCNs along the path from the root to the location
of the optimum's server at time $r_{q}$. 

These observations yield the lemma.
\end{proof}
Denote by $M$ the set of all charging nodes. To prove Lemma \ref{lem:MAD_OPT},
we show a preflow on the set of vertices $M\cup\{s\}$, where $s$
is the source node. 

The following definition of charging node investment is nearly identical
to the definition in the multilevel aggregation case.
\begin{defn}[Investing]
For a charging node $\mu_{1}=(e_{1},[\tau_{1}^{1},\tau_{2}^{1}))$
and an NCN $\mu_{2}=(e_{2},[\tau_{1}^{2},\tau_{2}^{2}))$, we say
that $\mu_{1}$ \emph{invested $x$ in $\mu_{2}$ }if $\Explore_{\tau_{1}^{1}}(e_{1})$
raised the counter $c_{e_{2}}$ by $x$ during the counter phase $[\tau_{1}^{2},\tau_{2}^{2})$
(not including recursive calls made by $\Explore_{\tau_{1}^{1}}(e_{1})$).
\end{defn}

Procedure \ref{proc:OSD_PreflowBuilder} is used to build the preflow.
As in the previous analyses, we define $\bar{E}$ to be the set of
possible edges between nodes of $M$ to themselves. As before, an
edge $\sigma$ exists in $\bar{E}$ from a charging node $\mu$ to
a charging node $\mu^{\prime}$ if $\mu$ invested in $\mu^{\prime}$,
and $\alpha(\sigma)$ is set to be the total invested amount.

\LinesNumbered \RestyleAlgo{boxruled}\renewcommand{\algorithmcfname}{Procedure}\DontPrintSemicolon
\begin{algorithm}[tb]
\caption{\label{proc:OSD_PreflowBuilder}PreflowBuilder - Online Service with
Delay}

\SetKwProg{Fn}{Function}{}{end}
\SetKwFunction{PreflowBuilder}{PreflowBuilder}
\SetKwFunction{SetColor}{SetColor}
\SetKw{Break}{break}

\textbf{Initialization.}

Let the set of vertices of $G$ be $M\cup\{s\}$, and initialize the
edge set to be $E=\emptyset$.

Initialize dictionary $\texttt{Color}[w]=\None$ for every $\mu\in M$.

\ForEach{$\mu=(e,[\tau_{1},\tau_{2}))\in M$ such that $\opt$ traversed
edge $e$ during $[\tau_{1},\tau_{2})$}{

set $\Color[\mu]\leftarrow\Special$

}

\ForEach{$\mu\in M$ such that $c(\mu)>0$}{

add a new edge $\sigma=(s,\mu)$ to $E$, and set $\alpha(\sigma)=c(\mu)$

}

\;

\Fn{\PreflowBuilder{}}{

\For{$i$ from $1$ to $k$}{

let $\mu\leftarrow(e_{i},[t_{i-1},t_{i}))$ be the RCN of the $i$'th
service.

\lIf{$\mathbb{I}_{i}=1$}{\SetColor{$\mu$,$\mu$}}

}

\For{$j$ from $1$ to $D$}{

\ForEach{NCN $\mu=(e,[\tau_{1},\tau_{2}))\in M$ such that $e$ is
of depth $j$}{

\ForEach{edge $\sigma\in E_{\mu}^{-}$ incoming to a node $\mu^{\prime}$}{

\lIf{\SetColor{$\mu,\texttt{Color}[\mu^{\prime}]$}$\neq\None$}{\Break}

}

}

}

}

\;

\lFn(\tcp*[h]{As in Procedure \ref{proc:FLDeadline_PreflowBuilder}}){\SetColor{$\mu$,$\mu^{\star}$}}{}
\end{algorithm}

We use the following definition for ease.
\begin{defn}[$Y_{\mu}$]
 For a charging node $\mu=(e,[\tau_{1},\tau_{2}))$, we define $Y_{\mu}$
to be $R_{e}$ if $\mu$ is a RRCN. Otherwise, we define $Y_{\mu}$
to be $T_{e}$.
\end{defn}

\begin{obs}
\label{obs:OSD_TreeContainment}If a node $\mu=(e,[\tau_{1},\tau_{2}))$
invested in a node $\mu^{\prime}=(e^{\prime},[\tau_{1},\tau_{2}))$,
then $Y_{\mu^{\prime}}\subseteq Y_{\mu}$.
\end{obs}

\begin{prop}[analogue of Proposition \ref{prop:MAD_NotBought}]
\label{prop:OSD_NotBought}Let $\mu=(e,[\tau_{1},\tau_{2}))$ such
that $\texttt{Color}[\mu]=\mu^{\star}$ for some charging node $\mu^{\star}=(e^{\star},[\tau_{1}^{\star},\tau_{2}^{\star}))$.
Then $\opt$ did not enter $Y_{\mu}$ during $[\tau_{1},\tau_{2}^{\star})$.
\end{prop}

\begin{proof}
Since $\texttt{Color}[\mu]=\mu^{\star}$, we must have that for the
RCN $\mu^{\star}$ we have that $\Color[\mu^{\star}]=\mu^{\star}$.
Thus, we have that $\mathbb{I}_{i}=1$ for $i$ such that $\tau_{2}^{\star}=t_{i}$,
and thus the optimum's server was on the same side of $e^{\star}$
as the algorithm's server before the service at $\tau_{2}^{\star}$.
Since we only consider critical trees on the other side of the major
edge, we have that the optimum's server was not internal to $Y_{\mu^{\star}}$
at time $\tau_{2}^{\star}$. Since $\Color[\mu^{\star}]\neq\Special$,
the optimum's server did not traverse $e^{\star}$ during $[\tau_{1}^{\star},\tau_{2}^{\star})$,
and thus was not internal to $Y_{\mu^{\star}}$ during $[\tau_{1}^{\star},\tau_{2}^{\star})$.

What follows is a similar inductive argument to that of Proposition
\ref{prop:FLDeadline_NotBought}. For the base case that $\mu=\mu^{\star}$,
we are done. We now prove the proposition by induction on the depth
of the propagation of the color $\mu^{\star}$ to $\mu$. Observe
that the color $\mu^{\star}$ was propagated to $\mu$ from another
charging node $\mu^{\prime}=(e^{\prime},[\tau_{1}^{\prime},\tau_{2}^{\prime}))$.
By induction, the optimum's server was not internal to $Y_{\mu^{\prime}}$
during $[\tau_{1}^{\prime},\tau_{2}^{\star})$. From Observation \ref{obs:OSD_TreeContainment},
we have that the optimum's server was not internal to $Y_{\mu}$ during
$[\tau_{1}^{\prime},\tau_{2}^{\star})$.

Since $\Color[\mu]\ne\Special$, the optimum's server did not traverse
$e$ during $[\tau_{1},\tau_{2})$. Since $\mu^{\prime}$ invested
in $\mu$, we have that $\tau_{1}^{\prime}\le\tau_{2}$, and thus
the optimum's server was not internal to $Y_{\mu}$ during $[\tau_{1},\tau_{2}^{\star})$
as required.
\end{proof}
\begin{obs}
\label{obs:OSD_w_e_is_enough}Corollary \ref{cor:MAD_w_e_is_enough}
from the multilevel aggregation case holds in this case as well. That
is, if $\sum_{\sigma\in E_{\mu}^{+}}\alpha(\sigma)\ge w(e)$, then
$\chi_{\mu}\ge0$.
\end{obs}

\begin{lem}
\label{lem:OSD_ValidPreflow}The preflow defined by Procedure \ref{proc:OSD_PreflowBuilder}
is valid.
\end{lem}

\begin{proof}
As in previous versions of this lemma, we need to show that $\chi_{\mu}\ge0$
for every $\mu\in M$. We separate according to cases.

\textbf{Case 1:} $\Color[\mu]=\Special$. In this case, $\opt$ incurs
a buying cost of $w(e)$ at $\mu$, completing the case according
to Observation \ref{obs:OSD_w_e_is_enough}.

\textbf{Case 2: }$\Color[\mu]=\mu^{\star}$ for some charging node
$\mu^{\star}$. In this case, observe that Observation \ref{obs:MAD_ColoredNotMinusInfinity}
applies for OSD as well. Thus, $\mu$ has invested in other nodes
a total of exactly $w(e)$, and thus $\sum_{\sigma\in E_{\mu}^{+}}\alpha(\sigma)\ge w(e)$.
Observation \ref{obs:OSD_w_e_is_enough} completes the proof for this
case.

\textbf{Case 3: }$\Color[\mu]=\None$. If there are no outgoing edges
from $\mu$, then clearly $\chi_{\mu}\ge0$ and we are done. Otherwise,
$\mu$ is an NCN, and there exists an outgoing edge $\sigma$ to some
node $\mu^{\prime}=(e^{\prime},[\tau_{1}^{\prime},\tau_{2}^{\prime}))$
with $\Color[\mu^{\prime}]=\mu^{\star}$ for some charging node $\mu^{\star}=(e^{\star},[\tau_{1}^{\star},\tau_{2}^{\star}))$.
Observe that since $\mu^{\prime}$ invested in $\mu$, we must have
that $\tau_{1}^{\prime}\le\tau_{2}$. Using Proposition \ref{prop:OSD_NotBought},
and the fact that $\Bought[\mu]=False$, we have that $\opt$ was
not internal to $Y_{\mu}$ during $[\tau_{1},\tau_{2}^{\star})$.
As in Case 3 of Lemma \ref{lem:MAD_ValidPreflow}, we locate a set
of requests internal to $Y_{\mu}$ due to which $\opt$ incurs delay
cost of $w(e)$ in $\mu$.

\textbf{Claim -- }There exists a set of requests $Q^{\prime}\subseteq Y_{\mu}$
such that $r_{q}\in[\tau_{1},\tau_{2})$ such that $d_{Q^{\prime}}(\tau_{2}^{\star})\ge w(e)$.
\begin{proof}[Proof of claim]
 Identical to the proof for the corresponding claim in the multilevel
aggregation analysis.
\end{proof}
Using the claim, observe that since the optimum's server was not internal
to $Y_{\mu}$ during $[\tau_{1},\tau_{2}^{\star})$, it has incurred
$w(e)$ delay due to the requests of $Q^{\prime}$. Due to the definition
of delay cost on an NCN, we have that $c_{d}(\mu)\ge w(e)$. This
completes the analysis of the case due to Observation \ref{obs:OSD_w_e_is_enough}.
\end{proof}
\begin{lem}
\label{lem:OSD_RootExcesses}For every root charging node $\mu=(e_{i},[t_{i-1},t_{i}))$
we have that $\chi_{\mu}\ge\mathbb{I}_{i}\cdot w(e_{i})$.
\end{lem}

\begin{proof}
If $\Color[\mu]\neq\None$, we have that $\chi_{\mu}\ge w(e_{i})$
using identical arguments to Cases 1 and 2 of Lemma \ref{lem:OSD_ValidPreflow}.

Otherwise, $\Color[\mu]=None.$ Observe that $\chi_{\mu}\ge0$, due
to Lemma \ref{lem:OSD_ValidPreflow}, which covers the case that $\mathbb{I}_{i}=0$.
Now, suppose that $\mathbb{I}_{i}=1$. We show that $\opt$ incurred
a delay cost of at least $w(e_{i})$ in $\mu$.

\textbf{Claim -- }There exists a set of requests $Q^{\prime}\subseteq Y_{\mu}$
such that $r_{q}\in[t_{i-1},t_{i})$ such that $d_{Q^{\prime}}(t_{i})\ge w(e)$.
\begin{proof}[Proof of Claim]
We denote by $Q$ the set of requests that became critical at $t_{i}$,
triggering the service. Observe that $d_{Q}(t_{i})\ge w(Y_{\mu})$,
and that $r_{q}<t_{i}$ for every $q\in Q$. Since $\Color[\mu]=\None$,
we must have that either $t_{i-1}=-\infty$ or $\lambda_{\mu}>t_{i}$. 

If $t_{i-1}=-\infty$, then $r_{q}\in[t_{i-1},t_{i})$ and choosing
$Q^{\prime}=Q$ yields the claim. Otherwise, $t_{i-1}\ne-\infty$,
and $\lambda_{\mu}>t_{i}$. In this case, we choose $\hat{Q}\subseteq Q$
to be the set of pending requests immediately after the service at
$t_{i-1}$. Since $\lambda_{\mu}>t_{i}$, $d_{\hat{Q}}(t_{i})\le w(Y_{\mu}^{\hat{Q}})-w(e_{i})\le w(Y_{\mu}^{Q})-w(e_{i})$.
Thus, we have that $d_{Q\backslash\hat{Q}}(t_{i})\ge w(e_{i})$. Observe
that $r_{q}\ge t_{i-1}$ for every $q\in Q\backslash\hat{Q}$, and
thus $r_{q}\in[t_{i-1},t_{i})$ for every $q\in Q\backslash\hat{Q}$.
Thus choosing $Q^{\prime}=Q\backslash\hat{Q}$ yields the claim.
\end{proof}
We now use this claim. Observe that the optimum's server was not internal
to $Y_{\mu}$ at $t_{i}$ (due to $\mathbb{I}_{i}=1$), and since
$\Color[\mu]\neq\Special$, the optimum's server was not internal
to $Y_{\mu}$ during $[t_{i-1},t_{i})$. Thus, the optimum incurs
a delay cost of $w(e_{i})$ due to $Q^{\prime}$. Now observe that:
\begin{itemize}
\item If $\mu$ is an SRCN, then $c_{d}(\mu)\ge w(e_{i})$.
\item If $\mu$ is an RRCN, then the algorithm's server was internal to
$T_{e_{i}}$ at time $t_{i}$. Since $\mathbb{I}_{i}=1$, the optimum's
server was internal to $T_{e_{i}}$ as well at $t_{i}$. Since $\Color[\mu]\neq\Special$,
the optimum's server stayed internal to $T_{e_{i}}$ during $[t_{i-1},t_{i})$.
Thus, $c_{d}(\mu)\ge w(e_{i})$.
\end{itemize}
In both cases, $c_{d}(\mu)\ge w(e_{i})$, completing the proof of
the case and lemma.
\end{proof}
\begin{proof}[of Lemma \ref{lem:OSD_OPT}]
 The proof of the lemma results from observing the subset $N\subseteq M$
of all root charging nodes. Lemma \ref{lem:OSD_RootExcesses} implies
that $\sum_{\mu\in N}\chi_{\mu}\ge\sum_{i=1}^{k}\mathbb{I}_{i}\cdot w(e_{i})$. 

We now use Proposition \ref{prop:Flow_SubsetLBSource} and Lemma \ref{lem:OSD_ChargeLBOpt}
to obtain 
\[
\sum_{i=1}^{k}\mathbb{I}_{i}\cdot w(e_{i})\le\omega_{Z}=\sum_{\mu\in M}c(\mu)\le3\cdot\opt^{B}+3D\cdot\opt^{D}
\]
 proving the lemma.
\end{proof}

\subsubsection{Proof of Main Theorem}

In this part of the analysis, we prove Theorem \ref{thm:OSD_HSTTheorem}. 

From Lemma \ref{lem:OSD_ALG}, we have that $\alg\le\gamma D\cdot\sum_{i=1}^{k}w(e_{i})$
for some constant $\gamma$.
\begin{defn}[Potential function $\phi(t)$]
We define the potential function $\phi(t)$ to be $\gamma D$ times
the distance between the algorithm's server and the optimum's server
at time $t$. 

Observe that $\phi(-\infty)=0$.
\end{defn}

For every $i\in[k]$, define the difference in potential $\Delta_{i}\phi=\phi(t_{i}^{+})-\phi(t_{i}^{-})$,
where $t_{i}^{-}$ is time $t_{i}$ immediately before the $i$'th
service and $t_{i}^{+}$ is time $t_{i}$ immediately after the $i$'th
service.

We define $\alg_{i}=\gamma Dw(e_{i})$, and $\opt_{i}=\mathbb{I}_{i}\cdot w(e_{i})$.
Observe from Lemmas \ref{lem:OSD_ALG} and \ref{lem:OSD_OPT} that
$\sum_{i}\alg_{i}\ge\alg$ and $\sum_{i}\opt_{i}\le3\cdot\opt^{B}+3D\cdot\opt^{D}$.
\begin{lem}
\label{lem:OSD_PotentialStep}For every $i\in[k]$, we have that $\alg_{i}\le4\gamma D\cdot\opt_{i}-\Delta_{i}\phi$.
\end{lem}

\begin{proof}
If $\mathbb{I}_{i}=1$, then $\opt_{i}=w(e_{i})$. Using Proposition
\ref{prop:OSD_MajorEdgeNotTooFar}, we have that $\Delta_{i}\phi\le3\gamma D\cdot w(e_{i})$.
Thus 
\[
\alg_{i}=\gamma Dw(e_{i})=4\gamma D\cdot\opt_{i}-3\gamma D\cdot w(e_{i})\le4\gamma D\cdot\opt_{i}-\Delta_{i}\phi
\]
as required.

Otherwise, $\mathbb{I}_{i}=0$. Then, $\opt_{i}=0$. Since the optimum's
server is on the other side of the edge $e_{i}$ than the algorithm's
server before the $i$'th service, and the algorithm finishes the
service on that other side of $e_{i}$, it must be that $\Delta_{i}\phi\le-\gamma D\cdot w(e_{i})$.
Thus,
\[
\alg_{i}=\gamma Dw(e_{i})\le4\gamma D\cdot\opt_{i}-\Delta_{i}\phi
\]
finishing the proof of the lemma.
\end{proof}
\begin{prop}
Denote the final value of $\phi$ by $\phi(\infty)$. Then 
\[
\sum_{i}\Delta_{i}\phi\ge\phi(\infty)-\gamma D\cdot\opt^{B}
\]
\end{prop}

\begin{proof}
Consider that $\phi(\infty)=\phi(\infty)-\phi(-\infty)$ can be constructed
by summing the changes to the potential function caused by moves of
the algorithm's server (which are the $\Delta_{i}\phi$) and changes
caused by moves of the optimum's server. Note that moving the optimum's
server by $x$ can increase $\phi$ by at most $\gamma Dx$. Thus,
\[
\phi(\infty)\le\sum_{i}\Delta_{i}\phi+\gamma D\cdot\opt^{B}
\]
yielding the proposition.
\end{proof}
\begin{cor}
\label{cor:OSD_PotentialBound}$\sum_{i}\Delta_{i}\phi\ge-\gamma D\cdot\opt^{B}$
\end{cor}

\begin{proof}[of Theorem \ref{thm:OSD_HSTTheorem}]
 Due to Lemma \ref{lem:OSD_PotentialStep}, we have that 
\[
\alg\le\sum_{i}\alg_{i}\le\sum_{i}4\gamma D\cdot\opt_{i}-\sum_{i}\Delta_{i}\phi
\]
Since $\sum_{i}\opt_{i}\le3\opt^{B}+3D\cdot\opt^{D}$, and using Corollary
\ref{cor:OSD_PotentialBound}, we have that 
\begin{align*}
\alg & \le4\gamma D\cdot\left(3\opt^{B}+3D\cdot\opt^{D}\right)+\gamma D\cdot\opt^{B}\\
 & \le13\gamma D\cdot\opt^{B}+12\gamma D^{2}\cdot\opt^{D}
\end{align*}
proving the theorem.
\end{proof}
\bibliographystyle{plain}
\bibliography{bibfile}

\appendix

\section{\label{appendix:AdditionalFigures}Additional Figures}

\begin{figure}[tb]
\subfloat[After $\protect\Open(u)$.]{\includegraphics[width=0.4\columnwidth]{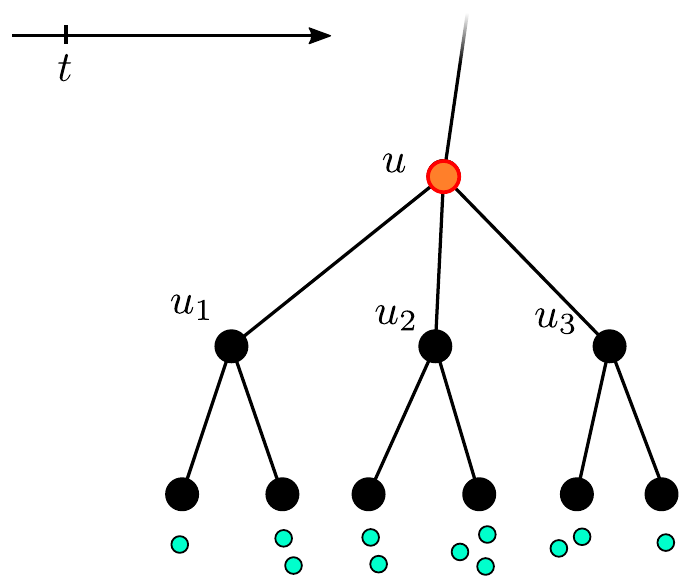}

}\hfill{}\subfloat[Time forwarding reaches $t^{\prime}$, the earliest deadline of a
pending request in $T_{u}$.]{\includegraphics[width=0.4\columnwidth]{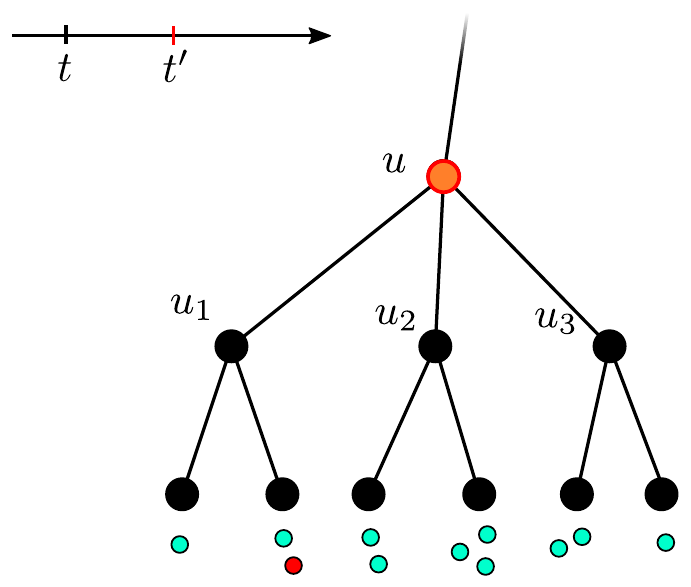}}

\subfloat[The request is connected to $u$.]{\includegraphics[width=0.4\columnwidth]{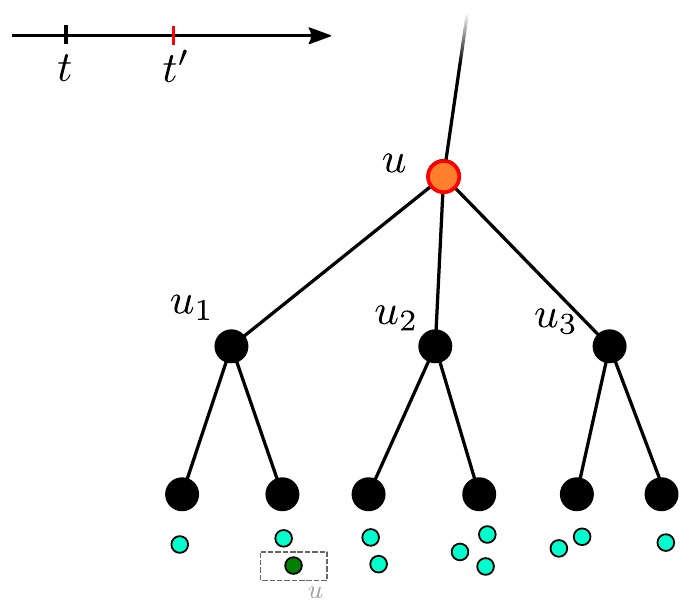}

}\hfill{}\subfloat[The new earliest deadline $t^{\prime\prime}$ is reached in time forwarding.]{\includegraphics[width=0.4\columnwidth]{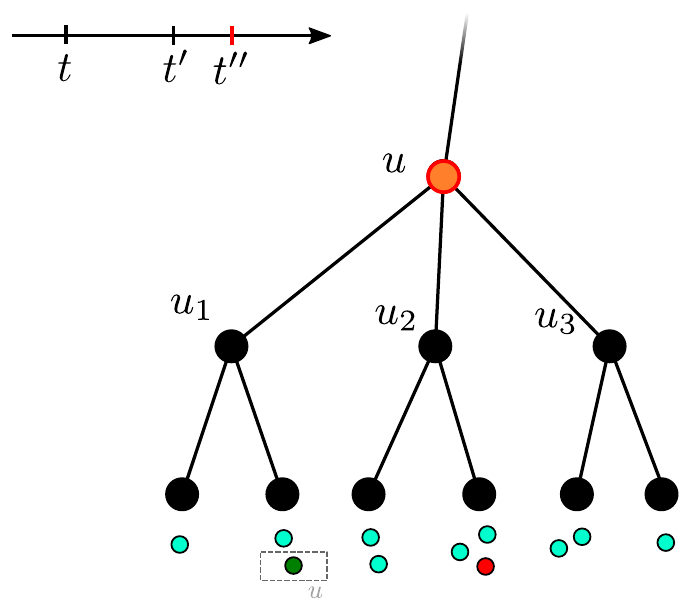}}
\begin{centering}
\subfloat[Investment triggers $\protect\Explore(u_{2})$, serving the earliest
deadline request.]{\includegraphics[width=0.4\columnwidth]{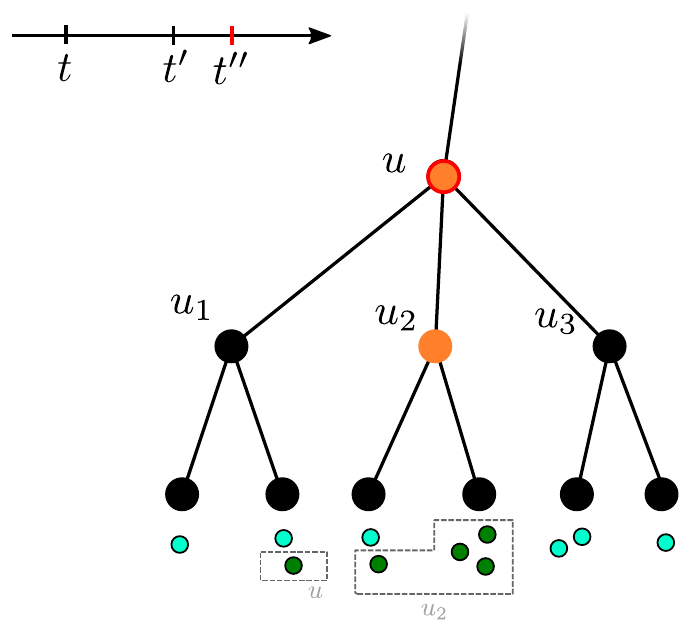}}
\par\end{centering}
\caption{\label{fig:FLDeadline_AlgorithmVisualization}Visualization of Algorithm
\ref{alg:FLDeadline} -- the operation of $\protect\Explore(u)$
at time $t$.}
\end{figure}

\end{document}